\documentclass{LMCS}

\def\dOi{10(1:17)2014}
\lmcsheading%
{\dOi}
{1--52}
{}
{}
{Aug.~27, 2012}
{Mar.~24, 2014}
{}

\ACMCCS{[{\bf Theory of computation}]: Semantics and
  reasoning---Program semantics---Denotational
  semantics\,/\,Categorical semantics; Logic---Linear logic\,/\,Type theory}

\usepackage{amssymb}
\usepackage{amsmath}
\usepackage{stmaryrd}
\usepackage{graphicx}
\usepackage{xy}
\xyoption{all}
\newcommand{\hide}[1]{}
\usepackage{url,hyperref}
\urldef{\mailsa}\path|{alfred.hofmann, ursula.barth, ingrid.haas, frank.holzwarth,|
\urldef{\mailsb}\path|anna.kramer, leonie.kunz, christine.reiss, nicole.sator,|
\urldef{\mailsc}\path|erika.siebert-cole, peter.strasser, lncs}@springer.com|    
\usepackage[height=.8cm,silent]{diagrams}
\newarrow{Into}C--->

\newdimen\proofrulebreadth \proofrulebreadth=.05em
\newdimen\proofdotseparation \proofdotseparation=1.25ex
\newdimen\proofrulebaseline \proofrulebaseline=2ex
\newcount\proofdotnumber \proofdotnumber=3
\let\then\relax
\def\hfi{\hskip0pt plus.0001fil}
\mathchardef\squigto="3A3B
\newif\ifinsideprooftree\insideprooftreefalse
\newif\ifonleftofproofrule\onleftofproofrulefalse
\newif\ifproofdots\proofdotsfalse
\newif\ifdoubleproof\doubleprooffalse
\let\wereinproofbit\relax
\newdimen\shortenproofleft
\newdimen\shortenproofright
\newdimen\proofbelowshift
\newbox\proofabove
\newbox\proofbelow
\newbox\proofrulename
\def\shiftproofbelow{\let\next\relax\afterassignment\setshiftproofbelow\dimen0 }
\def\shiftproofbelowneg{\def\next{\multiply\dimen0 by-1 }%
\afterassignment\setshiftproofbelow\dimen0 }
\def\setshiftproofbelow{\next\proofbelowshift=\dimen0 }
\def\setproofrulebreadth{\proofrulebreadth}

\def\prooftree{% NESTED ZERO (\ifonleftofproofrule)
\ifnum  \lastpenalty=1
\then   \unpenalty
\else   \onleftofproofrulefalse
\fi
\ifonleftofproofrule
\else   \ifinsideprooftree
        \then   \hskip.5em plus1fil
        \fi
\fi
\bgroup% NESTED ONE (\proofbelow, \proofrulename, \proofabove,
\setbox\proofbelow=\hbox{}\setbox\proofrulename=\hbox{}%
\let\justifies\proofover\let\leadsto\proofoverdots\let\Justifies\proofoverdbl
\let\using\proofusing\let\[\prooftree
\ifinsideprooftree\let\]\endprooftree\fi
\proofdotsfalse\doubleprooffalse
\let\thickness\setproofrulebreadth
\let\shiftright\shiftproofbelow \let\shift\shiftproofbelow
\let\shiftleft\shiftproofbelowneg
\let\ifwasinsideprooftree\ifinsideprooftree
\insideprooftreetrue
\setbox\proofabove=\hbox\bgroup$\displaystyle % NESTED TWO
\let\wereinproofbit\prooftree
\shortenproofleft=0pt \shortenproofright=0pt \proofbelowshift=0pt
\onleftofproofruletrue\penalty1
}

\def\eproofbit{% NESTED TWO
\ifx    \wereinproofbit\prooftree
\then   \ifcase \lastpenalty
        \then   \shortenproofright=0pt  % 0: some other object, no indentation
        \or     \unpenalty\hfil         % 1: empty hypotheses, just glue
        \or     \unpenalty\unskip       % 2: just had a tree, remove glue
        \else   \shortenproofright=0pt  % eh?
        \fi
\fi
\global\dimen0=\shortenproofleft
\global\dimen1=\shortenproofright
\global\dimen2=\proofrulebreadth
\global\dimen3=\proofbelowshift
\global\dimen4=\proofdotseparation
\global\count255=\proofdotnumber
$\egroup  % NESTED ONE
\shortenproofleft=\dimen0
\shortenproofright=\dimen1
\proofrulebreadth=\dimen2
\proofbelowshift=\dimen3
\proofdotseparation=\dimen4
\proofdotnumber=\count255
}

\def\proofover{% NESTED TWO
\eproofbit % NESTED ONE
\setbox\proofbelow=\hbox\bgroup % NESTED TWO
\let\wereinproofbit\proofover
$\displaystyle
}%
\def\proofoverdbl{% NESTED TWO
\eproofbit % NESTED ONE
\doubleprooftrue
\setbox\proofbelow=\hbox\bgroup % NESTED TWO
\let\wereinproofbit\proofoverdbl
$\displaystyle
}%
\def\proofoverdots{% NESTED TWO
\eproofbit % NESTED ONE
\proofdotstrue
\setbox\proofbelow=\hbox\bgroup % NESTED TWO
\let\wereinproofbit\proofoverdots
$\displaystyle
}%
\def\proofusing{% NESTED TWO
\eproofbit % NESTED ONE
\setbox\proofrulename=\hbox\bgroup % NESTED TWO
\let\wereinproofbit\proofusing
\kern0.3em$
}

\def\endprooftree{% NESTED TWO
\eproofbit % NESTED ONE
  \dimen5 =0pt% spread of hypotheses
\dimen0=\wd\proofabove \advance\dimen0-\shortenproofleft
\advance\dimen0-\shortenproofright
\dimen1=.5\dimen0 \advance\dimen1-.5\wd\proofbelow
\dimen4=\dimen1
\advance\dimen1\proofbelowshift \advance\dimen4-\proofbelowshift
\ifdim  \dimen1<0pt
\then   \advance\shortenproofleft\dimen1
        \advance\dimen0-\dimen1
        \dimen1=0pt
        \ifdim  \shortenproofleft<0pt
        \then   \setbox\proofabove=\hbox{%
                        \kern-\shortenproofleft\unhbox\proofabove}%
                \shortenproofleft=0pt
        \fi
\fi
\ifdim  \dimen4<0pt
\then   \advance\shortenproofright\dimen4
        \advance\dimen0-\dimen4
        \dimen4=0pt
\fi
\ifdim  \shortenproofright<\wd\proofrulename
\then   \shortenproofright=\wd\proofrulename
\fi
\dimen2=\shortenproofleft \advance\dimen2 by\dimen1
\dimen3=\shortenproofright\advance\dimen3 by\dimen4
\ifproofdots
\then
        \dimen6=\shortenproofleft \advance\dimen6 .5\dimen0
        \setbox1=\vbox to\proofdotseparation{\vss\hbox{$\cdot$}\vss}%
        \setbox0=\hbox{%
                \advance\dimen6-.5\wd1
                \kern\dimen6
                $\vcenter to\proofdotnumber\proofdotseparation
                        {\leaders\box1\vfill}$%
                \unhbox\proofrulename}%
\else   \dimen6=\fontdimen22\the\textfont2 % height of maths axis
        \dimen7=\dimen6
        \advance\dimen6by.5\proofrulebreadth
        \advance\dimen7by-.5\proofrulebreadth
        \setbox0=\hbox{%
                \kern\shortenproofleft
                \ifdoubleproof
                \then   \hbox to\dimen0{%
                        $\mathsurround0pt\mathord=\mkern-6mu%
                        \cleaders\hbox{$\mkern-2mu=\mkern-2mu$}\hfill
                        \mkern-6mu\mathord=$}%
                \else   \vrule height\dimen6 depth-\dimen7 width\dimen0
                \fi
                \unhbox\proofrulename}%
        \ht0=\dimen6 \dp0=-\dimen7
\fi
\let\doll\relax
\ifwasinsideprooftree
\then   \let\VBOX\vbox
\else   \ifmmode\else$\let\doll=$\fi
        \let\VBOX\vcenter
\fi
\VBOX   {\baselineskip\proofrulebaseline \lineskip.2ex
        \expandafter\lineskiplimit\ifproofdots0ex\else-0.6ex\fi
        \hbox   spread\dimen5   {\hfi\unhbox\proofabove\hfi}%
        \hbox{\box0}%
        \hbox   {\kern\dimen2 \box\proofbelow}}\doll%
\global\dimen2=\dimen2
\global\dimen3=\dimen3
\egroup % NESTED ZERO
\ifonleftofproofrule
\then   \shortenproofleft=\dimen2
\fi
\shortenproofright=\dimen3
\onleftofproofrulefalse
\ifinsideprooftree
\then   \hskip.5em plus 1fil \penalty2
\fi
}

\newcommand{\colim}{\operatorname{colim}}

\newtheorem*{terminology}{Terminology}

\newenvironment{proofof}[1]{\begin{proof}[Proof of {#1}.]}{\end{proof}}

\newenvironment{proofsketch}{\begin{proof}[Proof (sketch)]}{\end{proof}}

\newenvironment{proofnotes}{\begin{proof}[Proof notes]}{\end{proof}}

\newcommand{\vj}[3]{#1 \mathrel{\vdash^{v}} #2 \colon \! #3}
\newcommand{\pj}[3]{#1 \mathrel{\vdash^p} #2 \colon \! #3}
\newcommand{\vjname}{{\vdash^{v}}}
\newcommand{\pjname}{{\vdash^p}}

\newcommand{\veq}[4]{#1 \mathrel{\vdash^{v}} #2\equiv #3 \colon \! #4}
\newcommand{\peq}[4]{#1 \mathrel{\vdash^p} #2\equiv #3 \colon \! #4}

\newcommand{\STA}{\sigma}
\newcommand{\STB}{\tau}
\newcommand{\STC}{\upsilon}

\newcommand{\slet}[3]{#2 \, \codefont{to} \, #1 \ld #3}
\newcommand{\return}[1]{\codefont{return} \, #1}

\newcommand{\val}{\mathrm{Val}}

\newcommand{\nat}{\mathrm{nat}}
\newcommand{\EECstate}{\underline{\mathrm{S}}}
\newcommand{\EECarbstate}{\comptype{\mathsf{S}}}
\newcommand{\EECarbret}{\comptype{\mathsf{R}}}
\newcommand{\EECret}{\comptype{\mathrm{R}}}
\newcommand{\svar}s
\newcommand{\kvar}k

\newcommand{\geffrandom}{\codefont{random}}
\newcommand{\geffreadcell}[1]{\codefont{deref}}
\newcommand{\geffwritecell}[1]{\codefont{assign}}

\newcommand{\saccrandom}{\codefont{random}}
\newcommand{\saccreadcell}[1]{\codefont{read}_{#1}}
\newcommand{\saccwritecell}[1]{\codefont{write}_{#1}}

\newcommand{\ltj}[3]{#1 \vdash^v \, #2 \colon \! #3}

\newcommand{\pto}{\rightharpoonup}

\newcommand{\codefont}[1]{\mathtt{#1}}

\newcommand{\comptype}[1]{\underline{#1}}

\newcommand{\linlambda}{\comptype{\lambda}}

\newcommand{\ld}{\mathpunct{.}}
\newcommand{\bang}[1]{{!} \,#1}
\newcommand{\co}{\colon}

\newcommand{\TVX}{X}

\newcommand{\algX}{{\comptype{X}}}
\newcommand{\algY}{{\comptype{Y}}}

\newcommand{\VconstA}{\alpha}
\newcommand{\VconstB}{\beta}

\newcommand{\CconstA}{\comptype{\alpha}}
\newcommand{\CconstB}{\comptype{\beta}}
\newcommand{\CconstC}{\comptype{\gamma}}

\newcommand{\VA}{\mathsf{A}}
\newcommand{\VB}{\mathsf{B}}
\newcommand{\VC}{\mathsf{C}}
\newcommand{\VD}{\mathsf{D}}

\newcommand{\CA}{\comptype{\mathsf{A}}}
\newcommand{\CB}{\comptype{\mathsf{B}}}
\newcommand{\CC}{\comptype{\mathsf{C}}}

\newcommand{\lpop}{\multimap}
\newcommand{\tensor}{\otimes}
\newcommand{\ltensortype}[2]{{!} #1 \, {\tensor} \, #2}
\newcommand{\lpowertype}[2]{#1\to #2}
\newcommand{\prodtype}{\times}
\newcommand{\Ctimes}{\mathop{\underline{\times}}}
\newcommand{\Cone}{{\underline{1}}}

\newcommand{\lfun}{\multimap}

\newcommand{\algplus}{\oplus}

\newcommand{\algone}{{\underline 1}} %^{\circ}}
\newcommand{\valone}{{1}} %^{\circ}}
\newcommand{\algzero}{\comptype{0}}

\newcommand{\llambda}{\linlambda}
\newcommand{\lam}[3]{\lambda #1.\: #3}

\newcommand{\llam}[3]{\llambda #1.\: #3}

\newcommand{\vimage}[1]{\mathrm{?}(#1)}

\newcommand{\prj}[2]{\pi_{#1}(#2)}
\newcommand{\fst}[1]{\prj 1 {#1}}
\newcommand{\snd}[1]{\prj 2 {#1}}

\newcommand{\cbvimage}[1]{\codefont{?}(#1)}

\newcommand{\ltensorterm}[2]{{!  #1} \tensor #2}
\newcommand{\lpowerterm}[3]{\lambda #1.\, #3}

\newcommand{\algstar}{\star} %{*} %^{\circ}}

\newcommand{\letdot}[4]{{#3}\:\mathrm{to}\:{(\ltensorterm{#1}{#2})}.\;#4}

\newcommand{\lappl}[2]{#1[ #2 ]}

\newcommand{\In}[2]{#1 \colon  \! #2}
\newcommand{\rIn}[2]{#1 \colon  #2}
\newcommand{\aj}[4]{#1 \mid  \! #2 \, \vdash \, \rIn{#3}{#4}}
\newcommand{\aeq}[5]{#1 \mid  \! #2 \, \vdash \, \rIn{#3\equiv #4}{#5}}

\newcommand{\tj}[3]{\aj{#1}{{-}}{#2}{#3}}
\newcommand{\teq}[4]{\aj{#1}{{-}}{#2\equiv #3}{#4}}
\newcommand{\CBV}[1]{\mathrm{FGCBV}\!_{#1}}
\newcommand{\EEC}{\mathrm{EEC}}

\newcommand{\FGCBV}{FGCBV}
\newcommand{\ECBV}{ECBV}
\newcommand{\ECBVS}[1]{\mathrm{ECBV}_{#1}^{\states}}

\newcommand{\CBVtoEEC}[1]{#1^{\circ}}
\newcommand{\CBVtoEECbase}[2]{#2^{#1}}

\renewcommand{\vec}[1]{\bar{#1}}

\DeclareMathOperator{\Prod}{\textstyle{\prod}}

\newcommand{\VCat}{\fixedcatfont{V}} % Model of value types
\newcommand{\CCat}{\fixedcatfont{C}} % Model of computation types

\newcommand{\Kl}[1]{\mathbf{Kl}(#1)}

\newcommand{\CHom}[2]{\CCat(#1, #2)} % \CCat homset, object of \VCat
\newcommand{\CHomp}[2]{\CCat'(#1, #2)} % \CCat homset, object of \VCat

\newcommand{\FFunK}[1]{F^K} %_{#1}} % left adjoint in Kleisli adjunction given by obj #1
\newcommand{\UFunK}[1]{U^K} %_{#1}} % right adjoint in Kleisli adjunction given by obj #1

\newcommand{\ltensoriso}{\lambda}
\newcommand{\ltensorisoKl}{\lambda_{\mathbf{Kl}}}

\newcommand{\id}{\mathit{id}}
\newcommand{\ev}{\mathit{ev}}

\newcommand{\ltensor}[2]{#1 \cdot #2}
\newcommand{\ltensorname}{\ltensor {-_1}{-_2}}
\newcommand{\lpower}[2]{#2^{#1}}

\newcommand{\helpiso}[2]{\phi} % A map used in a definition below
\newcommand{\fixedcatfont}{\mathbf}

\newcommand{\denlb}{[\![}
\newcommand{\denrb}{]\!]}
\newcommand{\den}[1]{\denlb{#1}\denrb}

\newcommand{\pair}[2]{\langle #1 , #2 \rangle}
\newcommand{\tuple}[2]{\langle #1 , \dots , #2 \rangle}
\newcommand{\ituple}[2]{\langle #1 \rangle_{#2}} % an indexed tuple

\newcommand{\iso}{\cong}

\newcommand{\inv}[1]{#1^{-1}}

\newcommand{\SA}{A}
\newcommand{\SB}{B}
\newcommand{\SC}{C}

\newcommand{\algA}{\underline{A}}
\newcommand{\algB}{\underline{B}}
\newcommand{\algC}{\underline{C}}
\newcommand{\algD}{\underline{D}}

\newcommand{\stateobj}{{\underline{S}}}
\newcommand{\retobj}{\underline{R}}

\newcommand{\Set}{\mathbf{Set}}

\newcommand{\CatA}{\mathcal{A}}
\newcommand{\CatB}{\mathcal{B}}

\newcommand{\states}{{\underline{\mathrm{S}}}}

\newcommand{\SynEnrichedModel}{(\VSynE,\CSynE, \states)}

\newcommand{\SynKlModel}{(\VSynKl,\CSynKl, J)}

\newcommand{\VSynKl}{\mathcal{V}_{\textsc{fgcbv}}}
\newcommand{\CSynKl}{\mathcal{C}_{\textsc{fgcbv}}}
\newcommand{\VSynE}{\mathcal{V}_{\textsc{ecbv}}}
\newcommand{\CSynE}{\mathcal{C}_{\textsc{ecbv}}}

\newcommand{\Klmodel}{enriched Kleisli model}
\newcommand{\dKlmodel}{distributive Kleisli model}
\newcommand{\enrmodel}{enriched call-by-value model}
\newcommand{\denrmodel}{distributive enriched model}

\newcommand{\ENR}{\TwoCatFont{Enr}}
\newcommand{\dENR}{\TwoCatFont{dEnr}}
\newcommand{\CATECBV}{\TwoCatFont{Ecbv}}
\newcommand{\dCATECBV}{\TwoCatFont{dEcbv}}
\newcommand{\Freyd}{\TwoCatFont{Kleisli}}
\newcommand{\dFreyd}{\TwoCatFont{dKleisli}}
\newcommand{\CATECBVtheory}[1]{\dCATECBV_{#1}}
\newcommand{\Freydtheory}[1]{\dFreyd_{#1}}

\newcommand{\FreydToECBV}{\mathbf{St}}
\newcommand{\ECBVToFreyd}{\mathbf{Kl}}

\newcommand{\KlCat}[3]{\mathbf{Kl}_{#3}} % takes arguments V,C,S
\newcommand{\Klltensor}[2]{#1 \cdot_{\mathbf{Kl}} #2}

\newcommand{\KlHom}[3]{\KlCat{}{}{#1}(#2,#3)} % Enrichment of Kleisli categories

\newcommand{\stateiso}{\delta}
\newcommand{\VTwoCell}{\beta}
\newcommand{\CTwoCell}{\gamma}

\newcommand{\Cat}{\mathbf{Cat}}
\newcommand{\TwoCatFont}[1]{\mathfrak{#1}}
\newcommand{\CoprodCat}{\TwoCatFont{Coprod}}
\newcommand{\CAT}{\TwoCatFont{Cat}}
\newcommand{\CocompCat}{\TwoCatFont{Cocomp}}
\newcommand{\TwoMultiCat}{\TwoCatFont{K}}

\newcommand{\opcat}[1]{#1^{\mathrm{op}}}

\newcommand{\MonT}{T} % A monad, that's always useful

\newcommand{\lradj}[2]{{#1} \dashv {#2}} % an adjunction

\newcommand{\initobj}{0} % initial object
\newcommand{\GAP}{\hspace*{.5cm}}

\newcommand{\gnl}{\\[2ex]} % newline for gather environments. To control spacing
\newcommand{\eqdef}{\defeq}
\newcommand{\defeq}{\mathrel{\,\stackrel{\mbox{\tiny{$\mathrm{def}$}}}=\,}}

\newcommand{\vin}[2]{{\mathrm{in}_{#1}}(#2)}
\newcommand{\vinl}[1]{\vin{1}{#1}}
\newcommand{\vinr}[1]{\vin{2}{#1}}
\newcommand{\vcase}[5]{{\mathrm{case}} \, #1 \,\mathrm{of}\,( \vinl{#2}. \, #3 |  \vinr{#4}. \, #5)}

\newcommand{\algin}[2]{\comptype{\mathrm{in}}_{#1}(#2)}

\newcommand{\alginl}[1]{\algin{1}{#1}}
\newcommand{\alginr}[1]{\algin{2}{#1}}
\newcommand{\algcase}[5]{\comptype{\mathrm{case}} \, #1 \,\mathrm{of}\,( \alginl{#2}. \, #3 |  \alginr{#4}. \, #5)}

\newcommand{\eecproj}[2]{\pi_{#1}(#2)}
\newcommand{\eecprojb}[2]{(\pi_{#1}\,#2)}

\newcommand{\algimage}[1]{\comptype{\mathrm{?}}(#1)}

\newcommand{\sub}[3]{#1[^{#2}\!/\!_{#3}]}

\begin{document}

%\title{Linearly used state in models of call-by-value}
\title[Linear usage of state]{Linear usage of state\rsuper*}

\author[R.~E.~M{\o}gelberg]{Rasmus Ejlers M{\o}gelberg\rsuper a}
\address{{\lsuper a}IT University of Copenhagen, Denmark}
\thanks{{\lsuper a}Research supported by the Danish Agency for Science, Technology and Innovation.}
\email{mogel@itu.dk}

\author[S.~Staton]{Sam Staton\rsuper b}
\address{{\lsuper b}Radboud University Nijmegen, Netherlands}
\thanks{{\lsuper b}Research supported by EPSRC Fellowship EP/E042414/1,
ANR Projet CHOCO, the Isaac Newton Trust, and ERC Projects ECSYM and QCLS}
\email{s.staton@cs.ru.nl}

\titlecomment{{\lsuper*}This article expands on a paper presented at 
the Fourth International Conference on Algebra and Coalgebra in Computer Science (CALCO 2011).}

\begin{abstract}
  We investigate the phenomenon that \emph{every monad is a linear state
  monad}. We do this by studying a fully-complete state-passing
  translation from an impure call-by-value language to a new linear
  type theory: the enriched call-by-value calculus.  The results are not specific
  to store, but can be applied to any computational effect expressible
  using algebraic operations, even
  to effects that are not usually thought of as stateful.  There is a
  bijective correspondence between 
  generic effects in the source language
  and state access operations in the enriched call-by-value calculus.
\vspace{2mm}

  From the perspective of categorical models, the enriched call-by-value calculus 
  suggests a refinement of the traditional Kleisli models of effectful
  call-by-value 
  languages. The new models can be understood as enriched adjunctions.
\end{abstract}

\maketitle

\section{Introduction}

\subsection{Informal motivation}
The state-passing translation 
transforms a stateful program into a pure function. 
As an illustration, consider the following ML program
which uses a single fixed memory cell \verb|l : int ref|.
\begin{verbatim}
       - fun f x = let val y = !l in l := x ; y end ;
       val f = fn : int -> int
\end{verbatim}
The state-passing translation transforms that program into 
the following pure function which takes the state as an argument
and returns the updated state as a result.
\begin{verbatim}
       - fun f (x,s) = let val y = s val s' = x in (y,s') end ;      
       val f = fn : int * int -> int * int
\end{verbatim}
The state-passing translation is straightforward
if the program only uses a fixed, finite area of memory of type $\states$:
an impure program of type $\VA\pto\VB$ becomes a pure program
of type $\VA\times \states\to \VB\times \states$.

To what extent does the state-passing translation apply to programs with
other effects?
In this article we develop the idea
that, from a semantic perspective, 
all effects can be understood as state effects.
Central to our treatment of state is the idea of 
linear usage: in general, computations cannot copy the state and 
save it for later, nor
can they discard the state and insert a new one instead.
In 1972, Strachey wrote \cite{Strachey72}:
\begin{quote}
\emph{The state transformation produced by obeying a command is essentially irreversible and it is, by the nature of the computers we use, impossible to have
more than one version of $\sigma$ \emph{[the state]} at any one time.}
\end{quote}
Historically, the importance of the linearity of state
arose in the study of 
programs with private store. 
In this setting, the naive state-passing translation does not preserve 
contextual equivalence. 
For instance, the function `snapback' takes a
stateful computation $f \co \VA \pto \VB$ and returns the computation that
executes $f$ but then snaps back to the original state:
\hide{\[ 
\textrm{snapback}\defeq
\lam{f \co (\VA\times \states \to \VB\times \states)}{(\VA\pto \VB)}{\lam{x \co \VA \times \states}{}{\pair{\pi_1(f(x))}{\pi_2(x)}}} \co (\VA\pto \VB) \to (\VA\pto \VB)\]}
\[ 
\textrm{snapback}\defeq
\lam{f \co (\VA\times \states \to \VB\times \states)}{(\VA\pto \VB)}{\lam{(a,s)\co \VA \times \states}{}{\pair{\pi_1(f(a,s))}{s}}} \co (\VA\pto \VB) \to (\VA\pto \VB)\]
The snapback program does not arise as the state-passing translation
of an impure program. 
No impure program could tamper with the private store in the way 
that the snapback program does.
In other words, the state-passing translation
is not fully complete.
One can use this fact to show that 
contextual equivalence is not preserved by the 
state-passing translation. 
Sieber~\cite{DBLP:conf/mfcs/Sieber94} 
insisted that every function be wrapped in snapback to obtain 
full abstraction for his simple denotational model.

O'Hearn and Reynolds~\cite{OHearn:R:00}
resolved these difficulties with private store
by moving to a linear typing system.
Linear usage of state can be expressed syntactically 
by considering a stateful computation of
type $\VA \pto \VB$ as a linear map of type $\ltensortype \VA
\states \lfun \ltensortype \VB\states$.  The type of states~$\states$
must be used linearly, but the argument type $\VA$ and the return type $\VB$
can be used arbitrarily. 
The snapback program violates these linear typing constraints.

This notation is reminiscent of Girard's linear
logic~\cite{DBLP:journals/tcs/Girard87}.  Our starting point is
actually a more refined calculus, the enriched effect calculus, which
was developed by Egger, M\o gelberg and
Simpson~\cite{EEC:journal,Mogelberg:CSL:09,Mogelberg:fossacs:10,EEC:LCPS:journal} as a
way of investigating linear usage of resources such as state.

\subsubsection*{All effects are state effects}
In this paper we develop this linear typing discipline to show that
\emph{all} effects can be understood as state effects and that there
is \emph{always} a fully-complete linear-use state-passing
translation.  Our analysis applies even to effects that do not involve
the memory of a computer, like printing or coin-tossing.  (For this
reason we use the word `store' to refer to memory related effects and
`state' for the general notion.)  We now provide two informal
explanations of this general phenomenon.

A first informal explanation is that an inhabitant of the `state' type 
$\states$ 
is an entire history of the universe. The history of the universe 
certainly cannot be discarded or duplicated. 
To be slightly more precise, if the effect in question is 
printing, then a `state' is a string of everything that has been printed
so far. 

A second informal explanation involves Jeffrey's 
graphical notation~\cite{jeffrey-premonoidal-graphics}.
Jeffrey noticed that a naive graphical notation for 
composition of impure functions does not work,
for it describes how functions depend on their arguments 
but it does not describe the order of side effects:
\\
\begin{center}
\includegraphics{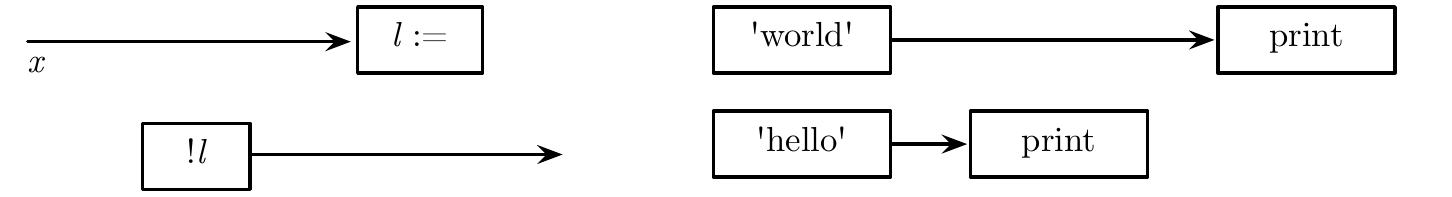}
\end{center}
To make the 
order of evaluation visible in the graphical notation, 
Jeffrey adds a special kind of edge which must be treated linearly.
This is what we call state.
\\
\begin{center}
\includegraphics{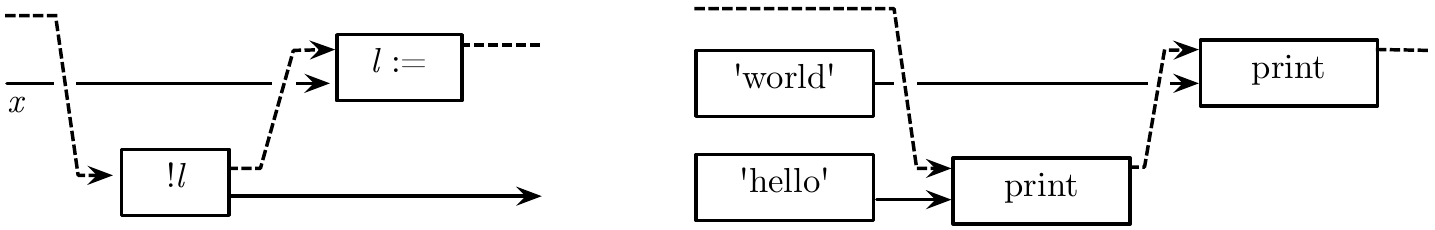}
\end{center}

Our contribution in this paper is foundational,
but let us speculate briefly on possible applications.
State plays an important role in many aspects of semantics,
including operational semantics and Hoare
logic.  
The type $\ltensortype A\states$ can be thought of as a type of 
configurations (program/state pairs) 
for a general abstract machine.
This might pave the way for a general framework for operational semantics, 
building 
on the ideas of Plotkin and Power~\cite{Plotkin:Power:08}.

We also note that variations on linear-use state-passing style
are already used to accommodate a broad class of effects 
within pure languages such as Clean~\cite{DBLP:journals/jfp/AchtenP95} and Mercury~\cite{DBLP:journals/jlp/SomogyiHC96}.

\subsection{The state-passing translation}
The source language of our translation is an impure functional
language with
product and function types:
\[
\sigma,\tau\, ::=\,
1~|~\sigma\times \tau~|~\sigma\pto \tau~|~\dots
\]
We also consider sum types and unspecified base types.
We adopt the 
call-by-value calling convention because it is 
most natural one for effectful programs.
For simplicity we choose a fine-grain language
in which the order of evaluation is totally explicit.
The full syntax, equational theory
and framework for denotational semantics is in Section~\ref{sec:fgcbv}.
Sum types are treated in Section~\ref{sec:sums}.

The types of the target language are essentially the minimal
fragment of the enriched effect calculus (EEC, \cite{Mogelberg:CSL:09,EEC:journal}) 
that is 
needed for the linear-use state-passing translation.
To enforce linear usage, EEC borrows some techniques and notation
from linear logic ($!$, $\tensor$, and $\lpop$)
but in a careful way, enforced by the following grammar:
\begin{align*}
\VA,\VB\,&::=\,1~|~\VA\times\VB~|~\CA\lpop \CB~|~\dots&&\textit{value types}
\\
\CA,\CB\,&::=\,\ltensortype \VA \CB~|~\states~|~\dots&&\textit{computation types}\text{ (underlined)}
\end{align*}
The full syntax, equational theory
and framework for denotational semantics is in 
Section~\ref{sec:ecbv}.

In Section~\ref{sec:translation}
we provide a linear-use state passing translation
from the source language to the target language
The translation on types 
takes a type $\tau$ 
of the source language to a value type~$\CBVtoEECbase \states\tau$ of the target language, by induction on the structure of types:
\[
\CBVtoEECbase \states 1\defeq 1
\qquad
\CBVtoEECbase \states {(\sigma\times \tau)} 
\defeq 
\CBVtoEECbase \states{\sigma}
\times \CBVtoEECbase \states\tau
\qquad
\CBVtoEECbase \states{(\sigma\pto\tau)}
\defeq 
\ltensortype{\CBVtoEECbase \states\sigma}\states\lpop
\ltensortype{\CBVtoEECbase \states\tau}\states\qquad\dots
\]
Theorem~\ref{thm:full:faithful} says that 
the linear-use state-passing translation is \emph{fully complete}:
it describes a bijective correspondence between
programs in the source language of type $\tau$ 
and programs in the target language of type $\CBVtoEECbase \states\tau$.
This means that contextual equivalence is preserved by 
linear-use state-passing translation.

\subsubsection*{Enriched category theory}
The constructions involved in the target language
have an elegant categorical formulation:
they can be modelled in any enriched 
category with copowers. For this reason we call 
the target language \emph{the enriched call-by-value calculus}.

Recall that in categorical semantics
a type $\VA$ is interpreted as an object $\den\VA$ of a category,
a context $\Gamma$ 
is interpreted as an object $\den\Gamma$ of a category,
and a term-in-context $\Gamma\vdash t : \VA$ is interpreted
as a morphism $\den t:\den\Gamma\to \den\VA$.
In the enriched call-by-value calculus,
there are two classes of type: computation types and value types.
A categorical analysis must thus involve two categories:
a category $\VCat$ whose objects interpret value types,
and a category $\CCat$ whose objects interpret computation types.
The structure of the types dictates the  
structure that the categories $\VCat$ and $\CCat$ must have:
\begin{itemize}
\item For the product types ($1$, $\times$), $\VCat$ must have
finite products in the categorical sense.
\item For the tensor type ($\ltensortype \VA\CB$),
$\VCat$ must act on $\CCat$. This means that there is a functor
$(\ltensorname):\VCat\times\CCat\to \CCat$ such that
$\ltensor 1\algX\cong \algX$ and 
$\ltensor{(A\times B)} \algX\cong \ltensor A {(\ltensor B \algX)}$.
\item For the linear function type $\lpop$,
the action $(\ltensorname)$ must have a right adjoint in
its first argument.
\end{itemize}
The linear function type $\lpop$ ensures that the 
space of 
morphisms $\algX \to \algY$ forms an object of $\VCat$.
So $\CCat$ is enriched in $\VCat$, and the action structure
provides copowers.
There are many examples of enriched models, including 
categories of algebras and Kleisli categories
(understood as `closed Freyd categories'~\cite{Levy:03}). 
Full details are in Section~\ref{sec:ecbv}. 

In Section~\ref{sec:relating:models} we explain that
the connection with Kleisli categories gives us a categorical 
explanation of the linear-use state-passing translation.
We use a semantic 
approach to prove the full completeness 
of the translation.
We show that the
traditional models of effectful call-by-value languages, using monads
and Kleisli constructions, form a coreflective subcategory 
of the enriched models, and that the state-passing translation is the
unit of the coreflection.  
\hide{Our main semantic result provides a
bijective correspondence between comodel structures on a state object
and model structures on the induced linear state monad. This extends
Plotkin and Power's correspondence between algebraic operations and
generic effects~\cite{Plotkin:Power:03} with a third component: state access operations.}

\subsection{Relationship with monads}
\label{sec:intro:monads}
%Moggi proposed to encapsulate computational effects
%within a purely
%functional language by using monads~\cite{Moggi:91}. The central
%idea behind this is to distinguish between a type of values such as
%$(\nat)$, and a type of computations $\MonT(\nat)$ that may return a
%value of type $\nat$ but can also do other things along the way.
%It is striking that there is no notion of state in the theory of monads. After all, imperative behaviour is often about changing or branching on the state of the machine. 
%\hide{While state is most naturally associated with certain effects like store effects, 
%in this paper 
%we shall see that all algebraic effects can be viewed in this way.\footnote{For this reason we use the terminology ``store effects'' for the specific (memory access operations) and reserve ``state'' for the general notion.}}%
%Of course, there are particular state monads.
%We will show that \emph{every} monad can be understood as a state monad
%under linear usage assumptions.

In an enriched situation, because $\states\lpop{(-)}$ is right adjoint to 
$\ltensortype {(-)}\states$,
any computation type $\states$ induces a monad on value types:
\begin{equation}
\label{eqn:statemonad}
\states\lpop \ltensortype {(-)}\states
\end{equation}
We call this a \emph{linear-use state monad}.
In Section~\ref{sec:monads} we show that 
every monad arises as a linear-use state monad. 
In brief, the argument is as follows.
The right notion of monad for programming languages is
`strong monad $T$ on a category $\VCat$ with finite products and 
Kleisli exponentials';
for every such monad $T$,
the Kleisli category $\CCat$ is $\VCat$-enriched and has copowers
and is thus a model of the enriched call-by-value language.
The object $1$ in $\CCat$ 
%(which can be thought of as
%the free algebra on the terminal object of $\VCat$) 
induces a linear-use state monad
on $\VCat$  \eqref{eqn:statemonad} 
which is isomorphic to
the original monad $T$ on $\VCat$.

%The enriched model can give some guidance as to why a monad 
%is relevant.
For the simplest case, consider the following monad on the category
of sets for storing a single bit. 
%\begin{equation}
%\label{eqn:bitstatemonad}
%\begin{gathered}
%T(X)\defeq X\times 2\times X\times 2\qquad
%\text{where }2\defeq \{0,1\}
%\\
%\eta_X(x)\defeq(x,0,x,1)
%\qquad
%\mu_X((x_0,i_0,x_1,i_1),j,(y_0,k_0,y_1,k_1),l)
%\defeq 
%(x_j,i_j,y_l,k_l)\text.
%\end{gathered}
%\end{equation}
If we let $\VCat=\CCat=\Set$ and 
let $\states=2$ (there are two states), 
then the linear-use state monad is the usual state monad
$TX = 2 \to (X \times 2)$.

In general, the linear-use state monad \eqref{eqn:statemonad} arises
from an adjunction that is parameterized in $\states$.
Atkey~\cite{a-parammonad} has investigated monads that arise from
parameterized adjunctions: in fact Atkey's parameterized monads are
essentially the same as our enriched models.  We return to this in
Section~\ref{sec:parammonad}.

\subsection{Algebraic theories and state access operations}
In order to explain the general connection between
effects and state, we turn to the analysis of effects begun by Plotkin and Power~\cite{Plotkin:Power:03}.
Categorically, a finitary monad on the category of sets
is essentially the same thing as an algebraic  theory.
Plotkin and Power proposed to use this
connection 
to investigate computational effects from the perspective of universal algebra.

The analysis of Plotkin and Power centres around the 
correspondence between \emph{algebraic operations} and \emph{generic effects}.
More precisely, for a monad $T$ the following data are equivalent:
\begin{enumerate}
\item An \emph{algebraic operation}:
roughly, a functorial assignment of an $n$-ary function
${X^n\to X}$ to the carrier of each $T$-algebra ${T(X)\to X}$;
\item A \emph{generic effect}: a function $1\to T(n)$. 
\end{enumerate}
For instance, 
consider the store monad for a single bit of memory 
(using the isomorphic presentation 
$T=((-)\times 2\times(-)\times 2)$).
Each $T$-algebra 
$\xi:T(X)\to X$ supports 
an algebraic operation 
$\mathrel{?_X}:X\times X\to X$
in a functorial way: 
let ${x\mathrel{?_X} y=\xi(x,0,y,1)}$. 
If we understand elements of $T$-algebras
as computations, then ${x\mathrel{?_X}y}$ is the computation that  reads
the bit in memory 
and then branches to either~$x$ or~$y$ depending on what was read.

The corresponding generic effect $\mathtt{deref}:1\to T(2)$ 
is given by $\mathtt{deref}() = (0,0,1,1)$.
It can be thought of as the command that 
reads the bit and returns the contents of memory.

We have already explained (\S\ref{sec:intro:monads})
that all monads can be understood as linear-use state monads,
$T=(\states\lpop(\ltensortype{(-)}\states))$.
The data for algebraic operations and generic effects
can equivalently be given in terms of the following structure on 
the state object $\states$:
\begin{enumerate}
\item[(3)]
A \emph{state access operation}: 
a morphism $\states\to \ltensortype n \states$
in the computation category $\CCat$.
\end{enumerate}
For instance, let $\VCat=\CCat=\Set$ and $\states=2$.
This gives a state monad for a single bit of memory,
isomorphic to $((-)\times 2\times (-)\times 2)$.
The state access operation corresponding to 
reading the bit is simply the function 
$\mathrm{read}:2\to 2\times 2$ 
given by $\mathrm{read}(i)=(i,i)$, 
which reads 
the store and returns the result along with the store.
\hide{
For another example,
consider the monad $2^*\times (-)$ for printing bits.
We let $\VCat=\Set$. We let $\CCat$ be the category
of whose objects are sets $X$ equipped with 
a function $\mathrm{p}_X:2\times X\to X$;
morphisms in $\CCat$ are functions that preserve the structure.
The idea is that $X$ is a set of computations
and $\mathrm{p}_X(i,x)$ is a computation that first prints~$i$ 
and then continues as $x$.
The set of bit-strings $2^*$ is an object of $\CCat$,
with ${p_{2^*}(i,s)=is}$. 
The resulting linear-use state monad is the monad for printing bits.

We can investigate this theory in terms of algebraic operations,
generic effects and state access operations.
\begin{itemize}
\item A $T$-algebra is the same thing as an object of $\CCat$,
and so the algebraic operations are the unary
functions $\mathrm{p}_X(0,-):X\to X$, $\mathrm{p}_X(1,-):X\to X$.
\item The generic effects are 
functions 
$\mathtt{print\,0}:1\to T(1)$,
$\mathtt{print\,1}:1\to T(1)$
which describe commands that print 0 and 1 respectively.
\item The state access operations are
homomorphisms $\mathrm{print\,}i: 2^*\to 2^*$,
which take a bit string -- the history of what has been 
printed so far -- and append 0 or 1.
\end{itemize}
}

In Section~\ref{sec:modelsoftheories}
we provide more examples and investigate a general framework 
for algebraic presentations of theories of effects
using algebraic operations, generic effects and state access operations.

\subsection{The enriched effect calculus}
The work presented here grew out of work on the enriched effect calculus by 
Egger, M\o gelberg and Simpson~(EEC, \cite{EEC:journal,Mogelberg:CSL:09,Mogelberg:fossacs:10,EEC:LCPS:journal}). 
The enriched call-by-value calculus that we introduce in this paper is
a fragment of EEC 
and our categorical semantics is based on their work. 
Every model of EEC contains a monad, and one of the observations of~\cite{Mogelberg:fossacs:10} (see also~\cite[Example~4.2]{EEC:journal})
was that this monad can always be represented as a linear state monad.
A special case of the embedding theorem~\cite[Theorem~4]{Mogelberg:CSL:09} 
shows that given a strong monad $T$ on a category $\VCat$ with finite products and 
Kleisli exponentials we can embed $\VCat$ in a model of EEC preserving the monad 
and all other structure. This 
gives the motto: \emph{every monad embeds in a linear state monad}.

Since the enriched call-by-value calculus is a 
fragment of EEC (as opposed to full EEC), 
we allow ourselves a broader class of
models. In contrast to the earlier work on EEC, 
we do not include products among the computation types,
since they are not needed in the state-passing translation,
and so in our models the category $\CCat$ does not need to 
have products. This allows us to build models from Kleisli categories,
which typically do not have products, and this makes the relationship 
with monad models and closed Freyd categories much more straightforward.
In particular, in our setting every monad \emph{is} a linear state monad.

Nonetheless, in Section~\ref{sec:ecbvtoeec} we show that $\EEC$ is a
conservative extension of the enriched call-by-value calculus. This
shows that there is a fully-complete linear-use state translation into
$\EEC$.  This result is further evidence that $\EEC$ is a promising
calculus for reasoning about linear usage of effects. The related
papers~\cite{Mogelberg:fossacs:10,EEC:LCPS:journal} show how the linear-use
continuation passing translation arises from a natural dual model
construction on models of $\EEC$, and use this to prove a full completeness theorem 
similar to that proven here for the linear-use state-passing translation.
In fact, from the point of view of
$\EEC$ the two translations are surprisingly similar: the linear-use
state-passing translation is essentially dual to the linear-use
continuation-passing translation. This observation goes back to the work on
$\EEC$ and indeed duality plays a key role in~\cite{Mogelberg:fossacs:10} (although the relationship with state wasn't made explicit there). We draw it out explicitly in
Section~\ref{sec:cps}.

\subsubsection*{Acknowledgements.}
We thank Alex Simpson for help and encouragement. Also thanks to Lars Birkedal, Jeff Egger, Masahito Hasegawa, Shin-ya Katsumata and Paul Levy for helpful discussions. Diagrams are typeset using the \texttt{xymatrix} package
and Paul Taylor's \texttt{diagrams} package.

% !TEX root = Mogelberg-staton.tex

\newcommand{\ccat}{\mathfrak{K}}
\newcommand{\dcat}{\mathbb{D}}

\newcommand{\cbvinp}[2]{\codefont{in}_{#1}^p(#2)}
\newcommand{\cbvimagep}[1]{\codefont{image}^p(#1)}
\newcommand{\cbvcasep}[5]{\codefont{case}^p~#1~\codefont{of}~(\cbvinl{#2}.#3|%
\cbvinr{#4}.#5)}

\newcommand{\geff}{e} 
\newcommand{\geffj}[3]{\pj{#1}{#2}{#3}}
\newcommand{\geffin}[2]{\codefont{in}_{#1}(#2)} % coproduct inclusion in the effect calculus
\newcommand{\geffinname}[1]{\codefont{in}_{#1}} % coproduct inclusion in the effect calculus
\newcommand{\geffabs}[2]{#1 \text{ to } \{#2 \}}

\newcommand{\cbvin}[2]{{\codefont{in}}_{#1}(#2)}
\newcommand{\cbvinl}[1]{\cbvin{1}{#1}}
\newcommand{\cbvinr}[1]{\cbvin{2}{#1}}
\newcommand{\cbvcase}[5]{\codefont{case}~#1~\codefont{of}~(\cbvinl{#2}.#3|%
\cbvinr{#4}.#5)}
\newcommand{\cbvcasen}[6]{\codefont{case}~#1~\codefont{of}~%
(\codefont{in_1}({#2}).#3|%
\dots|\codefont{in_{#4}}({#5}).#6)}
\newcommand{\cbvcasenp}[6]{\codefont{case}^p~#1~\codefont{of}~%
(\codefont{in_1}({#2}).#3|%
\dots|\codefont{in_{#4}}({#5}).#6)}

\newcommand{\arit}[2]{#1;  #2}
\newcommand{\arityj}[3]{#1: \arit{#2}{#3}}

\newcommand{\Homset}[3]{\mathrm{Hom}_{#1}({#2},{#3})}

\section{Enriched call-by-value: a calculus for enriched categories with copowers}
\label{sec:ecbv}
The target language for the linear-use state translation
is a new calculus called the 
\emph{enriched call-by-value calculus} (ECBV), that we now introduce.
As we will explain, it is an internal language for 
enriched categories with copowers.

The enriched call-by-value calculus 
is a fragment of the enriched effect calculus (EEC),
which was introduced by Egger et al.~\cite{Mogelberg:CSL:09,EEC:journal}
as a calculus for reasoning about linear usage in computational effects.
The types of ECBV
can be understood as a fragment of linear
logic that is expressive enough to describe the 
linear state monad, $\EECstate\lpop {\ltensortype {(-)}\EECstate}$.
We will not dwell on the connection with linear logic here.

\subsection{Type theory and equational theory}
\label{sec:ecbvsyntax}
The enriched call-by-value calculus has two collections of types:
value types and computation types.  We use $\VconstA, \VconstB, \dots$
to range over a set of \emph{value type constants}, and
$\CconstA,\CconstB,\dots$ to range over a disjoint set of
\emph{computation type constants}.  We then use 
upright letters $\VA,\VB,\dots$ to
range over value types, and underlined letters 
$\CA,\CB,\dots$ to range over
computation types, which are specified by the grammar below:
\begin{align*}
\VA,\VB \, & ::= \, \VconstA \,\mid \,\valone \,\mid \,\VA \prodtype \VB \, \mid\,  \CA \lpop \CB  \\
\CA,\CB \, & ::= \, \CconstA \,\mid  \,   \ltensortype{\VA}{\CB} \enspace .
\end{align*}
\noindent
Note that the construction $\ltensortype{\VA}{\CB}$ is indivisible:
the strings ${!}\VconstA$ and ${\CconstA\otimes \CconstB}$ 
are not well-formed types.
The stratification of types means that one cannot chain function 
types: the string $\CconstA\lpop(\CconstB\lpop \CconstC)$ is not well-formed.

Readers familiar with Levy's
Call-by-Push-Value~\cite{Levy:book} or $\EEC$~\cite{Mogelberg:CSL:09}
should note that there are no type constructors $F$ and $U$ for
shifting between value and computation types and that 
computation types are not included in the value types.
The only way to shift between value types and computation types
is by using tensor and function types. As we will see, this is the essence
of the state-passing translation.
\begin{figure*}[tph]
\framebox{
\begin{minipage}{.96\linewidth}
\emph{Types.}
\begin{align*}
\VA,\VB \, & ::= \, \VconstA \,\mid \,\valone \,\mid \,\VA \prodtype \VB \, \mid\,  \CA \lpop \CB  \\
\CA,\CB \, & ::= \, \CconstA \,\mid  \,   \ltensortype{\VA}{\CB} \enspace .
\end{align*}
\begin{center}
\line(1,0){350}\gnl
\end{center}
\emph{Term formation.}
\begin{center}\begin{gather*}
\prooftree
\justifies
\tj{\Gamma,\, \In{x}{\VA},\,\Gamma'}{x}{\VA}
\endprooftree
\GAP
\GAP
\prooftree
\justifies
\aj{\Gamma}{\In{z}{\CA}}{z}{\CA}
\endprooftree
\GAP
\GAP
\prooftree
\justifies 
\tj{\Gamma}{\algstar}{\valone}
\endprooftree
\gnl
\prooftree
\tj{\Gamma}{t}{\VA}
  \GAP
\tj{\Gamma}{u}{\VB}
\justifies 
\tj{\Gamma}{\pair{t}{u}}{\VA \prodtype \VB}
\endprooftree
\GAP\GAP
\prooftree
\tj{\Gamma}{t}{\VA_1 \prodtype \VA_2}
\justifies 
\tj{\Gamma}{\eecproj{i}{t}}{\VA_i}
\endprooftree
\gnl
\prooftree
\aj{\Gamma}{\In{z}{\CA}}{t}{\CB}
\justifies
\tj{\Gamma}{\llam{z}{\CA}{t}}{\CA \lfun \CB}
\endprooftree
\GAP 
\prooftree
\tj{\Gamma}{s}{\CA \lfun \CB} 
  \GAP
\aj{\Gamma}{\Delta}{t}{\CA} 
\justifies
\aj{\Gamma}{\Delta}{s [ t ]}{\CB}
\endprooftree
\gnl
\prooftree
\tj{\Gamma}{t}{\VA}
\GAP
\aj{\Gamma}{\Delta}{u}{\CB}
\justifies
\aj{\Gamma}{\Delta}{\ltensorterm{t}{u}}{\ltensortype{\VA}{\CB}}
\endprooftree
\GAP
\prooftree
\aj{\Gamma}{\Delta}{t}{\ltensortype{\VA}{\CB}}
\GAP
\aj{\Gamma, \, \In{x}{\VA}}{\In{z}{\CB}}{u}{\CC}
\justifies
\aj{\Gamma}{\Delta}{\letdot{x}{z}{t}{u}}{\CC}
\endprooftree
\end{gather*}
\end{center}
\begin{center}\mbox{}\\
\line(1,0){350}\gnl
\end{center}
\emph{Equality.} (We elide $\alpha$-equivalence, 
reflexivity, symmetry, transitivity and 
congruence laws.)
\begin{center}
\begin{gather*}
\begin{prooftree}
\tj\Gamma t 1
\justifies
\teq \Gamma t \star 1
\end{prooftree}
\ \quad\ 
\begin{prooftree}
\tj \Gamma {t_1}{\VA_1}
\quad
\tj \Gamma {t_2}{\VA_2}
\justifies 
\teq\Gamma {\prj i{\pair{t_1}{t_2}}} {t_i}{\VA_i}
\end{prooftree}
\gnl
\begin{prooftree}
{\tj\Gamma t {\VA_1\times \VA_2}}
\justifies
{\teq \Gamma{\pair{\fst{t}}{\snd{t}}} t {\VA_1\times \VA_2}}
\end{prooftree}
\gnl
\begin{prooftree}
{\aj\Gamma {z:\CA}t \CB}
\quad
{\aj\Gamma \Delta u \CA}
\justifies
{\aeq \Gamma \Delta {(\llam z \CA t)[u]}{t[u/z]}\CB}
\end{prooftree}
\qquad
\begin{prooftree}
{\tj\Gamma t {\CA\lpop \CB}}
\justifies
{\teq \Gamma {t}{\llam z\CA (t[z])}{\CA\lpop\CB}}
\end{prooftree}
\gnl
\begin{prooftree}
{\tj \Gamma t \VA}
\quad
{\aj\Gamma \Delta u \CB}
\quad
{\aj{\Gamma,x\co\VA}{z\co \CB}  v \CC}
\justifies
{\aeq \Gamma \Delta {\letdot x z {(\ltensortype t u)} v} {v[t/x,u/z]}\CC}
\end{prooftree}
\gnl
\begin{prooftree}
{\aj \Gamma \Delta t {\ltensortype\VA\CB}}
\quad
{\aj{\Gamma} {y\co\ltensortype \VA\CB} u \CC}
\justifies
{\aeq \Gamma \Delta {\letdot x z t {u[\ltensorterm x z/y]}} {u[t/y]}\CC}
\end{prooftree}
\end{gather*}
\end{center}\end{minipage}}
\caption{The enriched call-by-value calculus}
\label{figure:effects:typing}
\end{figure*}

The enriched call-by-value calculus has two basic typing judgements, written
\begin{equation}
\label{eq:form:typing:rules}
\tj{\Gamma}{t}{\VB} \GAP \text{and} \GAP \aj{\Gamma}{z \co \CA}t{\CB}
\end{equation}
In the both judgements, $\Gamma$ is an assignment of value types to
variables.  In the first judgement,~%
$\VB$ is a value type, and in the second judgement, both $\CA$ and
$\CB$ need to be computation types.  The second judgement should be
thought of as a judgement of linearity in the variable $z \co
\CA$. These judgements are defined inductively by the typing rules in
Figure~\ref{figure:effects:typing}, which are a restriction of the
rules of EEC~\cite{EEC:journal} 
to this type structure. In the figure,~$\Delta$ is an
assignment of a computation type to a single variable, as
in~(\ref{eq:form:typing:rules}).  The ideas behind the term language
for the function space $\CA\lpop\CB$ and the tensor
$\ltensortype\VA\CB$ go back to the early work on linear lambda
calculus.  In particular, the introduction rule for
$\ltensortype\VA\CB$ uses pairing, and the elimination rule uses a
pattern matching syntax.

%The equality theory includes
%$\alpha$, $\beta$ and $\eta$ rules and is exactly as
%for 
%EEC~\cite[Sec.~3]{Mogelberg:CSL:09}.

\subsection{Enriched call-by-value models}
\label{sec:adjmodels}
The categorical notion of model for {\ECBV}
involves basic concepts from enriched category theory~\cite{Kelly:Book}.
In summary, a model of the language comprises two categories,
$\VCat$ and $\CCat$, interpreting the value and computation types 
respectively; the function type $\CA\lpop\CB$ provides 
an enrichment of $\CCat$ in $\VCat$, and the universal property of the tensor
type $\ltensortype\VA\CB$ is the copower, sometimes called tensor.

We now make this precise. Let us recall some rudiments.
Following~\cite{JanelidzeKelly:actions,GordonPower:EnrichmentThroughVariation}, we begin with actions of
categories.  Let $\VCat$ be a category with finite products
(by which we mean that it has a chosen terminal object and chosen
binary products).
Recall that an \emph{action} of $\VCat$ on a category $\CCat$ is a functor
${\ltensor{}{}\colon\VCat\times\CCat\to\CCat}$ 
together with 
two natural isomorphisms, unit ${(\ltensor
  1\algA)\cong\algA}$ and associativity
${(\ltensor{(\SA\times\SB)}\algC)\cong(\ltensor\SA{(\ltensor{\SB}\algC)})}$,
that cohere with the isomorphisms arising from 
the product structure of~$\VCat$ in the following sense:
\[
\xymatrix{
\ltensor {(\SA\times 1)}\algD\ar[dr]^\cong\ar[d]_\cong 
&&
\ltensor {(1\times \SA)}\algD\ar[dr]^\cong\ar[d]_\cong 
\\
\ltensor{\SA}{(\ltensor 1\algD)}\ar[r]_\cong&\ltensor \SA\algD
&
\ltensor{1}{(\ltensor \SA\algD)}\ar[r]_\cong&\ltensor \SA\algD
}\]\[
\xymatrix{
\ltensor{((\SA\times \SB)\times\SC)}\algD
\ar[rr]^\cong\ar[d]_\cong
&&\ltensor{(\SA\times \SB)}{(\ltensor\SC\algD)}
\ar[d]^\cong
\\
\ltensor{(\SA\times(\SB\times \SC))}\algD
\ar[r]_\cong
&
\ltensor{\SA}{(\ltensor{(\SB\times \SC)}\algD)}
\ar[r]_\cong
&\ltensor{\SA}{(\ltensor{\SB}{(\ltensor \SC\algD)})}
}
\]
(We underline objects of $\CCat$ to distinguish them from objects of 
$\VCat$.)

An \emph{enrichment of a category $\CCat$ in $\VCat$ with copowers} is
determined by an action of $\VCat$ on $\CCat$ such that each functor
${(\ltensor-\algB)\colon\VCat\to\CCat}$ has a right adjoint,
${\CHom\algB-\colon\CCat\to\VCat}$.  Then $\ltensor\SA\algB$ is called
a copower, and $\CHom\algB\algC$ is called enrichment.  
We write $\Homset{\CCat}{\algB}{\algC}$ for the usual hom-set of $\CCat$
to distinguish it from the enrichment.
%Recall also
%that a \emph{power} is a right adjoint to
%$(\ltensor\SA-)\colon\CCat\to\CCat$ (we will need this in Section~\ref{sec:modelsoftheories}).

\begin{defi}
\label{def:enrichedmodel}
An \emph{enriched call-by-value model} (or simply \emph{enriched model}) is given by a category 
$\VCat$ with finite products
and a category $\CCat$ enriched in $\VCat$ with copowers.
\end{defi}
In Section~\ref{sec:examplesenriched}
we will illustrate the definition with some examples of enriched models.
First, let us clarify the 
semantics for {\ECBV} 
in an enriched model.
The interpretation is similar to the semantics of 
$\EEC$ proposed by Egger et al.~\cite{Mogelberg:CSL:09}. 
\begin{itemize}
\item 
A value type $\VA$ is interpreted as an object $\den\VA$ of $\VCat$,
and a computation type $\CA$ is interpreted 
as an object $\den\CA$ of $\CCat$, as follows.
The interpretation is defined by induction on the structure of types. 
First, for each value type constant~$\VconstA$, 
an object $\den\VconstA$ of $\VCat$ is given,
and for each computation type constant~$\CconstA$ 
an object $\den\CconstA$ of $\CCat$ is given.
The product types are interpreted as products in 
$\VCat$.
The remaining type constructions are interpreted using the 
enriched structure:
we let $\den{\ltensortype \VA\CB}\defeq(\ltensor {\den\VA}{\den\CB})$,
and $\den{\CA\lpop\CB}\defeq\CCat(\den\CA,\den\CB)$. 
\item A value context $\Gamma=(x_1\co\VA_1,\dots, x_n\co\VA_n)$
is interpreted as a product $\den{\VA_1}\times \dots\times \den{\VA_n}$.
in~$\VCat$.
A computation context $\Delta=(z\co\CA)$ is interpreted as the object
$\den\CA$ in~$\CCat$.
\item 
A judgement ${\tj\Gamma t\VA}$ is interpreted
as a morphism $\den\Gamma\to\den\VA$ in $\VCat$,
and a judgement ${\aj\Gamma\Delta t\CA}$ is interpreted
as a morphism $\ltensor{\den\Gamma}{\den\Delta}\to\den\CA$ in~$\CCat$.
This definition is made by induction on the structure of 
typing derivations, making use of the universal properties of the 
interpretations of the types.
For illustration, we consider the following two rules:
\[
\prooftree
\aj{\Gamma}{\In{z}{\CA}}{t}{\CB}
\justifies
\tj{\Gamma}{\llam{z}{\CA}{t}}{\CA \lfun \CB}
\endprooftree
\GAP \GAP 
\prooftree
\aj{\Gamma}{\Delta}{t}{\ltensortype{\VA}{\CB}}
\GAP
\aj{\Gamma, \, \In{x}{\VA}}{\In{z}{\CB}}{u}{\CC}
\justifies
\aj{\Gamma}{\Delta}{\letdot{x}{z}{t}{u}}{\CC}
\endprooftree
\]
In dealing with the linear lambda abstraction rule,
the induction principle gives us an interpretation
$\den {t}:\ltensor{\den\Gamma}{\den \CA}\to \den \CB$ in $\CCat$
which we use to form 
${\den{\llam z\CA t}: {\den\Gamma}\to \CCat(\den \CA,\den \CB)}$ 
in $\VCat$, using the natural bijection
that defines the relationship between the copower and the enrichment:
\[
\Homset{\CCat}{\ltensor{\SA}{\algB}}{\algC}\cong\Homset{\VCat}{\SA}{\CCat(\algB,\algC)}\text.
\]
For the pattern matching rule,
we assume morphisms 
\[\den t:\ltensor{\den\Gamma}{\den\Delta}\to \ltensor {\den \VA}{\den \CB}
\qquad\qquad
\den u:\ltensor{(\den\Gamma\times \den\VA)}{\den\CB}\to {\den \CC}\]
in $\CCat$ and use them to 
define $\den {\letdot x z t u}:\ltensor{\den\Gamma}{\den\Delta}\to\den{\CC}$
as the following composite:
\begin{align*}
\ltensor{\den\Gamma}{ \den\Delta}
\xrightarrow{\text{diag.}}
\ltensor{(\den\Gamma\times \den \Gamma)}{ \den\Delta}
\xrightarrow{\cong}
\ltensor{\den\Gamma}{(\ltensor{\den \Gamma} {\den\Delta})}
\xrightarrow {\ltensor {\den \Gamma}{ \den t}}
\ltensor{\den\Gamma}{(\ltensor{\den \VA} {\den\CB})}\quad
\\\xrightarrow{\cong}
\ltensor{(\den\Gamma\times \den \VA)}{ \den\CB}
\xrightarrow{\den u}
\den \CC
\end{align*}
\end{itemize}

\begin{prop}
The interpretation of the enriched call-by-value calculus 
in an enriched model 
$(\VCat,\CCat)$ is sound:
\begin{enumerate}
\item If $\teq \Gamma t u \VA$ then 
$\den t=\den u:\den \Gamma\to\den \VA$
in $\VCat$.
\item If $\aeq \Gamma \Delta t u \CA$ then 
$\den t=\den u:\ltensor{\den \Gamma}{\den\Delta}
\to\den \CA$
in $\CCat$.
\end{enumerate}
\end{prop}
\begin{proofnotes}
This is proved by induction on the structure of the 
derivation of~$(\equiv)$.
The following substitution lemma is helpful:
\hide{\item If $\tj \Gamma t \VA$ and $\aj {\Gamma,x\co\VA,\Gamma'}{\Delta} u \CB$ 
then $u[t/x]$ is the following composite  in $\VCat$:
\[
\ltensor{(\den\Gamma\times \den {\Gamma'})}{\den\Delta}
\xrightarrow{\text{diag.}}
\ltensor{(\den\Gamma\times\den\Gamma\times\den {\Gamma'})}{\den\Delta}
\xrightarrow{(\ltensor{\den\Gamma\times \den t\times \den {\Gamma'})}{\den\Delta}}
{\ltensor{(\den\Gamma\times \den\VA\times\den{\Gamma'})}{\den\Delta}}
\xrightarrow{\den u}
{\den \CB}
\]
}
\begin{quote}
\emph{If $\aj \Gamma \Delta t \CA$ and $\aj \Gamma{z\co \CA} u \CB$ 
then $u[t/z]$ is the following composite morphism in $\CCat$:}
\[
\ltensor{\den\Gamma}{\den\Delta}
\xrightarrow{\text{diag.}}
\ltensor{(\den\Gamma\times \den\Gamma)}{\den\Delta}
\xrightarrow{\cong}
{\ltensor{\den\Gamma}{(\ltensor{\den\Gamma}{\den\Delta})}}
\xrightarrow{\ltensor{\den\Gamma}{\den t}}
\ltensor{\den\Gamma}{\den \CA}
\xrightarrow{\den u}
{\den \CB}
\]
\end{quote}\vspace{-1.367\baselineskip}
\end{proofnotes}
\subsection{Examples of enriched models}
\label{sec:examplesenriched}
We now list some examples of enriched models (Definition~\ref{def:enrichedmodel}).
\begin{enumerate}
\item If $\VCat=\Set$ then a $\VCat$-enriched category 
is just a locally small category.
The copower $\ltensor A \algB$ is the $A$-fold 
iterated coproduct of $\algB$, if it exists.
The following three examples are instances of this example.
\item \label{ex:monoids}Let $\VCat=\Set$ and let $\CCat$ be 
the category of monoids and homomorphisms.
The copower $\ltensor A\algB$, where $A$ is a set and $\algB$ is a monoid,
can be described as a quotient of the free monoid on the product of sets,
$(A\times |\algB|)^*/_\sim$.
Here $(A\times |\algB|)^*$ is the 
set of strings built of pairs in $(A\times |\algB|)$,
which is a monoid under concatenation
with the empty string $\epsilon$ as unit.
The equivalence relation~$(\sim)$ is generated by
$(a,b).(a,b')\sim (a,b.b')$ and $\epsilon \sim (a,\epsilon)$.
There is a similar description of the copower for any algebraic theory.
\item
We can modify Example~(\ref{ex:monoids}) 
to make $\CCat$ the category of 
\emph{free monoids} and monoid homomorphisms.
That is, the objects are monoids of the form~$B^*$.
In this situation, the copower satisfies 
$\ltensor A{B^*}=(A\times B)^*$.
In this example $\CCat$ is the Kleisli category 
for the free monoid monad. We will revisit Kleisli categories in 
Section~\ref{sec:monads}.
\item Let $\VCat=\CCat=\Set$,
with $\CCat$ considered with the ordinary enrichment.
The copower $\ltensor A B$ is the cartesian product of sets.
This situation generalizes to the situation where $\VCat=\CCat$ is an arbitrary
cartesian closed category.
\item Let $\VCat$ be the category of $\omega$-cpo's and continuous functions, 
and let $\CCat$ be the category of pointed $\omega$-cpo's and strict functions.
The enrichment $\CCat(\algA,\algB)$ is the cpo of strict functions under
the pointwise order, and the copower $\ltensor A\algB$ is the smash product
$A_\bot\otimes \algB$.
\item In the next section we will investigate a model
built from the syntax of the enriched 
call-by-value calculus.
\end{enumerate}

\subsection{The syntactic enriched model}
\label{sec:syn:enr:model}

The types and terms of the enriched call-by-value 
calculus form an enriched model
which we call the syntactic enriched model. 

Let $\VCat$ be the category whose objects are value types 
and where a morphism $\VA\to\VB$ is a term in context
$\tj{x\co\VA}{t}{\VB}$ modulo the equational theory (Figure~\ref{figure:effects:typing}) and modulo renaming the free variable~$x$.
The identity morphism
is the term $\tj{x\co\VA}{x}{\VA}$,
and composition of morphisms
\[
\VA
\xrightarrow{\tj{x\co\VA}{t}{\VB}}
\VB
\xrightarrow{\tj{y\co\VB}{u}{\VC}}
\VC
\]
is given by substitution: $u\circ t\defeq (\tj{x\co\VA}{\sub u t y}{\VC})$.
Since morphisms are actually equivalence classes, the 
well-definedness of substitution depends on the following 
substitution lemma
\begin{quote}
If $\teq{x\co\VA}{t}{t'}\VB$ and
$\teq{y\co\VB}{u}{u'}\VC$ 
\\then $\teq {x\co \VA}{\sub u t y}
{\sub {u'} {t'} y}\VC$
\end{quote}
which is proved by induction on the derivation of 
$u\equiv u'$. 

The laws of associativity and identity for composition 
are immediate. For instance, associativity 
amounts to the following syntactic identity:
\begin{quote}
If $\tj{x\co\VA}{t}\VB$, \quad
$\tj{y\co\VB}{u}\VC$ 
and
$\tj{z\co \VC}v \VD$ 
\\then $\teq {x\co \VA}{\sub {(\sub v u z)} t y}
{\sub v {\sub u t y} z}\VD$\text.
\end{quote}

The category $\VCat$ has products, given by the product types.
The equations at product types
are exactly what is needed to guarantee the universal
properties.

Let $\CCat$ be the category whose objects are computation types
and where a morphism $\CA\to\CB$ is a term in context
$\aj{-}{z\co \CA}{t}{\CB}$ modulo the equational theory
and modulo renaming the free variable~$z$.
Identities and composition are defined in a similar way to $\VCat$.
The identity morphisms $\CA\to\CA$ are $\aj-{z\co\CA}{z}{\CA}$ 
and composition is by substitution.

The action of $\VCat$ on $\CCat$ is given on objects by the tensor type: 
let $\ltensor \VA\CB\defeq \ltensortype\VA\CB$.
Given morphisms 
\[\VA\xrightarrow{\tj {x\co\VA}t{\VA'}}\VA'\text{ in $\VCat$}
\qquad\text{and}\qquad
\CB\xrightarrow{\aj-{z\co\CB}u{\CB'}}\CB'\text{ in $\CCat$}\]
we define a morphism $\ltensor t u:(\ltensor \VA\CB)\to(\ltensor {\VA'}{\CB'})$
in $\CCat$ by
\[\ltensor t u\defeq \left(\aj-{z'\co{\ltensortype \VA\CB}}{\letdot x{z}{z'}{\ltensorterm t u}}{\ltensortype{\VA'}{\CB'}}\right)\text.\]
Functoriality follows from the equational theory of ECBV.
The unit and associativity isomorphisms are straightforward to define.
For example, associativity ($\ltensor {(A\times B)}\algC\cong \ltensor A{(\ltensor B \algC)}$) is given by exhibiting
 mutual inverses at all types:
\begin{align*}
&\aj-{z\co{\ltensortype{(\VA\times \VB)}\CC}}
{\letdot x{z'} z {\ltensorterm{\eecprojb1x}{(\ltensorterm{\eecprojb2x}{z'})}}}
{\ltensortype \VA{(\ltensortype \VB\CC)}}
\\
&\aj-{z\co{\ltensortype \VA{(\ltensortype \VB\CC)}}}{\letdot x{z'} z {\letdot y{z''} {z'}{\ltensorterm{(x,y)}{z''}}}}{\ltensortype{(\VA\times \VB)}\CC}
\end{align*}
It follows from the 
equational theory of ECBV that these are isomorphisms and are natural and 
coherent.

Finally, we discuss the enrichment of $\CCat$ in $\VCat$.
Given types $\VA$, $\CB$ and $\CC$ we describe a natural bijection
\[
\Homset{\CCat}{\ltensor\VA\CB}{\CC} \ \cong\ \Homset{\VCat}{\VA}{\CB\lpop \CC}
\]
From left to right the bijection takes a 
computation term ${\aj-{z\co{\ltensortype\VA\CB}}t\CC}$
to a value term 
${\tj{x\co\VA}{\llam b{\CB} {\sub t{(\ltensorterm xb)}z}}{\CB\lpop\CC}}$.
From right to left the bijection takes a 
value term
${\tj{x\co\VA}{u}{\CB\lpop\CC}}$
to a computation term
${\aj-{z\co(\ltensortype\VA\CB)}{\letdot xy z{u[y]}}\CC}$.

\subsection{Universal property of the syntactic enriched model}
\label{sec:enr:model:biinitial}

The syntactic model described in Section~\ref{sec:syn:enr:model} enjoys a universal property:
it is an initial object 
in a category of enriched models with structure preserving
functors as morphisms. Given any other enriched model, the unique morphism
from the syntactic model is given by interpretation of syntax in the model.

%As is usual in semantics [[going back to Lawvere?]] one can form a 
%category of enriched models and prove that the syntactic model is 
%initial, by showing that the interpretation of syntax into a model defines 
%a unique morphism of models. 
This semantic characterization 
of interpretation is standard in categorical semantics, and it
is useful for deriving syntactic results from semantics, as we shall see
in Section~\ref{sec:relating:models}. 
However, we shall also see (Section~\ref{sec:adj})
that we need to talk about morphisms of models preserving structure only 
up to isomorphism, and the syntactic model is not initial with respect 
to this collection of morphisms. Rather, interpretation defines a morphism
which is unique only up to unique isomorphism. In order to formulate
this kind of initiality, we need to be able to assess when 
two morphisms between enriched models are isomorphic.
Thus the enriched models
form a 
groupoid-enriched category,
i.e., a 2-category in which each 2-cell is invertible.
The idea of using groupoid-enriched 
categories of models has been around for a long time 
(e.g. \cite[\S 8]{dk-alggraph})
and has been used in previous work on the enriched effect 
calculus~\cite{Mogelberg:CSL:09,Mogelberg:fossacs:10}. 

A precise definition of the 2-category of enriched models $\ENR$ can be found in 
Appendix~\ref{app:cats:of:models}. Here we just state the initiality property of the syntactic model.
First recall the following definition from 2-category theory 
(see e.g. \cite[\S6]{k-twocatlimits}).%Power~\cite{Power:2:cat:notes}).

\begin{defi} \label{def:biinitial}
Let $\ccat$ be a 2-category. An object $\initobj$ of $\ccat$ is \emph{bi-initial} if
for any object $\SA$ the hom-category $\ccat(\initobj, \SA)$ is equivalent to 
the terminal category (i.e., the category with one object and one morphism).
\end{defi}
An equivalent way of stating bi-initiality is to say that 
for any other object $\SA$ 
there exists a 1-cell $\initobj \to \SA$
which is unique up to unique 2-isomorphism.
%such that any for any other 1-cell $f: \initobj \to \SA$, there is a unique isomorphism
%from $\initmorph$ to $f$. 

The syntactic model of Section~\ref{sec:syn:enr:model} is bi-initial in the category $\ENR$, but in this paper
we are much more interested in a category of enriched models  $(\VCat, \CCat)$ 
with a specified state object $\states$ in $\CCat$ (because of the relation to the state-passing
translation), and so we formulate bi-initiality with respect to a category 
$\CATECBV$ of these. Like all other structure, 1-cells of $\CATECBV$ are 
only required to preserve the state objects up to isomorphism. (See Appendix~\ref{app:CATECBV} 
for a full definition). We write $\SynEnrichedModel$ for the enriched model 
obtained as in Section~\ref{sec:syn:enr:model} from the syntax of the enriched call-by-value calculus
extended with a special computation type constant $\states$.

\begin{thm} \label{thm:ecbv:biinitial}
The model $\SynEnrichedModel$ is bi-initial in $\CATECBV$. The unique 
morphisms with domain $\SynEnrichedModel$ are given by interpretation of syntax into models.
\end{thm}

%%% Local Variables: 
%%% mode: latex
%%% TeX-master: "mogelberg-staton"
%%% End: 

% !TEX root = Mogelberg-staton.tex

\section{Fine-grain call-by-value, a calculus for Kleisli models}
\label{sec:fgcbv}

The source language for our state-passing translation is a call-by-value language equipped with an equational theory to be thought of as generated by some operational semantics, as in~\cite{Plotkin:65}. 
We use a `fine grain' call-by-value language, 
following Levy et al.~\cite{Levy:03,Levy:book}.
We use $\alpha$ to range over type constants. The types are given by the grammar
\[\sigma,\tau ::= 
\alpha \mid 1 \mid \sigma \times \tau \mid
\sigma \pto \tau \, .
\]
The function space $\pto$ is a call-by-value
one, which takes a value and produces a computation. 

The fine-grain call-by-value calculus ({\FGCBV}) has two typing
judgements, one for values and one for producers. These are written
$\vj{\Gamma}{V}{\STA}$ and $\pj{\Gamma}{M}{\STA}$.  The latter should
be thought of as typing computations which produce values in the type
judged but may also perform side-effects along the way. In both
judgements the variables of the contexts are to be considered as
placeholders for values. 
Typing rules along with equality rules are given in
Figure~\ref{fig:lambdac:typing}.
\begin{figure*}[tb]
\framebox{
\begin{minipage}{.96\linewidth}
%\footnotesize
\emph{Types.}
\[\sigma,\tau ::= 
\alpha \mid 1 \mid \sigma \times \tau \mid
\sigma \pto \tau \, .
\]
\begin{center}
\line(1,0){350}\gnl
\end{center}
\emph{Term formation.}
\begin{center}
\begin{gather*}
\begin{prooftree}
\justifies
\vj{\Gamma, x \co \STA, \Gamma'}{x}{\STA}
\end{prooftree}
\GAP
\begin{prooftree}
\justifies
\vj{\Gamma}{\star}{1} 
\end{prooftree}
\GAP
\begin{prooftree}
\vj{\Gamma}{V}{\STA_1 \times \STA_2}
\justifies
\vj{\Gamma}{\prj i{V}}{\STA_i}
\end{prooftree}
\GAP
\begin{prooftree}
\vj{\Gamma}{V_1}{\STA_1}
\GAP
\vj{\Gamma}{V_2}{\STA_2}
\justifies
\vj{\Gamma}{\pair{V_1}{V_2}}{\STA_1 \times \STA_2}
\end{prooftree}
\gnl
\begin{prooftree}
\vj{\Gamma}{V}{\STA}
\justifies
\pj{\Gamma}{\return{V}}{\STA}
\end{prooftree}
\GAP 
\begin{prooftree}
\pj{\Gamma}{M}{\STA}
\GAP
\pj{\Gamma, x \co \STA}{N}{\STB} 
\justifies
\pj{\Gamma}{\slet{x}{M}{N}}{\STB}
\end{prooftree}
\gnl
\begin{prooftree}
\pj{\Gamma, x \co \STA}{N}{\STB} 
\justifies
\vj{\Gamma}{\lam{x}{\STA}{N}}{\STA \pto \STB} 
\end{prooftree}
\GAP
\begin{prooftree}
\vj{\Gamma}{V}{\STA \pto \STB}
\GAP
\vj{\Gamma}{W}{\STA}
\justifies
\pj{\Gamma}{V \, W}{\STB}
\end{prooftree}
\end{gather*}
\end{center}
\begin{center}\mbox{}\gnl
\line(1,0){350}\gnl
\end{center}
\emph{Equality.} (We elide $\alpha$-equivalence,
reflexivity, symmetry, transitivity and 
congruence laws.)
\begin{center}
\begin{gather*}
\begin{prooftree}
\vj\Gamma M 1
\justifies
\veq \Gamma M \star 1
\end{prooftree}
\ \quad\ 
\begin{prooftree}
\vj \Gamma {V_1}{\sigma_2}
\quad
\vj \Gamma {V_2}{\sigma_2}
\justifies 
\veq\Gamma {\prj i{\pair{V_1}{V_2}}} {V_i}{\sigma_i}
\end{prooftree}
\ \quad\ 
\begin{prooftree}
{\vj\Gamma V {\sigma_1\times \sigma_2}}
\justifies
{\veq \Gamma{\pair{\fst{V}}{\snd{V}}} V {\sigma_1\times \sigma_2}}
\end{prooftree}
\gnl
\begin{prooftree}
\pj{\Gamma,x\co\STA}{M}{\tau}
\quad
\vj{\Gamma}{V}{\sigma}
\justifies
\veq \Gamma {(\lam{x}{\STA}{M})\,V} {M[V/x] } {\tau}
\end{prooftree}
\qquad
\begin{prooftree}
{\vj{\Gamma}{V}{\STA\pto\STB}}
\justifies
\veq\Gamma
{\lam{x}{\STA}{(V \, x)}}V {\STA\pto\STB}
\end{prooftree}
\gnl
\begin{prooftree}
\pj{\Gamma}M\STA
\justifies 
\peq\Gamma{\slet{x}{M}{\return{x}}} M \STA
\end{prooftree}
\qquad
\begin{prooftree}
\vj\Gamma V \STA
\quad\pj{\Gamma,x\co \STA}N\STB
\justifies
\peq \Gamma {\slet{x}{\return{V}}{N}}
{N[V/x]} \STB
\end{prooftree}
\gnl
\begin{prooftree}
{\pj\Gamma M \STA}
\quad
{\pj{\Gamma,x\co\STA} N \STB}
\quad
{\pj{\Gamma,y\co\STB} P \STC}
\justifies
\peq\Gamma{\slet{y}{(\slet{x}{M}{N})}{P}}{\slet{x}{M}{(\slet{y}{N}{P})}}\STC
\end{prooftree}
\end{gather*}%
\end{center}
\end{minipage}
}
\caption{Fine-grain call-by-value.
}
\label{fig:lambdac:typing}
\end{figure*}

The call-by-value language is called `fine grain' because the order of 
evaluation is explicit. Notice that the string
$(f(x),g(y))$ is not well-formed syntax:
one must specify the order of evaluation,
for instance, like this:
\[
\slet{x'} {f(x)}{\slet {y'}{g(y)}{(x',y')}}\text.
\]
Translations from a `coarser grain', more natural programming language
are given by Levy et al.~(\cite[\S 3]{Levy:03}, \cite[\S A.3.3]{Levy:book}).

\subsection{Interpretation in Kleisli models}
\label{sec:interpkleilsi}
The most natural way to interpret fine-grain call-by-value is 
two have two categories $\VCat$ and $\CCat$ to interpret
the judgements $\vjname$ and $\pjname$,
but to insist that the two categories have exactly 
the same objects, since in this language there is only one 
class of types.
\begin{defi}
\label{def:monadmodel}
An \emph{enriched Kleisli model}
is an enriched call-by-value model $(\VCat,\CCat)$ (Def.~\ref{def:enrichedmodel})
together with an 
identity-on-objects functor $J:\VCat\to\CCat$ that strictly
preserves copowers, which means that
$J(A\times B)=\ltensor A J(B)$ (naturally in $A$ and $B$)
and 
that the canonical isomorphisms induced by the product structure
are the coherent unit and associativity isomorphisms of the copowers:
\[
\ltensor1 JA =J(1\times A)\cong JA
\qquad
\ltensor{(A\times B)} JC = J((A\times B)\times C)\cong J(A\times (B\times C))
=
\ltensor A {(\ltensor B {JC})}\text.\]
\end{defi}
We will sometimes say `Kleisli model'
for `enriched Kleisli model'.
We use the name `Kleisli' because this definition
captures the situation where
$\CCat$ is the Kleisli category for a strong monad on $\VCat$.
The correspondence is explained in
Section~\ref{sec:monads}.

Kleisli models have earlier been called
`closed Freyd categories' by Levy et al.~\cite{Levy:03}.
Their original definition of closed Freyd category is based
on premonoidal categories;
the relationship with actions of categories 
and Kleisli models is observed by Levy~\cite[B.10]{Levy:book}.

A semantics for {\FGCBV} is given in a Kleisli model in a standard way.
\begin{itemize}
\item Each base type is given an interpretation 
$\den\alpha$ as an object of $\VCat$.
This interpretation is extended to all types:
$\den 1$ is the terminal object of $\VCat$;
$\den{\sigma\times \tau}$ is the product of 
$\den\sigma$ and $\den\tau$;
and 
$\den{\sigma\pto\tau}$ is defined using the enriched
structure of $\CCat$:
$\den{\sigma\pto\tau}\defeq \CCat(\den{\sigma}, \den{\tau})$.
\item 
A context $\Gamma=(x_1\co\sigma_1,\dots,x_n\co\sigma_n)$ 
is interpreted in $\VCat$ as a product 
$\den{\sigma_1}\times \dots \times \den{\sigma_n}$. 
\item 
A value type judgement 
${\vj{\Gamma}{V}{\STA}}$  
is interpreted as a 
morphism $\den{\Gamma}\to\den{\STA}$ in $\VCat$
and a producer type judgement $\pj{\Gamma}{M}{\STA}$ 
is interpreted as a morphism 
$\den{\Gamma}\to \den{\STA}$
in~$\CCat$.
This is defined by induction on the structure of derivations,
using the universal properties of Kleisli models.
For illustration we consider the following rule.
\[\begin{prooftree}
\pj{\Gamma}{M}{\STA}
\GAP
\pj{\Gamma, x \co \STA}{N}{\STB} 
\justifies
\pj{\Gamma}{\slet{x}{M}{N}}{\STB}
\end{prooftree}
\]
The induction hypothesis gives us two morphisms in $\CCat$
\[
\den\Gamma\xrightarrow{\den M}\den \sigma
\qquad\qquad
\den {\Gamma}\times \den\sigma\xrightarrow{\den N} \den \tau
\]
and we use these to define a morphism in $\CCat$ that interprets
$(\slet x M N)$:
\[
\den \Gamma\xrightarrow {J(\text{diag})}
\den\Gamma\times\den\Gamma
\xrightarrow = 
\ltensor{\den\Gamma}{\den \Gamma}
\xrightarrow{\ltensor{\den\Gamma}{\den M}}
\ltensor{\den \Gamma}{\den \sigma}
\xrightarrow{=}
{\den \Gamma\times\den \sigma}
\xrightarrow{\den N}
\den \tau \, .
\]
As another example $\den{\return V} = J(\den{V})$.

\end{itemize}
This defines a sound and complete notion of model for 
{\FGCBV}.
\begin{prop}[\cite{Levy:03}, Prop.~5.1]
The interpretation of the 
fine-grain call-by-value calculus in a Kleisli model is sound:
\begin{enumerate}
\item If $\veq \Gamma V W \sigma$ then 
$\den V=\den W:\den \Gamma\to\den \sigma$
in $\VCat$.
\item If $\peq \Gamma M N \sigma$ then 
$\den M=\den N:\den \Gamma\to\den \sigma$
in $\CCat$.
\end{enumerate}
\end{prop}

\subsection{Relationship with monads}
\label{sec:monads}
We now explain the connection between enriched Kleisli models 
and Kleisli categories for a monad.
For more detail, see the paper by Levy et al.~\cite{Levy:03}.

From the syntactic point of view, 
the fine-grain call-by-value language 
can be thought of as a variant of Moggi's $\lambda_C$~\cite{Moggi:89}:
the type construction $(1\pto(-))$ is a
monad.  

From the semantic point of view, 
recall that a $\lambda_C$-model~\cite{Moggi:89}
is a category with finite products and a strong monad 
with Kleisli exponentials. We now explain these conditions.

Let $\VCat$ be a category with finite products,
and let $(T,\eta,\mu)$ be a monad on $\VCat$.
Let $\CCat$ be the Kleisli category for $T$: the objects of $\CCat$ are 
the objects of $\VCat$ and a morphism $A\to B$ in $\CCat$ is a morphism
$A\to T(B)$ in $\VCat$. 
There is an identity-on-objects functor
$J\colon \VCat\to \CCat$ 
which takes a morphism $f\colon A\to B$ in $\VCat$
to the morphism $(\eta_B\cdot f):A\to B$ in $\CCat$.

A strength for a monad $T$ is usually expressed 
as a family of morphisms $A\times T(B)\to T(A\times B)$
that respect the structure of the monad.
In fact, a monad is strong if and only if 
there is an action of $\VCat$ on $\CCat$ and the
identity-on-objects functor $J\colon \VCat\to\CCat$ 
preserves it.
The strength is needed to 
take a morphism $f\colon B\to B'$ in $\CCat$
to  a morphism $(\ltensor A f):\ltensor A B\to \ltensor A{B'}$ 
in $\CCat$.

The requirement of Kleisli exponentials is normally 
expressed as the requirement that for all $A$ and $B$,
the hom-functor 
$\Homset{\VCat}{(-)\times A}{TB}:\opcat{\VCat}\to\Set$ is representable.
To endow~$\VCat$ with Kleisli exponentials is to 
give a right adjoint for the action,
i.e. an enrichment of $\CCat$ in $\VCat$.

Conversely, every enriched Kleisli model $(\VCat,\CCat,J)$ 
induces a strong monad on $\VCat$ with Kleisli exponentials.
The monad is defined using the closed structure:
$T(A)\defeq \CCat(1,A)$.
The Kleisli category for this monad is isomorphic to $\CCat$.
On the other hand, if we begin with a monad, build the Kleisli category
and then take the monad $\CCat(1,A)$, we recover a monad that is isomorphic
to the one that we started with. 
In this sense, enriched Kleisli models and $\lambda_C$-models are equivalent.
Note that they are not exactly the same,
for although the hom-set $\Homset{\VCat}{A}{\CCat(1,B)}$ is in bijection with
$\Homset{\CCat}{A}{B}$, the sets are typically not identical.

\subsection{The syntactic Kleisli model}
\label{sec:syn:Kl:model}

The types and terms of the fine-grain call-by-value calculus form 
a syntactic model.
We employ the same kinds of technique as 
for enriched call-by-value in Section~\ref{sec:syn:enr:model}.
\begin{itemize}
\item The objects of both $\VCat$ and $\CCat$ are the types of FGCBV.
\item A morphism $\STA\to \STB$ in $\VCat$ is a value judgement
$\vj{x\colon \STA}{V}\STB$ modulo the equational theory $\equiv$
(Figure~\ref{fig:lambdac:typing}) and modulo renaming the free variable~$x$.
Identities are variables and composition is by substitution. 
\item 
A morphism $\STA\to \STB$ in $\CCat$ 
is a computation judgement $\pj{x\colon \STA}M \STB$ modulo the equational theory
$\equiv$ and renaming the free variable~$x$.
The identity $\STA\to \STA$ is $\pj{x\colon \STA} {\return x}\STA$.
Composition is \emph{not} by substitution,
since one cannot substitute a producer term for a variable.
Rather, the composite of 
\[
\STA\xrightarrow {\pj{x\colon \STA}{M}\STB} \STB
\xrightarrow {\pj{y\colon \STB} N \STC} \STC
\]
in $\CCat$ is 
$\pj{x\colon \STA}{\slet {y} M N} \STC$.
\item The product structure in $\VCat$ is given by the product
types, projections and pairing.
\item The action of $\VCat$ on $\CCat$ is given on objects 
by the binary product types: let $\ltensor \STA \STB\defeq \STA \times \STB$.
On morphisms, given $A\xrightarrow {\vj{x\colon \STA}V {\STA'}} {\STA'}$ 
in $\VCat$ 
and 
$\STB\xrightarrow {\pj{y\colon \STB}M{\STB'}}{\STB'}$
in $\CCat$,
we define 
\[(\ltensor \STA \STB\xrightarrow{\ltensor V M} \ltensor {\STA'} {\STB'} )
\ \defeq \ 
\pj {z\colon \STA\times \STB}
{\slet {y'} {\sub M {\snd z}y} {\return {\pair {\sub V{\fst z}x} {y'} }}}
{\STA'\times \STB'} \]
\item The enrichment is given by
$\CCat(\STA,\STB)\defeq (\STA\pto \STB)$.
\item The identity-on-objects functor
$J\colon \VCat\to \CCat$ takes a morphism
${\STA\xrightarrow {\vj{x\colon \STA} V \STB} \STB}$ 
in $\VCat$
to the morphism
${\STA\xrightarrow {\pj{x\colon \STA} {\return V} \STB} \STB}$ 
in $\CCat$.
\end{itemize}

\noindent We write $\SynKlModel$ for the syntactic Kleisli model.

\subsection{Universal property of the syntactic Kleisli model}
\label{sec:syn:Kl:biinitial}

Appendix~\ref{app:Kleisli:cat} defines the 2-category $\Freyd$ of Kleisli models. As was the case for $\CATECBV$, 
1-cells are only required to preserve structure up to isomorphism. 

\begin{thm} \label{thm:fgcbv:biinitial}
The syntactic Kleisli model $\SynKlModel$ is bi-initial in $\Freyd$. The unique 
morphisms with domain $\SynKlModel$ are given by interpretation of syntax into models.
\end{thm}

%%% Local Variables: 
%%% mode: latex
%%% TeX-master: "mogelberg-staton"
%%% End: 

% !TEX root = Mogelberg-staton.tex

\section{The linear-use state-passing translation}
\label{sec:translation}

This section defines the linear-use state-passing translation
from the fine-grain call-by-value calculus to 
the enriched call-by-value calculus, 
and states the main syntactic results of this paper:
fullness on types and full completeness. Together these assert that 
the linear-use state-passing translation is an equivalence of languages.

We now fix a computation type $\EECarbstate$ of \ECBV.
For now, it can be an arbitrary computation type; later
we will make it a distinguished basic type to achieve a 
full completeness result.
We will describe a translation 
from {\FGCBV} to {\ECBV}. 
When $\EECarbstate$ is thought of as a type of states,
then this translation reads as a state-passing translation.

\renewcommand{\CBVtoEEC}{\CBVtoEECbase\EECarbstate}
\newcommand{\CBVtoEECV}{\CBVtoEECbase \EECarbstate}
\newcommand{\CBVtoEECP}[1]{\CBVtoEECbase \EECarbstate{#1}_\svar}
We translate \FGCBV\ types $\sigma$ to \ECBV\ value types
$\CBVtoEEC\sigma$:
\begin{align*}
\CBVtoEEC{\alpha} & \, \defeq \VconstA &
\CBVtoEEC{(\STA \times \STB)} & \, \defeq\, 
\CBVtoEEC{\STA} \prodtype \CBVtoEEC{\STB} &
\CBVtoEEC{1} & \, \defeq \algone&
\CBVtoEEC{(\STA \pto \STB)} & \, \defeq \, \ltensortype{(\CBVtoEEC{\STA})}{\EECarbstate} \lfun \, \ltensortype{(\CBVtoEEC{\STB})}{\EECarbstate}
\end{align*}
We extend this translation to type contexts,
taking an FGCBV type context $\Gamma$ to an
ECBV type context $\CBVtoEEC\Gamma$.

The translation on terms is syntax-directed.
We pick a variable $s$, completely fresh.
The translation takes an FGCBV value type judgement
${\vj{\Gamma}{V}{\STA}}$  
to an \ECBV\
judgement ${\tj{\CBVtoEEC{\Gamma}}{\CBVtoEECV{V}}{\CBVtoEEC{\STA}}}$,
and an FGCBV producer judgement $\pj{\Gamma}{M}{\STA}$ 
to an \ECBV\ judgement ${\aj{\CBVtoEEC{\Gamma}}{s \co \EECarbstate}{\CBVtoEECP{M}}{\ltensortype{(\CBVtoEEC{\STA})}{\EECarbstate}}}$. The translation is defined as follows.
\begin{gather*}
\CBVtoEECV{x} \defeq x 
\ \quad\ \CBVtoEECV{\star} \defeq \star 
\ \quad \ 
\CBVtoEECV{\pair{V}{W}}\defeq \pair{\CBVtoEECV{V}\!}{\CBVtoEECV{W}} 
\ \quad\ 
\CBVtoEECV{(\prj 1{V})}\defeq\prj 1{\CBVtoEECV{V}} 
\ \quad\ 
\CBVtoEECV{(\prj 2{V})}\defeq\prj 2{\CBVtoEECV{V}} 
\\
\begin{aligned}
\CBVtoEECP{(\return{V})} & \, \defeq \ltensorterm{(\CBVtoEECV{V})}{s} 
&
\CBVtoEECP{(\slet{x}{M}{N})} & \, \defeq \letdot{x}{s}{\CBVtoEECP{M}}{\CBVtoEECP{N}} 
\\
\CBVtoEECV{(\lam{x}{\STA}{N})} & \, \defeq \llam{z}{\ltensortype{\CBVtoEEC{\STA}}{\EECarbstate}}{\letdot{x}{s}{z}{\CBVtoEECP{N}}} &
\CBVtoEECP{(V \, W)} & \, \defeq \lappl{\CBVtoEECV{V}\,}{\ltensorterm{(\CBVtoEECV{W})}{s}} 
\end{aligned}
\end{gather*}
In the case for $\lambda$-abstraction, the $z$ is chosen to be completely fresh.

The translation respects types. For instance
\[\begin{prooftree}
{\pj{\Gamma,x\co\STA}{N}{\STB}}
\justifies
{\vj{\Gamma}{\lam{x}{\STA}N}{\STA\pto \STB}}
\end{prooftree}
\quad\text{becomes}\quad
\hide{
\begin{prooftree}
{\aj{\CBVtoEEC\Gamma,x\co \CBVtoEEC\STA}{s\co\EECarbstate}{\CBVtoEECP N}{\CBVtoEEC\STB}}
\justifies
{\tj{\CBVtoEEC\Gamma}{\llam {z}{\ltensortype{\CBVtoEEC{\STA}}{\EECarbstate}}{\letdot{x}{s}{z}{\CBVtoEECP{N}}}}
{{\ltensortype{\CBVtoEEC\STA}\EECarbstate}\lfun{\ltensortype{\CBVtoEEC\STB}\EECarbstate}}}
\end{prooftree}}
\begin{prooftree}
\[
\[
\justifies
{\aj{\CBVtoEEC\Gamma}{z\co \ltensortype {\CBVtoEEC\STA}\EECarbstate}
{z}{\ltensortype {\CBVtoEEC\STA}\EECarbstate}}\]
\quad
{\aj{\CBVtoEEC\Gamma,x\co \CBVtoEEC \STA}{s\co \EECarbstate}
{\CBVtoEECP N}{\ltensortype {\CBVtoEEC\STB}\EECarbstate}}
\justifies
{\aj{\CBVtoEEC\Gamma}{z\co \ltensortype {\CBVtoEEC\STA}\EECarbstate}
{\letdot x s z{\CBVtoEECP N}}{\ltensortype {\CBVtoEEC\STB}\EECarbstate}}
\]
\justifies
{\tj{\CBVtoEEC\Gamma}{\llam {z}{\ltensortype{\CBVtoEEC{\STA}}{\EECarbstate}}{\letdot{x}{s}{z}{\CBVtoEECP{N}}}}
{{\ltensortype{\CBVtoEEC\STA}\EECarbstate}\lfun{\ltensortype{\CBVtoEEC\STB}\EECarbstate}}}
\end{prooftree}
\]

\begin{thm}The linear-use state-passing translation is sound: \label{thm:soundness}
\begin{enumerate}
\item
If $\veq\Gamma VW\sigma$ then 
$\teq{\CBVtoEEC \Gamma}{\CBVtoEECV{V}}{\CBVtoEECV{W}}{\CBVtoEEC \sigma}$.
\item 
If $\peq\Gamma MN\sigma$ then 
$\aeq{\CBVtoEEC \Gamma}{s\co \EECarbstate}{\CBVtoEECP{M}}{\CBVtoEECP{N}}{\CBVtoEEC \sigma}$.
\end{enumerate}
\end{thm}
This result can be proved by induction on the structure of equality
($\equiv$) derivations, but it can also be derived semantically as 
we shall see in Section~\ref{sec:sp:transl:semantics}.

\subsection{Full completeness}
We now state our main theorems: fullness on types and full completeness
on terms.
To state fullness on types we need to talk about isomorphism of types in {\ECBV}. 
This can be defined in the usual way: for
value types, an isomorphism $\VA \iso \VB$ is given by two judgements,
$\tj{x \co \VA}{t}{\VB}$ and $\tj{y \co \VB}{u}{\VA}$, such that
$u[t/y] \equiv x$ and $t[u/x] \equiv y$. For computation types, $\CA \iso \CB$ is
witnessed by closed terms of type $\CA \lfun \CB$, $\CB \lfun \CA$
composing in both directions to identities. 
We note the following type isomorphisms, inherited from 
the enriched effect calculus~\cite[\S 3]{Mogelberg:CSL:09}:
\begin{align}
\CA & \, \iso\, \ltensortype{\algone}{\CA}  \label{eq:tensor:one}
&
\ltensortype{\VA}{(\ltensortype{\VB}{\CC})} & \, \iso\, \ltensortype{(\VA \times \VB)}{\CC} 
\end{align}

\begin{thm}[Fullness on types]
\label{thm:full-on-types}
Let $\VA$ be a value type of {\ECBV} formed using no other computation type constants than $\EECstate$. Then there exists an {\FGCBV} type $\STA$ such that $\CBVtoEEC{\STA} \iso \VA$.
\end{thm}

\begin{proof}
By induction on the structure of types. The interesting case $\CA \lfun \CB$ uses the fact that any computation type not using any $\CconstA$ other than $\EECstate$ is isomorphic to one of the form $\ltensortype{\VC}{\EECstate}$, which follows from the isomorphisms~(\ref{eq:tensor:one}). 
\end{proof}
\noindent
We now state our main syntactic result.
\renewcommand{\CBVtoEEC}{\CBVtoEECbase\EECstate}
\renewcommand{\CBVtoEECV}{\CBVtoEECbase \EECstate}
\renewcommand{\CBVtoEECP}[1]{\CBVtoEECbase \EECstate{#1}_\svar}
\begin{thm}[Full completeness] \label{thm:full:faithful}\mbox{}
\begin{enumerate}
\item \label{item:faithful:values} 
If $\vj{\Gamma}{V,W}{\STA}$ 
and $\teq{\CBVtoEEC\Gamma}{\CBVtoEECV{V}}{\CBVtoEECV{W}}{\CBVtoEEC\STA}$ 
then $\veq \Gamma V W\STA$.
\item
If  $\pj{\Gamma}{M,N}{\STA}$ and 
$\aeq{\CBVtoEEC\Gamma}{\svar\co \EECstate}{\CBVtoEECP{M}}{\CBVtoEECP{N}}{\ltensortype{\CBVtoEEC\STA}\EECstate}$ then 
$\peq{\Gamma}MN{\STA}$.
\item For any $\tj{\CBVtoEEC{\Gamma}}{t}{\CBVtoEEC{\STA}}$ there is a term $\vj{\Gamma}{V}{\STA}$ such that $\teq{\CBVtoEEC\Gamma}t{ \CBVtoEEC{V}}{\CBVtoEEC\STA}$. 
\item For any $\aj{\CBVtoEEC{\Gamma}}{s \co \EECstate}{t}{\ltensortype{(\CBVtoEEC{\STA})}{\EECstate}}$ there is a producer term $\pj{\Gamma}{M}{\STA}$ such that 
$\aeq{\CBVtoEEC\Gamma}{s\co\EECstate}t {\CBVtoEEC{M}}{\ltensortype{(\CBVtoEEC{\STA})}{\EECstate}}$.
\end{enumerate}
\end{thm}
%
%Theorem~\ref{thm:full:faithful} can be proved syntactically as follows. Consider first the fragment of {\ECBV} with no other computation type constants than $\EECstate$, and only the value type constants of $\CBV E$. This fragment is equivalent to a variant of {\ECBV} where the only computation types are the ones of the form $\ltensortype{\VA}{\EECstate}$ with corresponding variants of the typing rules for $\ltensortype{\VA}{\CB}$.
%The translation $\CBVtoEEC{(-)}$ gives a bijection from $\CBV E$ types to value types of $\ECBVS{E}$, and one can define an inverse to this translation. Further type constants can be added to $\ECBVS{E}$ without changing the result; 
%this can be proved via a normalization theorem for $\ECBVS{E}$ 
%which follows the one for EEC (to appear in~\cite{EEC:journal}).
%
\noindent In Section~\ref{sec:sp:transl:semantics} we sketch a semantic proof of Theorems~\ref{thm:full-on-types} and~\ref{thm:full:faithful}.

%%% Local Variables: 
%%% mode: latex
%%% TeX-master: "mogelberg-staton"
%%% End: 

% !TEX root = Mogelberg-staton.tex

\section{A semantic proof of full completeness}
\label{sec:relating:models}

In this section we present two constructions on models. The first 
(\S\ref{sec:enrtokleisli})
constructs a Kleisli model (Def.~\ref{def:monadmodel})
from an enriched model (Def.~\ref{def:enrichedmodel})
with a specified
state object.  The second (\S\ref{sec:kleisli-to-enriched-models})
constructs an enriched model with a state
object from a given Kleisli model. The state-passing translation
arises from the first of these constructions. These two constructions
form a bi-adjunction exhibiting the category of Kleisli models as a
coreflective subcategory of the category of enriched models with chosen
state objects (\S\ref{sec:adj}).  
In Section~\ref{sec:sp:transl:semantics} we shall see
how to use these facts to explain full completeness of the 
linear-use state-passing
translation
(Theorem~\ref{thm:full:faithful}).

\subsection{From enriched models with state to Kleisli models}\label{sec:enrtokleisli}

Given an {\enrmodel} $(\VCat, \CCat)$ with a state object $\stateobj$
in $\CCat$, we can form an {\Klmodel} $\ECBVToFreyd(\VCat, \CCat,
\stateobj) \eqdef (\VCat, \KlCat{\VCat}{\CCat}{\stateobj},
J_{\stateobj})$, where the category $\KlCat{\VCat}{\CCat}{\stateobj}$
has the same objects as $\VCat$ and hom-sets
\[
\Homset{\KlCat{\VCat}{\CCat}{\stateobj}}{\SA}{\SB} \defeq 
\Homset{\CCat}{\ltensor{\SA}{\stateobj}}{\ltensor{\SB}{\stateobj}}
\]
Composition in $\KlCat{\VCat}{\CCat}{\stateobj}$ is just composition as in $\CCat$. (This is an isomorphic presentation of the Kleisli category for the monad $\CHom{\stateobj}{\ltensor{-}{\stateobj}}$ on $\VCat$.) The functor $J_{\stateobj}$ is the identity on objects and maps $f\co \SA \to \SB$ to  $\ltensor{f}{\stateobj}$.

\begin{lem} \label{lem:Kl:well-def}
For any {\enrmodel} with state $(\VCat, \CCat, \stateobj)$ the data $(\VCat, \KlCat{\VCat}{\CCat}{\stateobj}, J_{\stateobj})$ defines an {\Klmodel}.
\end{lem}

\begin{proof}
The action $\Klltensor{(-_1)}{(-_2)} \co \VCat \times \KlCat{\VCat}{\CCat}{\stateobj} \to \KlCat{\VCat}{\CCat}{\stateobj}$ is defined on objects as $\Klltensor{\SA}{\SB} = \SA \times \SB$. On morphisms it maps $f\co \SA \to \SA', g\co \ltensor{\SB}{\stateobj} \to \ltensor{\SB'}{\stateobj}$ to the composite
\[
\ltensor{(\SA \times \SB)}{\stateobj} \xrightarrow\iso \ltensor{\SA}{(\ltensor{\SB}{\stateobj})} \xrightarrow{\ltensor{f}g} 
\ltensor{\SA'}{(\ltensor{\SB'}{\stateobj})} \xrightarrow\iso \ltensor{(\SA' \times \SB')}{\stateobj}
\]
which is easily seen to be functorial. 

The right adjoint to $\Klltensor{(-)}{\SA}$ is $\KlHom{\stateobj}{\SA}{-} \eqdef \CHom{\ltensor{\SA}{\stateobj}}{\ltensor{(-)}{\stateobj}}$. \end{proof}

The construction $\ECBVToFreyd$ described above extends to a 2-functor $\ECBVToFreyd \co \CATECBV \to \Freyd$
from the 2-category of enriched models to the 2-category of 
Kleisli models.
See Appendix~\ref{app:ECBVToFreyd} for details.

\subsection{From Kleisli models to enriched models with state}
\label{sec:kleisli-to-enriched-models}
%To go the other way is almost trivial since a 
Any Kleisli model is trivially an enriched model, so for the opposite
construction we just need to pick a state object in a Kleisli
model. We define $\FreydToECBV(\VCat, \CCat, J) \eqdef (\VCat, \CCat,
1)$, where $1$ is the terminal object of $\VCat$ considered 
as an object of $\CCat$.
This definition
extends to a 2-functor $\FreydToECBV \co \Freyd \to \CATECBV$,
as shown in Appendix~\ref{app:adjunction}.

The motivation for this definition is that, as we now show, the 2-category $\CATECBV$ can be seen as a 
2-category of enriched adjunctions, and the 2-functor $\FreydToECBV$ can be seen as an inclusion of
Kleisli adjunctions into $\CATECBV$. 

Let $(\VCat, \CCat)$ be an enriched model. By an \emph{enriched adjunction} we mean 
an adjunction $\lradj{F}{U} \co \CCat \to \VCat$ equipped with a natural 
coherent isomorphism
$F(\SA\times\SB)\iso\ltensor{\SA}{F(\SB)}$.
When $\VCat$ is cartesian closed, this is equivalent to
the usual definition, i.e. a natural isomorphism
${\CCat(F(-_1),-_2)\iso \VCat(-_1,U(-_2))}$ in $\VCat$ (see e.g.~\cite{kelly-enrichedadj}). 

Any choice of state object gives an enriched adjunction,
since
${(\ltensor -\stateobj)}$ 
is left adjoint to  ${\CHom\stateobj-\colon \CCat\to\VCat}$.
The following proposition (first noted for $\EEC$, \cite[Proof of Thm.~4]{Mogelberg:CSL:09}, \cite{EEC:journal}) shows that every enriched adjunction arises in this way:
\begin{prop}[\cite{EEC:journal}]
\label{prop:adjmodels}
Let $(\VCat,\CCat)$ be an enriched model.
If ${F\dashv U\colon \CCat\to\VCat}$ is an enriched adjunction
then it is naturally isomorphic to
the enriched adjunction induced by $F(1)$.
\end{prop}
So enriched adjunctions correspond essentially bijectively 
%SS: not bijectively since only up to natural iso
to state objects.
In particular the state object corresponding to a Kleisli model is $1$. 
Monads induced by state objects can be described in {\ECBV} as $\stateobj \lpop \ltensortype{(-)}{\stateobj}$.
By the correspondence between Kleisli models and strong monads we arrive at the slogan: 
\emph{Every strong monad is a linear-use state monad.} 
More directly, the slogan can be derived from the 
isomorphism $\KlCat{}{} T(1,\SA\times 1)\cong T(\SA)$, which holds for 
the Kleisli category $\KlCat{}{}T$ of any strong monad~$T$. 

(\emph{Remark:} if $T$ is a strong monad on $\VCat$
then, for any object~$\stateobj$ of $\KlCat{}{}T$, 
the linear-use state monad 
$\KlCat{}{}T(\stateobj,\ltensor{(-)}{\stateobj})$ on $\VCat$ is also known as the
application of the state monad transformer to~$T$, as in~\cite{lhj-monad-transformers}.)

%% So we can equivalently consider $\ECBVstate$ models as enriched adjunctions. 
%The construction $\FreydToECBV(\VCat, \CCat, J) \eqdef (\VCat, \CCat, 1)$ extends to a 2-functor 
%$\FreydToECBV \co \Freyd \to \CATECBV$. (See Appendix~\ref{app:adjunction} for details.)

\subsection{A coreflection}
\label{sec:adj}

The constructions $\ECBVToFreyd$ and $\FreydToECBV$ form a bi-adjunction
between the 2-categories of Kleisli models and of enriched models
with state. One intuition for this is that it expresses the minimality
property of Kleisli resolutions of monads.

\begin{thm} \label{thm:adj} The 2-functor $\FreydToECBV \co \Freyd
  \to \CATECBV$ is left biadjoint to $\ECBVToFreyd$, i.e., for any
  pair of objects $(\VCat, \CCat, J)$ and $(\VCat', \CCat',
  \stateobj)$ of $\Freyd$ and $\CATECBV$ respectively, there is an
  equivalence of categories
  \[\CATECBV(\FreydToECBV(\VCat, \CCat, J),(\VCat', \CCat',
  \stateobj))\ \simeq\ \Freyd((\VCat, \CCat, J),\ECBVToFreyd(\VCat',
  \CCat', \stateobj))
\]
natural in $(\VCat, \CCat, J)$  and $(\VCat', \CCat', \stateobj)$. 
Moreover, the unit of the adjunction $\eta \co \id_{\CATECBV} \to \ECBVToFreyd\circ\FreydToECBV$ is an isomorphism.
\end{thm}
See Appendix~\ref{app:adjunction} for a proof of Theorem~\ref{thm:adj}. Since left bi-adjoints preserve 
bi-initial objects we get the following
connection between the syntactic enriched model
$\SynEnrichedModel$ and the syntactic Kleisli model $\SynKlModel$.
\begin{cor} \label{cor:eq:syn:models}
$\SynEnrichedModel$ and $\FreydToECBV\SynKlModel$ are equivalent as objects of $\CATECBV$. 
\end{cor}

%Explicitly the unit of the adjunction is constructed as follows. 
Since 
$\ECBVToFreyd(\FreydToECBV(\VCat, \CCat, J)) = (\VCat, \KlCat{\VCat}{\CCat}{1}, J_1)$ the unit of the adjunction 
can be described as the pair $(\id_{\VCat}, G \co \CCat \to \KlCat{\VCat}{\CCat}{1})$ where $G(f \co \SA \to \SB)$ is the composite
\[
\ltensor{\SA}{1} \xrightarrow\iso \SA \xrightarrow f \SB  \xrightarrow\iso \ltensor{\SB}{1} 
\]
using the isomorphism $J(\pi) \co \ltensor{\SA}{1} = \SA \times 1 \to \SA$.
The pair $(\id_{\VCat}, G)$ preserves the enrichment only up to isomorphism, 
and  
this is our main motivation for using 2-categories of models
(see also the discussion in Section~\ref{sec:enr:model:biinitial}).

\subsection{A semantic explanation of the state-passing translation}
\label{sec:sp:transl:semantics}

The linear-use state-passing translation is essentially the
interpretation of fine-grain call-by-value into the model obtained
by applying the construction $\ECBVToFreyd$ of
Section~\ref{sec:relating:models} to the syntactic {\enrmodel} of
Section~\ref{sec:syn:enr:model}. In this model judgements
$\vj{\Gamma}{V}{\STA}$ and $\pj{\Gamma}{M}{\STA}$ are interpreted as
judgements of the form
\begin{align*}
\tj{x \co \Prod\CBVtoEEC{\Gamma}}{\den V}{\CBVtoEEC{\STA}} && \aj{-}{x \co \ltensortype{(\Prod\CBVtoEEC{\Gamma})}{\states}}{\den M}{\CBVtoEEC{\STA}}
\end{align*}
respectively, where $\Prod\CBVtoEEC{\Gamma}$ is the product of all the types appearing in $\CBVtoEEC{\Gamma}$.

\begin{lem} \label{lem:sp:transl:sem}
Let $\SynEnrichedModel$ be the syntactic {\Klmodel} of Section~\ref{sec:syn:enr:model}. 
The interpretation of FGCBV into $\ECBVToFreyd\SynEnrichedModel$ models a type $\STA$ as $\CBVtoEEC\STA$. Let $\Gamma = x_1 \co \STA_1 \dots x_n \co \STA_n$ be a context of FGCBV and let $\vj{\Gamma}{V}{\STB}$ and $\pj{\Gamma}{M}{\STB}$ be judgements of FGCBV. Then $V$ and $M$ are interpreted as the equivalence classes of the terms
\begin{align*}
&\tj{x \co (\Prod\CBVtoEEC{\Gamma})}{\sub{\CBVtoEEC V}{\pi_1x\dots \pi_nx}{x_1\dots x_n}}{\CBVtoEEC{\STB}} \\ 
&\aj{-}{x \co \ltensortype{(\Prod\CBVtoEEC{\Gamma})}{\states}}{\letdot{z}{s}{x}{(\sub{\CBVtoEEC M}{\pi_1z\dots \pi_nz}{x_1\dots x_n})}}{\CBVtoEEC{\STB}}
\end{align*}
\end{lem}

Soundness of the state-passing translation
(Theorem~\ref{thm:soundness}) follows immediately 
from Lemma~\ref{lem:sp:transl:sem}.
Fullness on types and 
full completeness (Theorems~\ref{thm:full-on-types} and \ref{thm:full:faithful})
are also  consequences.

\begin{proofof}{Theorems~\ref{thm:full-on-types} and \ref{thm:full:faithful}}
By Theorem~\ref{thm:fgcbv:biinitial} and Lemma~\ref{lem:sp:transl:sem} the state-passing translation is the unique (up to isomorphism) 1-cell
\[(F,G) \co \SynKlModel \to \ECBVToFreyd{\SynEnrichedModel}\]
in $\Freyd$. It suffices to show that this is an equivalence in $\Freyd$, because this implies that $F$ and $G$ are both equivalences of categories, in particular they are essentially full on objects (proving Theorem~\ref{thm:full-on-types}) and full and faithful (proving Theorem~\ref{thm:full:faithful}). 

By initiality $(F,G)$ must be isomorphic to the composite
\[\SynKlModel \xrightarrow{\eta} \ECBVToFreyd(\FreydToECBV\SynKlModel) \xrightarrow{\simeq} \ECBVToFreyd\SynEnrichedModel\]
of the unit of the adjunction (which is an isomorphism by Theorem~\ref{thm:adj}) and $\ECBVToFreyd$ applied to the equivalence of Corollary~\ref{cor:eq:syn:models}. Since this composite is an equivalence, so is $(F,G)$. 
\end{proofof}

\section{Sums}\label{sec:sums}
It is routine to add sum types to the languages considered in 
Section~\ref{sec:ecbv} and \ref{sec:fgcbv},
and the state-passing translation extends straightforwardly.
We now summarize the details.
\subsection{Sums in the enriched call-by-value calculus}
\label{sec:sums:enriched}
We add sums to the enriched call-by-value calculus,
for both value and computation types.
The required modifications to the calculus are given in 
Figure~\ref{figure:effects:sumtyping}. The resulting calculus
is still a fragment of the enriched effect calculus~\cite{EEC:journal}.
\begin{figure*}[tb]
\framebox{
\begin{minipage}{.96\linewidth}
%\footnotesize
\emph{Types.}
\begin{align*}
\VA,\VB \, & ::= \, \VconstA \,\mid \,\algone \,\mid \,\VA \prodtype \VB \, \,\mid \, \CA \lpop \CB  \mid\, 0\,\mid\, \VA + \VB\\
\CA,\CB \, & ::= \, \CconstA  \,\mid\,\ltensortype{\VA}{\CB} \,\mid  \,   \algzero\,\mid\,\CA \algplus \CB\enspace .
\end{align*}
\begin{center}
\line(1,0){350}\gnl
\end{center}
\emph{Term formation.}
The following rules are in addition 
to the rules in Figure~\ref{figure:effects:typing}.
\begin{center}
\begin{gather*}
\prooftree
\tj{\Gamma}{t}{0}
\justifies
\tj{\Gamma}{\vimage{t}}{\VA}
\endprooftree
\GAP\GAP
\prooftree
\aj{\Gamma}{\Delta}{t}{\algzero}
\justifies
\aj{\Gamma}{\Delta}{\algimage{t}}{\CA}
\endprooftree
\gnl
\prooftree
\tj{\Gamma}{t}{\VA_i}
\justifies
\tj{\Gamma}{\vin{i}{t}}{\VA_1 + \VA_2}
\using{(i=1,2)}
\endprooftree
\GAP
\prooftree
\aj{\Gamma}{\Delta}{t}{\CA_i}
\justifies
\aj{\Gamma}{\Delta}{\algin{i}{t}}{\CA_1 \algplus \CA_2}
\using{(i=1,2)}
\endprooftree
\gnl
\prooftree
\tj{\Gamma}{t}{\VA_1 + \VA_2}
\GAP
\tj{
\Gamma, \, \In{x_1}{\VA_1}}{u_1}{\VC}
\GAP
\tj{\Gamma, \, \In{x_2}{\VA_2}}{u_2}{\VC}
\justifies
\tj{\Gamma}{\vcase{t}{x_1}{u_1}{x_2}{u_2}}{\VC}
\endprooftree
\gnl 
\prooftree
\aj{\Gamma}{\Delta}{s}{\CA_1 \algplus \CA_2}
\GAP
\aj{\Gamma}{\In{x_1}{\CA_1}}{t_1}{\CC}
\GAP
\aj{\Gamma}{\In{x_2}{\CA_2}}{t_2}{\CC}
\justifies
\aj{\Gamma}{\Delta}{\algcase{s}{x_1}{t_1}{x_2}{t_2}}{\CC}
\endprooftree
\end{gather*}
\end{center}
\begin{center}\mbox{}\\
\line(1,0){350}\gnl
\end{center}
\emph{Equality.} 
The following rules are in addition 
to the rules in Figure~\ref{figure:effects:typing}.
\begin{center}
\begin{gather*}
\prooftree
\tj{\Gamma}{t}{0}
\GAP
\tj{\Gamma, x\co 0}{u}{\VA}
\justifies
\teq{\Gamma}{\vimage t}{\sub utx}{\VA}
\endprooftree
\GAP\GAP
\prooftree
\aj{\Gamma}\Delta{t}{\algzero}
\GAP
\aj{\Gamma}{x\co \algzero}{u}{\CA}
\justifies
\aeq{\Gamma}\Delta{\algimage t}{\sub utx}{\CA}
\endprooftree
\gnl
\prooftree
\tj{\Gamma}{t}{\VA_i}
\GAP
\tj{\Gamma, \, \In{x_1}{\VA_1}}{u_1}{\VB}
\GAP
\tj{\Gamma, \, \In{x_2}{\VA_2}}{u_2}{\VB}
\justifies
\teq{\Gamma}{\vcase{(\vin i t)}{x_1}{u_1}{x_2}{u_2}}{u_i[t/{x_i}]}{\VB}
\using
{(i=1,2)}
\endprooftree
\gnl
\prooftree
\aj{\Gamma}\Delta{t}{\CA_i}
\GAP
\aj{\Gamma}{\In{x_1}{\CA_1}}{u_1}{\CB}
\GAP
\aj{\Gamma}{\In{x_2}{\CA_2}}{u_2}{\CB}
\justifies
\aeq{\Gamma}\Delta{\algcase{(\algin i t)}{x_1}{u_1}{x_2}{u_2}}{u_i[t/{x_i}]}{\CB}
\using
{(i=1,2)}
\endprooftree
\gnl
\prooftree
\tj{\Gamma}{t}{\VA_1 + \VA_2}
\GAP
\tj{\Gamma, z \co \VA_1 + \VA_2}{u}{\VB}
\justifies
\teq{\Gamma}{\sub utz}{\vcase{t}{x_1}{\sub u{\vin1{x_1}}{z}}{x_2}{\sub u{\vin2{x_2}}{z}}}{\VB}
\endprooftree
\gnl
\prooftree
\aj{\Gamma}{\Delta}{t}{\CA_1 \algplus \CA_2}
\GAP
\aj{\Gamma}{z\co\CA_1 \algplus \CA_2}{u}{\CB}
\justifies
\aeq{\Gamma}\Delta{\sub utz}{\algcase{t}{x_1}{\sub u{\algin1{x_1}}{z}}{x_2}{\sub u{\algin2{x_2}}{z}}}{\CB}
\endprooftree
\end{gather*}
\mbox{}\\\end{center}
\end{minipage}
}
\caption{Additional rules for sum types in Enriched Call-by-Value
}
\label{figure:effects:sumtyping}
\end{figure*}
We now extend the notion of model to accommodate the sums.
\begin{terminology}
Recall that a \emph{distributive category} 
is a category with finite products and 
coproducts such that the canonical morphisms
$0\to A \times 0$ 
and $((A\times B)+(A\times C)) \to (A\times (B+ C))$
are isomorphisms.

If a distributive category $\VCat$ 
has an action $({\ltensor{}{}})$ on a category $\CCat$ with finite coproducts
$(\algzero,\algplus)$,
then we say that the situation is \emph{distributive}
if the four canonical morphisms 
${\algzero\to \ltensor A \algzero}$,
${(\ltensor A \algB)\algplus (\ltensor A \algC)\to \ltensor A {(\algB\algplus\algC)}}$,
$\algzero\to \ltensor 0\algA$ and 
$(\ltensor A \algC)\algplus (\ltensor B \algC)\to \ltensor {(A+B)} {\algC}$
are isomorphisms.
If the action is enriched, i.e.~% 
each functor $(\ltensor -\algA):\VCat\to\CCat$ has a right adjoint,
then this definition of distributive coproducts 
amounts to the usual notion of 
enriched coproducts.
(Note that when the action is enriched then the third
and fourth canonical morphisms
are necessarily isomorphisms since left adjoints preserve 
colimits.)
\end{terminology}
\begin{defi}\label{def:distrenrichedmodel}
A \emph{distributive enriched model}
is given by a distributive category $\VCat$ 
and a category $\CCat$ enriched in $\VCat$ with copowers and 
enriched coproducts.
\end{defi}
It is straightforward to extend the sound interpretation
of the enriched call-by-value calculus in enriched models
(\S\ref{sec:adjmodels})
to a sound interpretation of 
enriched call-by-value calculus with sums in distributive enriched models.

Of the examples in Section~\ref{sec:examplesenriched},
(1)--(5) are distributive.
The syntactic model of the version of the calculus with sums is 
a distributive enriched model, and it is bi-initial for the suitable 
generalization of morphism.
\subsection{Sums in the fine-grain call-by-value calculus.}
\label{sec:sums:fgcbv}
It is equally straightforward to add sums to the fine-grain call-by-value
language. This is summarized in Figure~\ref{fig:lambdac:sumtyping}.
\begin{figure*}[tb]
\framebox{
\begin{minipage}{.96\linewidth}
%\footnotesize
\emph{Types.}
\[\sigma ::= 
\alpha \mid 1 \mid \sigma \times \sigma \mid
\sigma \pto \sigma \mid 0\mid\sigma+\sigma
\]
\begin{center}
\line(1,0){350}\gnl
\end{center}
\emph{Term formation.}
The following rules are in addition 
to the rules in Figure~\ref{fig:lambdac:typing}.
\begin{center}
\begin{gather*}
\begin{prooftree}
\vj{\Gamma}{V}{0}
\justifies
\vj{\Gamma}{\cbvimage{V}}{\STA}
\end{prooftree}
\qquad\qquad
\begin{prooftree}
\vj{\Gamma}{V}{\STA_1}
\justifies
\vj{\Gamma}{\cbvin{1}{V}}{\STA_1+\STA_2}
\end{prooftree}
\qquad\qquad
\begin{prooftree}
\vj{\Gamma}{V}{\STA_2}
\justifies
\vj{\Gamma}{\cbvin{2}{V}}{\STA_1+\STA_2}
\end{prooftree}
\gnl
\begin{prooftree}
\vj{\Gamma}{V}{\STA_1+\STA_2}
\GAP
\vj{\Gamma,x_1\co\STA_1}{W_1}{\STB}
\GAP
\vj{\Gamma,x_1\co\STA_2}{W_2}{\STB}
\justifies
\vj{\Gamma}{\cbvcase V {x_1} {W_1} {x_2} {W_2}}{\STB}
\end{prooftree}
\end{gather*}
\end{center}
\begin{center}\mbox{}\\
\line(1,0){350}\gnl
\end{center}
\emph{Equality.} 
The following rules are in addition 
to the rules in Figure~\ref{fig:lambdac:typing}.
\begin{center}
\begin{gather*}
\begin{prooftree}
\vj{\Gamma}{V}{0}
\quad
\vj{\Gamma, x \co 0}{W}\sigma
\justifies
\veq{\Gamma}{\cbvimage{V}}{\sub{W}{V}{x}}\sigma
\end{prooftree}
\gnl
\begin{prooftree}
\vj{\Gamma}{V}{\STA_i}
\GAP
\vj{\Gamma,x_1\co\STA_1}{W_1}{\STB}
\GAP
\vj{\Gamma,x_1\co\STA_2}{W_2}{\STB}
\justifies
\veq{\Gamma}{\cbvcase {\cbvin i V} {x_1} {W_1} {x_2} {W_2}}
{W_i[V/x_i]}{\STB}
\using
{(i=1,2)}
\end{prooftree}
\gnl
\begin{prooftree}
\vj{\Gamma}{V}{\STA_1+\STA_2}
\GAP
\vj{\Gamma, z \co \STA_1+\STA_2}{W}{\STB}
\justifies
\veq{\Gamma}{\sub WVz}{\cbvcase V {x_1} {\sub W{\cbvin 1 {x_1}}z} {x_2} {\sub W{\cbvin 2 {x_2}}z}}
{\STB}
\end{prooftree}
\end{gather*}
\mbox{}\\
\end{center}
\end{minipage}
}
\caption{Additional rules for sum types in Fine-Grain Call-by-Value.
}
\label{fig:lambdac:sumtyping}
\end{figure*}

We only include constructors/destructors as 
value terms, but from these we can derive constructors/destructors 
for producer terms, as follows.
\begin{align*}
&
\cbvimagep M \ \defeq\ \slet x M {\return {\cbvimage x}}
\qquad\qquad
\cbvinp i M \ \defeq \ \slet x M {\return {\cbvin i x}}
\\[3pt]
&\cbvcasep M {x_1} {N_1} {x_2} {N_2} 
\\&
\hspace{4.5cm}
 \defeq\ 
\slet{z}{M}{(\cbvcase z {x_1} {\lam{w}{1}{N_1}} {x_2} {\lam{w}{1}{N_2}}) (\star)} 
\end{align*}
where $w,z$ are fresh. 

These constructions have derived typing rules
\begin{equation}
\label{eq:derived:rules:sums:p}
\begin{array}{c}
\begin{prooftree}
\pj{\Gamma}{M}0
\justifies
\pj{\Gamma}{\cbvimagep{M}}\STA
\end{prooftree}
\GAP\GAP
\begin{prooftree}
\pj{\Gamma}{M_i}{\STA_i}
\justifies
\pj{\Gamma}{\cbvinp i{M}}{\STA_1+\STA_2}
\using
{(i=1,2)}
\end{prooftree}
\\[9mm]
\begin{prooftree}
\pj{\Gamma}{M}{\STA_1+\STA_2}
\GAP
\pj{\Gamma,x_i\co\STA_i}{N_i}{\STB}
\,\,
(i = 1,2)
\justifies
\pj{\Gamma}{\cbvcasep M {x_1} {N_1} {x_2} {N_2}}{\STB}
\end{prooftree}
\end{array}
\end{equation}
\\

For example:
\[
\begin{prooftree}
    \pj{\Gamma} M{\sigma_1+\sigma_2}
    \begin{prooftree}
      \begin{prooftree}
        \begin{prooftree}
          \pj{\Gamma,x_1\colon\sigma_1} {N_1}{\tau}
          \justifies 
          \vj{\Gamma,x_1\colon\sigma_1}{\lambda w.\,N_1}{1\pto \tau}
        \end{prooftree}
        \ \ 
        \begin{prooftree}
          \pj{\Gamma,x_2\colon\sigma_2} {N_2}{\tau}
          \justifies 
          \vj{\Gamma,x_2\colon\sigma_2}{\lambda w.\,N_2}{1\pto \tau}
        \end{prooftree}
        \justifies 
        \vj{\Gamma,z}
        {\cbvcase z {x_1} {\lambda w.\,N_1} {x_2} {\lambda w.\,N_2}}
        {1\pto \tau}
      \end{prooftree}
      \ \ 
      \begin{prooftree}
        \justifies
        \vj{\Gamma,z}\star 1
      \end{prooftree}
      \justifies 
      \pj{\Gamma,z\colon \sigma_1+\sigma_2}
      {(\cbvcase z {x_1} {\lambda w.\,N_1} {x_2} {\lambda w.\,N_2})\,(\star)}
      \tau
    \end{prooftree}
    \justifies 
    \pj{\Gamma}{\cbvcasep M {x_1} {N_1} {x_2} {N_2}}{\STB}
\end{prooftree}
\]
\hide{\[
\begin{prooftree}
\pj{\Gamma}M 0
\qquad
\begin{prooftree}
\begin{prooftree}
\begin{prooftree}
\justifies
\vj{\Gamma,x\co 0}x 0
\end{prooftree}
\justifies
\vj{\Gamma,x\co 0}{\cbvimage x}{1\pto \STA}
\end{prooftree}
\qquad
\begin{prooftree}
\justifies
\vj{\Gamma,x\co 0}\star\STA
\end{prooftree}
\justifies
\pj{\Gamma,x\co 0}{(\cbvimage x)(\star)}{\STA}
\end{prooftree}
\justifies
\pj{\Gamma}{\cbvimagep{M}}\STA
\end{prooftree}
\]}

We now refine the notion of Kleisli model (Def.~\ref{def:monadmodel})
to accommodate the sums.
%We do this following Power~\cite[Def.~36]{Power:GenericModels:06}.
\begin{defi}
\label{def:closeddistrFreyd}
A \emph{distributive enriched Kleisli model} 
(distributive Kleisli model for short)
is 
a distributive enriched model (Def~\ref{def:distrenrichedmodel}) together with
an identity-on-objects functor $J:\VCat\to\CCat$ that strictly
preserves copowers and coproducts.
\end{defi}
Note that in any Kleisli model, $J$ will preserve
coproducts because it has a right adjoint,
$\CCat(1,-)$.
We insist moreover that it preserves coproducts strictly.

Note that a distributive Kleisli model is what 
Power~\cite[Def.~36]{Power:GenericModels:06}
calls a distributive closed Freyd category.

It is straightforward to extend our interpretation
of the fine-grain call-by-value calculus in Kleisli models
(\S\ref{sec:interpkleilsi})
 to an interpretation of 
the calculus with sums in distributive Kleisli models.
The interpretation is sound and there is a syntactic model.

In Section~\ref{sec:monads} we discussed the equivalence 
between enriched Kleisli models and strong monads with Kleisli exponentials.
This equivalence extends to an equivalence between 
distributive enriched Kleisli models, and strong monads on distributive 
categories with Kleisli exponentials. 
Likewise, the constructions $\FreydToECBV$ and 
$\ECBVToFreyd$ of Section~\ref{sec:relating:models} 
extend to constructions deriving {\denrmodel}s from 
{\dKlmodel}s and vice versa, and an extension of Theorem~\ref{thm:adj}
states that $\FreydToECBV$ and 
$\ECBVToFreyd$ exhibit a 2-category of distributive Kleisli models 
as a coreflective subcategory of a 2-category of distributive enriched models. 

\subsection{Sums and the state-passing translation}
\label{sec:sums-sps}
It is straightforward to adapt the state-passing translation to 
accommodate sums.
\renewcommand{\CBVtoEEC}{\CBVtoEECbase\EECarbstate}%
\renewcommand{\CBVtoEECV}{\CBVtoEECbase \EECarbstate}%
\renewcommand{\CBVtoEECP}[1]{\CBVtoEECbase \EECarbstate{#1}_\svar}%
Recall that each type~$\sigma$ of FGCBV is translated
to a value type $\CBVtoEEC\sigma$ of ECBV. We set
\[\CBVtoEEC{0}\,\defeq\,0\qquad\qquad
\CBVtoEEC{(\STA+\STB)} \, \defeq \,
\CBVtoEEC \STA+\CBVtoEEC \STB\text.\]
We extend this translation to type contexts,
taking an FGCBV type context $\Gamma$ to an
ECBV type context $\CBVtoEEC\Gamma$.

Recall that 
the translation on terms takes an FGCBV value type judgement
${\vj{\Gamma}{V}{\STA}}$  
to an \ECBV\ 
judgement ${\tj{\CBVtoEEC{\Gamma}}{\CBVtoEECV{V}}{\CBVtoEEC{\STA}}}$,
and takes an FGCBV producer judgement $\pj{\Gamma}{M}{\STA}$ 
to an \ECBV\ judgement ${\aj{\CBVtoEEC{\Gamma}}{s \co \EECarbstate}{\CBVtoEECP{M}}{\ltensortype{(\CBVtoEEC{\STA})}{\EECarbstate}}}$.
We extend the translation in Section \ref{sec:translation}
straightforwardly, as follows:
\begin{gather*}
\CBVtoEECV{\cbvimage V} \defeq \vimage{\CBVtoEECV V}
\ \quad\ 
\CBVtoEECV{\cbvin 1 V} \defeq \vin 1{\CBVtoEECV V}
\ \quad\ 
\CBVtoEECV{\cbvin 2 V} \defeq \vin 2{\CBVtoEECV V}
\\
\CBVtoEECV{\cbvcase V{x_1}{W_1}{x_2}{W_2}}
\ \defeq \ 
\vcase{\CBVtoEECV{V}}{x_1}{\CBVtoEECV{W_1}}{x_2}{\CBVtoEECV{W_2}}
\end{gather*}
The translation remains sound.
The full definability results (Theorems~\ref{thm:full-on-types}
and~\ref{thm:full:faithful}) continue to hold in the presence of sums.

%%% Local Variables: 
%%% mode: latex
%%% TeX-master: "mogelberg-staton"
%%% End: 

\section{Remarks on the linear-use continuation-passing translation}
\label{sec:cps}

We now briefly emphasise that the linear-use continuation-passing
translation arises as a formal
dual of the linear-use state-passing translation.  
This is not a new observation: Hasegawa noticed it 
in the context of classical linear logic 
(\cite[\S 8]{Hasegawa:Flops:02},\cite{pmp-mec})
and indeed it informed the earlier work on the enriched effect calculus.

%This is a
%development of an idea of Hasegawa~\cite{Hasegawa:Flops:02} and
%P\'edrot~\cite{pmp-mec}. 
%This observation goes back to the work on the enriched effect calculus and is implicit in~\cite{Mogelberg:fossacs:10}. 
%It was also noticed by Hasegawa~\cite{Hasegawa:Flops:02} in the context of classical linear logic.

The linear-use continuation-passing translation 
was first elaborated by Berdine, O'Hearn, Reddy 
and Thielecke~\cite{Berdine:02}, 
but our main reference is 
the more recent work by Egger, M\o gelberg and
Simpson~\cite{Mogelberg:fossacs:10,EEC:LCPS:journal}. They showed that the linear-use
continuation-passing translation can be extended to an involutive
translation of the enriched effect calculus to itself, and derived a
full completeness result from this.  That work, in turn, stems from
Hasegawa's full completeness result~\cite{Hasegawa:Flops:02} for a
linear-use continuation-passing translation into dual intuitionistic /
linear logic.

Following~\cite{Mogelberg:fossacs:10},
our development is fuelled by the following categorical observation.
If a category $\CCat$ is enriched in $\VCat$ with copowers, then we
can form the dual category $\opcat\CCat$ which is also enriched in
$\VCat$, but now with powers instead of copowers.  (Recall that an
enriched category $\CCat$ has powers $\algY^\TVX$ if the functor
$\Homset{\VCat}{\TVX}{\CCat(-,\algY)}:\opcat\CCat\to \Set$ is
representable.)
When viewed under this duality, the state-passing translation
becomes the continuation-passing translation, as we now explain.

In Figure~\ref{figure:ecbv:powers}, we provide
an internal language for enriched categories with powers.
We call this the `CPS variant of ECBV', because
it is the target of the continuation passing translation 
(see below).
The key difference with Figure~\ref{figure:effects:typing} is that we 
have replaced the tensor type ($\ltensortype\VA\CB$) by
a power type ($\lpowertype \VA\CB$).
It is another fragment of the enriched effect calculus.
If $\VCat$ is a category with products
and $\CCat$ is a category enriched in $\VCat$ with powers,
then we can interpret this CPS variant of ECBV in $(\VCat,\CCat)$
through a variation of the interpretation in Section \ref{sec:adjmodels}.
The power type is interpreted using the categorical powers:
${\den {\lpowertype\VA\CB}\defeq
\den \CB^{\den\VA}}$.
A computation judgement $\aj\Gamma\Delta t \CA$ 
is interpreted as a morphism
$\den t:\den\Delta\to\den\CA^{\den\Gamma}$ in~$\CCat$.

\begin{figure*}[t]
\framebox{
\begin{minipage}{.96\linewidth}
\emph{Types.}
\begin{align*}
\VA,\VB \, & ::= \, \VconstA \,\mid \,\valone \,\mid \,\VA \prodtype \VB \, \mid\,  \CA \lpop \CB  \\
\CA,\CB \, & ::= \, \CconstA \,\mid  \,   \lpowertype{\VA}{\CB} \enspace .
\end{align*}
\begin{center}
\line(1,0){350}\gnl
\end{center}
\emph{Term formation.}
\begin{center}
\begin{gather*}
\prooftree
\justifies
\tj{\Gamma,\, \In{x}{\VA},\Gamma'}{x}{\VA}
\endprooftree
\GAP
\GAP
\prooftree
\justifies
\aj{\Gamma}{\In{z}{\CA}}{z}{\CA}
\endprooftree
\GAP
\GAP
\prooftree
\justifies 
\tj{\Gamma}{\algstar}{\valone}
\endprooftree
\gnl
\prooftree
\tj{\Gamma}{t}{\VA}
  \GAP
\tj{\Gamma}{u}{\VB}
\justifies 
\tj{\Gamma}{\pair{t}{u}}{\VA \prodtype \VB}
\endprooftree
\GAP\GAP
\prooftree
\tj{\Gamma}{t}{\VA_1 \prodtype \VA_2}
\justifies 
\tj{\Gamma}{\eecproj{i}{t}}{\VA_i}
\endprooftree
\gnl
\prooftree
\aj{\Gamma}{\In{z}{\CA}}{t}{\CB}
\justifies
\tj{\Gamma}{\llam{z}{\CA}{t}}{\CA \lfun \CB}
\endprooftree
\GAP 
\prooftree
\tj{\Gamma}{s}{\CA \lfun \CB} 
  \GAP
\aj{\Gamma}{\Delta}{t}{\CA} 
\justifies
\aj{\Gamma}{\Delta}{s [ t ]}{\CB}
\endprooftree
\gnl
\prooftree
\aj{\Gamma,x\co \VA}\Delta{t}{\CB}
\justifies
\aj{\Gamma}{\Delta}{\lpowerterm{x}{\VA}{t}}{\lpowertype{\VA}{\CB}}
\endprooftree
\GAP
\prooftree
\aj{\Gamma}{\Delta}{t}{\lpowertype{\VA}{\CB}}
\GAP
\tj{\Gamma, \, \In{x}{\VA}}{u}{\VA}
\justifies
\aj{\Gamma}{\Delta}{t\,u}{\CB}
\endprooftree
\end{gather*}
\end{center}
\begin{center}
\mbox{}\\
\line(1,0){350}\gnl
\end{center}
\emph{Equations: } Equations of Figure~\ref{figure:effects:typing} 
but with equations
for tensor types replaced by:
\begin{center}
\begin{gather*}
\begin{prooftree}
{\aj{\Gamma,x\co \VA}\Delta t \CB}
\quad
{\tj\Gamma u \VA}
\justifies
{\aeq \Gamma \Delta {(\lpowerterm x \CA t)\,u}{t[u/x]}\CB}
\end{prooftree}
\qquad
\begin{prooftree}
{\aj\Gamma\Delta t {\lpowertype \VA \CB}}
\justifies
{\aeq {\Gamma,x\co\VA}\Delta {t}{\lpowerterm x\VA (t\,x)}{\lpowertype\VA\CB}}
\end{prooftree}
\end{gather*}
\end{center}
\end{minipage}}
\caption{A CPS variant of 
  the enriched call-by-value calculus.}
\label{figure:ecbv:powers}
\end{figure*}

Following the categorical analysis above, 
we define a bijection $(-)^\circ$ between the types of
ECBV with this CPS variant:
\begin{equation}\label{eqn:SPStoCPS}\begin{gathered}
\alpha^\circ\defeq \alpha\qquad \qquad 
1^\circ \defeq 1\qquad\qquad 
(\VA\times\VB)^\circ \defeq (\VA^\circ\times \VB^\circ)
\\
(\CA\lpop\CB)^\circ \defeq (\CB^\circ\lpop \CA^\circ)
\qquad\qquad 
(\ltensortype \VA\CB)^\circ \defeq \lpowertype{\VA^\circ}{\CB^\circ}
\end{gathered}\end{equation}
This bijection extends to terms-in-context straightforwardly, and the resulting 
translation is a restriction of the linear-use cps involution of 
the enriched effect calculus studied in~\cite{Mogelberg:fossacs:10,EEC:LCPS:journal}.

We achieve a linear-use continuation-passing translation
by composing the state-passing translation of Section~\ref{sec:translation}
with this bijection (\ref{eqn:SPStoCPS}).
For clarity, we now write this out explicitly.
We fix a computation type $\EECarbret$ of \ECBV,
thought of as a return type.
%For a value type $\VA$, we define a shorthand 
%$\neg \VA\defeq (\lpowertype \VA\EECarbret)$ for the corresponding 
%computation type.
%
\newcommand{\lneg}[1]{\lpowertype {#1}\EECarbret}%
\renewcommand{\CBVtoEEC}{\CBVtoEECbase\EECarbret}%
\renewcommand{\CBVtoEECV}{\CBVtoEECbase \EECarbret}%
\renewcommand{\CBVtoEECP}[1]{\CBVtoEECbase \EECarbret{#1}_\kvar}%
We translate \FGCBV\ types $\sigma$ to \ECBV\ value types
$\CBVtoEEC\sigma$:
\begin{align*}
\CBVtoEEC{\alpha} & \, \defeq \VconstA &
\CBVtoEEC{(\STA \times \STB)} & \, \defeq\, 
\CBVtoEEC{\STA} \prodtype \CBVtoEEC{\STB} &
\CBVtoEEC{1} & \, \defeq \valone&
\CBVtoEEC{(\STA \pto \STB)} & \, \defeq \, 
(\lneg{(\CBVtoEEC{\STB})}) \lfun \, (\lneg{(\CBVtoEEC{\STA})})
\end{align*}
We extend this translation to type contexts,
taking an FGCBV type context $\Gamma$ to an
ECBV type context $\CBVtoEEC\Gamma$.

The translation on terms is syntax-directed.
We pick a variable $\kvar$, completely fresh.
The translation takes an FGCBV value type judgement
${\vj{\Gamma}{V}{\STA}}$  
to an \ECBV\
judgement ${\tj{\CBVtoEEC{\Gamma}}{\CBVtoEECV{V}}{\CBVtoEEC{\STA}}}$,
and it take an FGCBV producer judgement $\pj{\Gamma}{M}{\STA}$ 
to an \ECBV\ judgement ${\aj{\CBVtoEEC{\Gamma}}{\kvar \co \lneg {\CBVtoEEC \STA}}
{\CBVtoEECP{M}}{\EECarbret}}$.
\begin{gather*}
\CBVtoEECV{x} \defeq x 
\ \!\quad\!\ \CBVtoEECV{\star} \defeq \star 
\ \!\quad\! \ 
\CBVtoEECV{\pair{V}{W}}\defeq \pair{\CBVtoEECV{V}\!}{\CBVtoEECV{W}} 
\ \!\quad\!\ 
\CBVtoEECV{(\prj 1{V})}\defeq\prj 1{\CBVtoEECV{V}} 
\ \!\quad\!\ 
\CBVtoEECV{(\prj 2{V})}\defeq\prj 2{\CBVtoEECV{V}} 
\\
\begin{aligned}
\CBVtoEECP{(\return{V})} & \, \defeq \kvar \,(\CBVtoEECV V)
&
\CBVtoEECP{(\slet{x}{M}{N})} & \, \defeq 
(\llam \kvar {\lneg{\CBVtoEEC \sigma}}{\CBVtoEECP M})\,
[\lpowerterm x {\CBVtoEEC \sigma}{\CBVtoEECP N}]
\\
\CBVtoEECV{(\lam{x}{\STA}{N})} & \, \defeq 
\llam{k}{\lneg{\CBVtoEEC\STB}}{\lpowerterm x{\STA}{\CBVtoEECP N}} &
\CBVtoEECP{(V \, W)} & \, \defeq 
(\CBVtoEECV V\,[k])\,\CBVtoEECV W
\end{aligned}
\end{gather*}

The continuation-passing translation inherits the following results from the 
soundness and full completeness of the state-passing translation 
(Theorems~\ref{thm:soundness} and~\ref{thm:full:faithful}). 
\begin{prop}
For any computation type $\EECarbret$,
the continuation-passing translation is sound: 
\begin{enumerate}
\item
If $\veq\Gamma VW\sigma$ then 
$\teq{\CBVtoEEC \Gamma}{\CBVtoEECV{V}}{\CBVtoEECV{W}}{\CBVtoEEC \sigma}$.
\item 
If $\peq\Gamma MN\sigma$ then 
$\aeq{\CBVtoEEC \Gamma}{\kvar\co \lpowertype{\CBVtoEEC\sigma}\EECarbret}
{\CBVtoEECP{M}}{\CBVtoEECP{N}}{\EECarbret}$.
\end{enumerate}
\end{prop}
\renewcommand{\lneg}[1]{\lpowertype {#1}\EECret}%
\renewcommand{\CBVtoEEC}{\CBVtoEECbase\EECret}%
\renewcommand{\CBVtoEECV}{\CBVtoEECbase \EECret}%
\renewcommand{\CBVtoEECP}[1]{\CBVtoEECbase \EECret{#1}_\kvar}%
\begin{prop}[\cite{Mogelberg:fossacs:10}, Corollary~1]
Let $\EECret$ be a computation type constant.
The continuation-passing translation is fully complete, in the following 
sense.
\begin{enumerate}
\item \label{item:cps:faithful:values} 
If 
$\vj{\Gamma}{V,W}{\STA}$ 
and 
$\teq{\CBVtoEEC\Gamma}{\CBVtoEECV{V}}{\CBVtoEECV{W}}{\CBVtoEEC\STA}$ 
then $\veq \Gamma V W\STA$.
\item 
If  
$\pj{\Gamma}{M,N}{\STA}$
and 
$\aeq{\CBVtoEEC\Gamma}{\kvar\co \lneg {\CBVtoEEC\STA}}{\CBVtoEECP{M}}{\CBVtoEECP{N}}{\EECret}$ then 
$\peq{\Gamma}MN{\STA}$.
\item For any $\tj{\CBVtoEEC{\Gamma}}{t}{\CBVtoEEC{\STA}}$ there is  $\vj{\Gamma}{V}{\STA}$ such that $\teq{\CBVtoEEC\Gamma}t{ \CBVtoEEC{V}}{\CBVtoEEC\STA}$. 
\item For any $\aj{\CBVtoEEC{\Gamma}}{\kvar \co \lneg{\CBVtoEEC\STA}}{t}{\EECret}$ there is a producer term $\pj{\Gamma}{M}{\STA}$ such that 
$\aeq{\CBVtoEEC\Gamma}{\kvar \co \lneg{\CBVtoEEC\STA}}{t} {\CBVtoEEC{M}}{\EECret}$.
\end{enumerate}
\end{prop}

\noindent The full completeness result is the same as 
the one obtained by Egger 
et al.~\cite[Corollary~1]{Mogelberg:fossacs:10} 
except that loc.~cit.~describes a translation on the full enriched
effect calculus rather than this fragment of it.

\subsection{Sums and products}
The CPS variant of the enriched call-by-value calculus can be 
extended so that value types are closed under sums and 
computation types are closed under products. 
Thus the types are as follows:
\begin{align*}
\VA,\VB \, & ::= \, \VconstA \,\mid \,1 \,\mid \,\VA \prodtype \VB \, \mid\,  \CA \lpop \CB \,\mid\,0\,\mid\,\VA+\VB \\
\CA,\CB \, & ::= \, \CconstA \,\mid  \,   \lpowertype{\VA}{\CB} 
\,\mid\, \underline 1\,\mid\,\CA\mathop{\underline\times}\CB\enspace .
\end{align*}
(For brevity, we omit the term language, which is a fragment of 
the enriched effect calculus.)
The type system is designed to be dual to the enriched call-by-value language 
with sums. 
The translation from that language to this one~\eqref{eqn:SPStoCPS} 
is extended as follows:
\[(0)^\circ\defeq 0
\qquad
(\VA+\VB)^\circ \defeq \VA^\circ + \VB^\circ
\qquad
(\algzero)^\circ \defeq \Cone
\qquad
(\CA\algplus \CB)^\circ \defeq 
\CA^\circ\Ctimes\CB^\circ
\]
and the analysis
of Section \ref{sec:sums-sps} can be converted to
a fully-complete linear-use continuation-passing style translation
from fine-grain call-by-value with sums.

% !TEX root = Mogelberg-staton.tex

\section{Effect theories}
\label{sec:modelsoftheories}
To illustrate the nature of the state-passing translation
we endow our calculi with effects.
We do this is in a general way, by following the programme of 
Plotkin, Power and others~\cite{Plotkin:Power:03}
whereby a theory of effects is presented as an algebraic theory.

We discuss how to add effects to the source and target languages of
the state-passing translation, FGCBV and ECBV. Our central observation
is that to accommodate an algebraic theory of effects in the enriched
call-by-value calculus it is necessary and sufficient to supply the
chosen state type $\EECstate$ with the structure of a comodel.  The
idea of state being a comodel, particularly for the theory of store,
arose in the work of Plotkin, Power and
Shkaravska~\cite{Power:Shkaravska:04,Plotkin:Power:08}.

The section is structured as follows. We begin with an overview
in which we study the situation for a particular algebraic effect.
We then proceed to look at a general notion of effect theory~(\S\ref{sec:eff:theories}),
its relationship with the state-passing 
translation~(\S\ref{sec:effsps}) and notions of 
model and comodel~(\S\ref{sec:effmodels},\ref{sec:effmodels-example}).

\subsection{Overview}
\label{sec:overview}
\newcommand{\ifp}[3]{\codefont{if}^p~#1~\codefont{then}~#2~\codefont{else}~#3}
\newcommand{\cbvif}[3]{\codefont{if}~#1~\codefont{then}~#2~\codefont{else}~#3}
\newcommand{\sif}[3]{\mathrm{if}~#1~\mathrm{then}~#2~\mathrm{else}~#3}
\newcommand{\oread}[2]{#1\mathrel{?}#2}
\newcommand{\oprint}[1]{\mathrm{p}_{#1}}
\newcommand{\eprint}{\mathtt{print}}
\newcommand{\sprint}{\mathrm{print}}
\newcommand{\etoss}{\mathtt{toss}}
\newcommand{\stoss}{\mathrm{toss}}
\newcommand{\owrite}[2]{\mathrm{update}_{#1}(#2)}
\newcommand{\oflip}[1]{\mathrm{f}(#1)}
\newcommand{\cpsflip}[1]{\mathrm{f}[#1]}
\newcommand{\opread}[3]{#2\mathrel{?_{#1}}#3}
\newcommand{\opwrite}[3]{\mathrm{w}_{#2,#1}(#3)}
\newcommand{\opflip}[2]{\mathrm{f}_{#1}(#2)}
\newcommand{\oreadname}{(?)}
\newcommand{\oflipname}{(\mathrm{f})}
\newcommand{\owritename}[1]{\mathrm{w}_{#1}}
\newcommand{\eread}{\mathtt{deref}}
\newcommand{\ewrite}{\mathtt{assign}}
\newcommand{\eflip}{\mathtt{flip}}
\newcommand{\sread}{\mathrm{read}}
\newcommand{\sflip}{\mathrm{flip}}
\newcommand{\sassign}{\mathrm{write}}
\newcommand{\sderef}{\mathrm{read}}
We give an overview of the situation
in terms of a particular algebraic theory:
an algebraic theory for accessing a single bit of memory.
This is an algebraic theory in the classical sense 
(like monoids, groups, rings, etc.).
It has one binary operation $\oreadname$, a unary operation
$\oflipname$
and the following four equations:
\hide{\begin{equation}
\label{eqn:squareband}
x \equiv \oread {\owrite 0{x_0}} {\owrite 1{x_1}}
\qquad
\owrite i{\oread {x_0}{x_1}}
\equiv
\owrite i{x_i}
\qquad
\owrite i{\owrite j{x}}
\equiv
\owrite j{x}
\end{equation}}
\begin{equation}
\label{eqn:squareband}
\oread {(\oread v x)}{(\oread y z)}\equiv
\oread v z
\qquad
x\equiv \oread x x
\qquad
\oflip {\oflip x}\equiv x
\qquad
\oflip {\oread x y}
\equiv
\oread {\oflip y}{\oflip x}
\end{equation}
Here is some intuition. If $x$ and $y$ are computations,
then $\oread x y$ is the computation that first reads the bit 
in memory
and then branches to $x$ or to $y$ depending on whether the bit was set.
If~$x$ is a computation then $\oflip x$ is the computation that first
flips the bit in memory ($0$ to $1$, $1$ to~$0$) and
then continues as~$x$.
There are derived operations
$\owrite 0x \defeq \oread x{\oflip x}$ and 
$\owrite 1x \defeq \oread {\oflip x}x$,
which first write $0$ (resp.~$1$) and then continue as $x$.

We now explain how to accommodate this algebraic theory in
the fine-grain and enriched call-by-value calculi.

\subsubsection{Fine-grain call-by-value: algebraic operations and generic effects.}
\label{sec:fgcbv-geff}
In the fine-grain call-by-value calculus~(\S\ref{sec:fgcbv}), 
the algebraic theory~(\ref{eqn:squareband})
can be accommodated in two 
equivalent ways: by adding algebraic operations and by adding generic effects.
In adding the operations, we add the following term formation rules for each 
type $\sigma$:
\begin{equation} \label{eq:syntax:flip:alg:op}
\begin{prooftree}
\pj\Gamma M\sigma \quad\pj\Gamma N\sigma
\justifies
\pj\Gamma {\opread \sigma M N}\sigma
\end{prooftree}
\qquad\qquad
\begin{prooftree}
\pj\Gamma M\sigma
\justifies
\pj\Gamma {\opflip \sigma M}\sigma
\end{prooftree}
\end{equation}
We also add the equations in (\ref{eqn:squareband})
at each type, and an algebraicity equation
for each operation~(e.g.~\cite[Def~3.14]{Pretnar-Thesis}):
\[
\begin{prooftree}
\pj{\Gamma}{M_1}{\sigma}
\quad
\pj{\Gamma}{M_2}{\sigma}
\quad
\pj{\Gamma,x\colon \sigma}{N}{\tau}
\justifies
\peq\Gamma
{\slet x{(\opread \sigma {M_1}{M_2})}N}
{\opread \tau {(\slet x{M_1}N)}{(\slet x{M_2}N)}}
\tau
\end{prooftree}
\]\[
\begin{prooftree}
\pj{\Gamma}{M}{\sigma}
\quad
\pj{\Gamma,x\colon \sigma}{N}{\tau}
\justifies
\peq\Gamma{\slet x{\opflip \sigma M}N}
{\opflip\tau {\slet x M N}}\tau
\end{prooftree}\]
The result is a programming language with higher types and a single bit of 
storage.

The second way to accommodate the algebraic 
theory into the fine-grain call-by-value calculus
is by adding generic effects.
For this we need sum types (\S\ref{sec:sums:fgcbv}).
The idea is that an expression in $n$ variables in the algebraic 
theory corresponds to a ground producer term of type~$n$ ($=1+\dots+1$).
Thus we add the following axioms for term formation:
\begin{equation}\label{eq:syntax:flip:gen:eff}
\begin{prooftree}
\justifies
\pj{\Gamma}{\eread(\star)}{2}
\end{prooftree}
\qquad\qquad
\begin{prooftree}
\justifies
\pj{\Gamma}{\eflip(\star)}{1}
\end{prooftree}
\end{equation}
Informally, $\eread(\star)$ is a computation that returns
the boolean value of the memory cell,
and $\eflip(\star)$ is a computation that flips the value in the memory cell.
An important observation of Plotkin and Power~\cite{Plotkin:Power:03} 
is that the algebraic operations can be recovered at all
types from the generic effects,
as follows:
\[
\opread \sigma M N \defeq \ifp {\eread(\star)} M N
\qquad
\opflip \sigma M\defeq \eflip(\star);M
\]
where we use some shorthand notation:
\begin{align*}
\ifp {\eread(\star)} M N&\defeq \cbvcasep {\eread(\star)} {x_1}M{x_2}N
\\
\eflip(\star);M&\defeq
\slet x{\eflip(\star)}M
\end{align*}
Conversely the generic effects can be derived from the 
algebraic operations:
\[
\eread(\star) \defeq \opread {2} {\cbvinp 1 \star}{\cbvinp 2 \star}
\qquad 
\eflip(\star) \defeq \opflip {1}{\star}
\]
(The subscript $1$ on $\opflip 1\star$ is the unit type.)
We can thus write the four equations~(\ref{eqn:squareband}) 
directly in terms of generic effects:
\begin{equation}
\label{eqn:effstore}
\begin{gathered}
\peq{}{\slet x{\eread(\star)} {\slet y {\eread(\star)}}{\return \pair xy}}{\slet x{\eread(\star)} {\return {\pair xx}}}
{2\times 2}
\\[6pt]
\peq{}{\return{(\star)}}{\eread(\star);\return(\star)}1
\qquad\qquad
\peq{}{\eflip(\star);\eflip(\star)}{\return(\star)} 1
\\[6pt]
\peq{}{\eflip(\star);\eread(\star)}{\slet x {\eread(\star)}{\eflip(\star);\return{(\neg x)}}} {2}
\end{gathered}
\end{equation}
writing $\neg x$ for $\cbvif x {\cbvinr \star}{\cbvinl \star}$.

The two derived operations for writing a bit can 
be combined into a single command~$\ewrite$:
\[
\begin{prooftree}
\vj{\Gamma}{V}{2}
\justifies
\pj{\Gamma}{\ewrite(V)\defeq
\ifp {(\slet x {\eread(\star)}{\return{(V~\codefont{xor}~x)}})}{\eflip(\star)}
{\return (\star)}}
{1}
\end{prooftree}
\]
where $\codefont{xor}$ is the evident binary operation on values of type~$2$.
Using this derived command,
the four equations for accessing the bit of memory can be equivalently
written as the
three program equations of 
Plotkin and Power \cite{PlotkinPower:fossacs02}:
\begin{align}
\peq{-}
{&\,\return{(\star)}} 
{\slet{x}{\geffreadcell{l}(\star)}{\geffwritecell{l}(x)}}
1
\label{eq:GS1} 
\\
\peq {x\co2}
{& \,\geffwritecell{l}(x);\geffreadcell{l}(\star)}
{\geffwritecell{l}(x);\return{(x)}} 2
\label{eq:GS2} 
\\
\peq {x,y\co2}
{& \,\geffwritecell{l}(x);\geffwritecell{l}(y)}
{\geffwritecell{l}(y)} 1
\label{eq:GS3} 
\end{align}
which state that reading a cell and then writing the same value is the same as doing nothing, 
that writing and the reading yields the value just written, 
and that the effect of two writes equals that of the second.
\hide{
This simple example of an algebraic theory has been 
highlighted by Melli\'es, although our presentation (\ref{eqn:squareband})
appears to be novel.}

\subsubsection{Enriched call-by-value and state access operations}
How can we accommodate the algebraic theory for a bit of memory~(\ref{eqn:squareband}) in the
enriched call-by-value calculus?
In this section we develop the following observation.
Whereas in FGCBV each type should be a \emph{model}
of the theory,
in that (\ref{eq:syntax:flip:alg:op}) gives 
terms $\oreadname$ 
and $\oflipname$ at each type $\sigma$,
in ECBV the distinguished state type $\EECstate$ 
should be a \emph{comodel} of the theory, 
which means that there are ground value terms
\begin{equation}
\label{eqn:sreadflip}
\sread:\EECstate\lpop \EECstate\oplus\EECstate
\qquad\text{and}\qquad
\sflip:\EECstate\lpop \EECstate\end{equation}
which we call \emph{state access operations}.
(It is called a \emph{co}model 
because the arrows have been reversed and ($\times$)
has become ($\oplus$).)
Using the isomorphism $(\EECstate\oplus\EECstate)\cong
(\ltensortype 2\EECstate)$, we can 
understand the type of 
$(\sread)$ as ${\EECstate\lpop \ltensortype 2\EECstate}$.
The idea is that the 
read operation takes a state and returns the value stored 
in that state. 
It also returns a state: this is necessary because state is linear and
cannot be discarded or duplicated. 
Notice that, under the state-passing translation,
the two generic effects (\ref{eq:syntax:flip:gen:eff}) become 
the two state access operations.

The four equations~(\ref{eqn:squareband}) 
are also appealing when written in terms
of state access operations.
\begin{equation}\label{eqn:statestore}
\begin{aligned}
&\aeq{-}{\svar\co \EECstate}
{\letdot{b}{s'}{\sread[\svar]}{\letdot{b'}{s''}{\sread[s']}
{\ltensorterm{\pair b{b'}}{s''}}}\\&\hspace{5.2cm}}
{\letdot{b}{s'}{\sread[\svar]}{\ltensorterm{\pair bb}{s'}}}
{\ltensortype{(2\times 2)}\EECstate}
\gnl
&\aeq{-}{\svar\co\EECstate}
\svar
{\letdot b {s'}{\sread[s]}{s'}}
\EECstate
\gnl
&\aeq{-}{\svar\co\EECstate}
{\sflip[\sflip[\svar]]}
\svar
\EECstate
\gnl
&
\aeq{-}{\svar\co\EECstate}
{\sread[\sflip[\svar]]}
{\letdot b {s'}{\sread[s]}
\ltensorterm {\neg b}{\sflip[\svar']}}
{\ltensortype 2\EECstate}
\end{aligned}\end{equation}
Notice that the second equation says that the read operation does not change the state.

The two derived operations for writing a bit 
can be combined into a single state access operation:
\begin{multline*}
\mathrm{write}\defeq 
\llam x {}
{\letdot b s x 
{\letdot {b'}{s'}{\sread[s]}
{\big((\sif {(b\ \text{xor}\ b')}
{(\llam {s}{}{\sflip[s]})}
{(\llam {s}{S}{s})})
[s']\big)}}}
\\ :\ 
\ltensortype{2}\EECstate
\lpop \EECstate
\end{multline*}
Intuitively, $\mathrm{write}[\ltensorterm b s]$
writes the value $b$ to the state $s$,
 returning the updated state.

In Section~\ref{sec:translation}
we have seen a fully-complete state-passing translation 
from FGCBV to ECBV. 
This translation extends to a fully-complete translation
from FGCBV with generic effects to ECBV with state access operations.

\subsubsection{Continuation passing style and algebraic operations}
Finally, we turn to the linear-use continuation-passing style. 
In this setting, 
it is natural to require that 
the distinguished return type $\EECret$ be a model of the theory.
This is dual to the situation with state-passing style,
where the distinguished state type $\EECstate$ 
is a comodel of the theory.

More precisely, we extend the CPS variant of ECBV (\S~\ref{sec:cps}) with 
the theory of memory access
by adding ground value terms
\[
\oreadname:\EECret\mathop{\underline \times}\EECret \lpop \EECret
\qquad\text{and}\qquad
\oflipname:\EECret \lpop \EECret
\]
satisfying the following equations:
\begin{align*}
&
\aeq {-}{k\colon (\EECret\mathop{\underline\times} \EECret) \mathop{\underline\times}(\EECret\mathop{\underline\times}\EECret)}
{\oread 
{(\oread {(\pi_1(\pi_1\,k))}{(\pi_1(\pi_2\,k))})}
{(\oread {(\pi_2(\pi_1\,k))}{(\pi_2(\pi_2\,k))})}\\
&\hspace{6cm}}
{\oread 
{(\pi_1(\pi_1\,k))}{(\pi_2(\pi_2\,k))}}
{\EECret}
\\
&
\aeq {-}{k\colon \EECret}
k{\oread kk}\EECret
\\
&
\aeq {-}{k\colon \EECret}
{\cpsflip{\cpsflip k}}k\EECret
\\
&
\aeq {-}{k\colon \EECret\times \EECret}
{\cpsflip{\oread {(\pi_1\,k)}{(\pi_2\,k)}}}
{\oread {(\cpsflip{\pi_2\,k})}{(\cpsflip{\pi_1\,k})}}
\EECret\end{align*}
Thus
the generic effects in the source language 
endow the return type~$\EECret$ 
of the linear-use continuation-passing translation 
with the structure of a model for the algebraic 
theory.

\subsubsection{Further simple examples of algebraic theories for computational effects.}
\label{sec:furtherexamplesalgtheories}

The theory of accessing a bit of memory is perhaps the simplest example 
of a stateful effect.
The connections between algebraic operations, generic effects and state access operations
also work for less state-like effects.

\subsubsection*{Printing}
The algebraic theory of printing a single bit has 
two unary function symbols, $\oprint 0$ and~$\oprint 1$. 
For instance, the term $\oprint 0(\oprint 1(x))$ 
should be understood as the computation that first prints~$0$, then prints~$1$, then continues as $x$.
There are no equations in this theory.

The generic effects for printing can be grouped together into one
producer term
\[{\pj{x\colon 2}{\eprint{(x)}}1}\]
thought of as a command that prints its argument.

As a state access operation, 
we have a function
$\sprint : \ltensortype{2}\EECstate\lpop\EECstate$
which, given a bit and a state, returns a new state.
Intuitively, $\EECstate$ is a list of everything printed so far,
and $\sprint$ appends its first argument to its second argument.

\subsubsection*{Probability}
There are different algebraic theories of probabilistic choice.
The simplest one is the theory of `mean-values'
considered by Heckmann~\cite{h-probdom}
(but perhaps first introduced by Acz\'el~\cite{on-mean-values}):
it has one binary function 
symbol~$\odot$ and its axioms are
the medial law, idempotence and commutativity:
\[
(u\odot x)\odot(y\odot z)\equiv 
(u\odot y)\odot(x\odot z)
\qquad
x\equiv x\odot x
\qquad
x\odot y\equiv y\odot x
\]
The idea is that a computation $x\odot y$ tosses a coin, and proceeds as $x$ 
if heads and $y$ if tails.

The generic effect for $\odot$ is $\pj{\Gamma}{\etoss(\star)}{2}$ which,
intuitively, tosses the coin and returns the result.
In this style, the first equation is written
\[\peq{-}
{\slet{x}{\etoss(\star)}
\slet{y}{\etoss(\star)}
{(x,y)}}
{\slet y{\etoss(\star)}
\slet x{\etoss(\star)}
{(x,y)}}
{2\times 2}
\]
It says that it doesn't matter which order you toss coins.

The state access operation 
$\stoss\colon \EECstate\lpop\ltensortype{2}\EECstate$ 
can be thought of as making the coin an explicit parameter:
we can think of $\EECstate$ as a type of coins.
In this style, the second equation  
\[
\aeq{-}{s\colon\EECstate}
{s}
{\letdot{b}{s'}{\stoss[s]}{s'}}
{\EECstate}
\]
says that when you toss a coin you get the same coin back.
The third equation
\[
\aeq{-}{s\colon\EECstate}
{\stoss[s]}
{\letdot{b}{s'}{\stoss[s]}{(\ltensorterm{(\neg b)}s')}}
{\ltensortype 2\EECstate}
\]
says that if you toss a coin it is the same as tossing a coin 
and turning it once more without looking.
This illustrates that probability is not really stateful and so 
for this effect the approach based on algebraic operations
is perhaps the most profitable perspective. 
The point is that different computational effects
are better suited to different approaches (algebraic operations,
generic effects, and state access operations)
even though all three approaches are always available.

\subsection{State access operations, algebraic operations, and generic effects}
\label{sec:sao-so-geff}
We now make precise the informal connections made between 
state access operations, generic effects and algebraic operations
in the previous section.
We do this by revisiting the results
of Plotkin and Power~\cite{Plotkin:Power:03} 
in the context of an 
enriched model $(\VCat,\CCat)$.

In the previous section we focused on the classical situation
where arities are natural numbers.
However, from the perspective of 
state access operations, generic effects and algebraic operations
have little to do with
natural numbers per se.
It is just as easy to allow arities to be arbitrary 
objects of the base category $\VCat$.
In the following result we do not make an artificial 
restriction to natural numbers. Instead we consider an operation 
with arity $[A_1,\dots,A_n]$ and a parameter from $B$,
where $A_1,\dots,A_n,B$ are objects of $\VCat$.
The classical case of an $n$-ary operation is recovered 
by setting the objects $A_1,\dots,A_n,B$ to all be the terminal object $1$.
\begin{thm}
\label{thm:stateaccess-geneff-algop}
Let $(\VCat,\CCat)$ be an enriched model with sums. Let $\stateobj$ be an object
of $\CCat$. Let $A_1\dots A_n$ and $B$ be objects of $\VCat$.
The following data are equivalent:
\begin{enumerate}
\item A state access operation:
a morphism $\ltensor B \stateobj\to 
\ltensor{A_1}\stateobj
\algplus\dots\algplus\ltensor{A_n}\stateobj$ 
in $\CCat$.
\item A generic effect:
a morphism $B\to T_\stateobj(A_1+\dots+A_n)$ in $\VCat$,
where ${T_\stateobj}$ is the monad
${\CCat(\stateobj,\ltensor {(-)}\stateobj)}$.
\item An algebraic operation:
a $\VCat$-natural family of morphisms in $\VCat$
\[\left\{\Prod_{i=1}^n {(U_\stateobj\,\algX)^{A_i}}\to (U_\stateobj\,\algX)^{B}\right\}_{\algX\in\CCat}\]
where $U_\stateobj(\algX)\defeq \CCat(\stateobj,\algX)$.
\end{enumerate}
\end{thm}

\noindent The last point requires some explanation.
First, even though $\VCat$ is not cartesian closed,
exponentials with base $U_\stateobj(\algX)$ exist:
$(U_\stateobj\,\algX)^A\cong \CCat(\ltensor A\stateobj,\algX)$.
Second, the constructions
\[F\defeq \Prod_{i=1}^n {(U_\stateobj\,(-))^{A_i}}
\qquad G\defeq (U_\stateobj\,(-))^{B}
\]
can be understood as 
$\VCat$-functors $F,G:\CCat\to\VCat$,
since there are families of morphisms
\[
\{F_{\algX,\algY}:\CCat(\algX,\algY)\times F(\algX)\to F(\algY)\}_{\algX,\algY\in\CCat}
\qquad\quad
\{G_{\algX,\algY}:\CCat(\algX,\algY)\times G(\algX)\to G(\algY)\}_{\algX,\algY\in\CCat}
\]
in $\VCat$ that satisfy the laws for functors (respecting identities and 
composition).
Thirdly, 
a family of morphisms ${\{\phi_\algX:F(\algX)\to G(\algX)\}_{\algX\in\CCat}}$ 
is called $\VCat$-natural if the following diagram commutes in $\VCat$
for all $\algX$ and $\algY$:
\[
\xymatrix@C+20pt{
\CCat(\algX,\algY)\times F(\algX)\ar[r]^{\CCat(\algX,\algY)\times \phi_\algX}\ar[d]_{F_{\algX,\algY}}
&
\CCat(\algX,\algY)\times G(\algX)\ar[d]^{G_{\algX,\algY}}
\\
F(\algY)\ar[r]_{\phi_\algY} &G(\algY)
}
\]
(It is perhaps more compelling if
algebraic operations are defined as structure
on a $\VCat$-category of $T_\stateobj$-algebras,
but this $\VCat$-category cannot be constructed without further 
assumptions on $\VCat$ --- see \cite[\S7]{Plotkin:Power:03}.)
\begin{proof}[Proof of Theorem~\ref{thm:stateaccess-geneff-algop}]
To see the connection between~(1) and~(2),
consider the following bijections:
\begin{align*}
%\ordinary
\Homset{\CCat}{\ltensor B\stateobj}{\ltensor{A_1}\stateobj
\algplus\dots\algplus\ltensor{A_n}\stateobj}
&\ \cong\ 
\Homset{\VCat}{B}{\CCat(\stateobj,\ltensor{A_1}\stateobj
\algplus\dots\algplus\ltensor{A_n}\stateobj)}
\\&\ \cong\ 
\Homset{\VCat}{B}{\CCat(\stateobj,
\ltensor{{(A_1+\dots+A_n)}}\stateobj)}
\\&\ =\ 
\Homset{\VCat}{B}{T_\stateobj(A_1+\dots+A_n)}\text.
\end{align*}
To see the connection between~(1) and~(3),
we note that 
\[
\Prod_{i=1}^n(U_\stateobj\,\algX)^{A_i}
\ \cong\ 
\CCat(\ltensor{A_1}\stateobj\oplus\dots\oplus \ltensor{A_n}\stateobj,
\algX)
\qquad\text{and  }\qquad
(U_\stateobj\,\algX)^{B}
\ \cong\ 
\CCat(\ltensor{B}\stateobj,\algX)
\]
and the enriched Yoneda lemma gives
\[
\mathbf V\text{-Nat}(\CCat(\ltensor{A_1}\stateobj\oplus\dots\oplus \ltensor{A_n}\stateobj,-),
\CCat(\ltensor{B}\stateobj,-))
\ \cong\ 
\Homset{\CCat}{\ltensor{B}\stateobj}{\ltensor{A_1}\stateobj\oplus\dots\oplus \ltensor{A_n}\stateobj}\text.
\]
We remark that the equivalence of~(2) and~(3) is essentially Theorem~2 of \cite{Plotkin:Power:03}.\end{proof}

\subsection{Effect theories} 
\label{sec:eff:theories}
In the previous section we described connection between 
state access operations, generic effects and algebraic operations.
As we explained, the natural level of generality for this is more 
sophisticated than the classical setting:
the arity
of an operation is a list
and we allow the operation to take a parameter.
This suggests a generalization of algebraic theories
that we call `effect theories', since 
they are useful from the computational perspective.

The illustration in Section~\ref{sec:overview} 
involves storage of a single bit.
A motivating example of effect theory arises from modifying 
that theory above to allow storage of a more interesting datatype.
In FGCBV, we would like to have 
an (abstract) type $\val$ of 
storable values, and 
generic effects 
$\geffreadcell{l}$ and $\geffwritecell{l}$ with typing judgements
\begin{equation}
\label{eqn:gs-typing}
\begin{prooftree}
\phantom{\vj{\Gamma}{V}{\val}}
\justifies
\pj{\Gamma}{\geffreadcell{l}(\star)}{\val}
\end{prooftree}
\qquad\qquad
\begin{prooftree}
\vj{\Gamma}{V}{\val}
\justifies
\pj{\Gamma}{\geffwritecell{l}(V)}{1}
\end{prooftree} \end{equation}
We add to the theory of equality for 
FGCBV (Fig.~\ref{fig:lambdac:typing})
the three equations for global store proposed by 
Plotkin and Power~\cite{PlotkinPower:fossacs02}
(\ref{eq:GS1}--(\ref{eq:GS3})):
\begin{align*}
\peq{}
{&\,\return{(\star)}} 
{\slet{x}{\geffreadcell{l}(\star)}{\geffwritecell{l}(x)}}
1
\\
\peq {x\co\val}
{& \,\geffwritecell{l}(x);\geffreadcell{l}(\star)}
{\geffwritecell{l}(x);\return{(x)}} \val
\\
\peq {x,y\co\val}
{& \,\geffwritecell{l}(x);\geffwritecell{l}(y)}
{\geffwritecell{l}(y)} 1
\end{align*}

Our notion of effect theory accommodates the classical kinds of theory 
in the overview (\S~\ref{sec:overview})
and also the more general kind of theory of memory access illustrated
above. It is roughly the same as that used by Plotkin 
and Pretnar~\cite[\S 3]{Plotkin:Pretnar:09}. 
The main difference is in the presentation:
we use generic effects rather than algebraic operations.
Rather than introducing a new calculus for expressing the allowable
equations of an effect theory, we use the first-order fragment of FGCBV. 

%Our notion of effect theory accommodates the classical kinds of theory 
%in the overview (\S~\ref{sec:overview})
%and also the more general kind of theory of memory access illustrated
%above.
%We do not allow arbitrary extensions of 
%the fine-grain call-by-value language, however.
%Effect theories
%are extensions of FGCBV
%that are `algebraic' in the following sense: the connection between
%state access operations and generic effects is straightforward.
%
%Plotkin and Pretnar have also defined
%a notion of effect theory \cite[\S 3]{Plotkin:Pretnar:09}. 
%Their effect theories are roughly the same as ours. 
%The main difference is in the presentation:
%we use generic effects rather than algebraic operations.

\subsubsection*{Value theories}
Before we introduce effect theories we briefly discuss
\emph{value theories}, which are simple extensions 
of the value judgements of FGCBV.
By a value signature we shall simply mean a signature for a
many-sorted algebraic theory in the usual sense. This means a set of
type constants ranged over by $\alpha, \beta$, and a set of term
constants $f$ with a given arity $f \co (\alpha_1, \ldots, \alpha_n)
\to \beta$, where the $\alpha_i, \beta$ range over type constants. We
can extend {\FGCBV} along a value signature by adding the type
constants and the typing rule
\begin{equation}
\label{eq:value:constant}
\begin{prooftree}
\ltj{\Gamma}{t_i}{\alpha_i}  \, (i = 1,\ldots ,n)
\justifies 
\ltj{\Gamma}{f(t_1, \ldots, t_n)}{\beta}
\end{prooftree}
\end{equation}
for every term constant $f \co (\alpha_1, \ldots, \alpha_n) \to \beta$
in the signature. 
A value theory is a value
signature 
with a set of equations, i.e. pairs of terms
typable in the same context ${\ltj{\Gamma}{V=W}{\beta}}$, where $V,W$
are formed only using variable introduction and the rule
(\ref{eq:value:constant}).

\subsubsection*{Effect theories}
An \emph{effect signature} consists of a value theory and a set of
effect constants each with an assigned arity 
$\arityj{\geff}{\vec\beta}{\vec \alpha_1 + \ldots + \vec \alpha_n}$ consisting of a list
$\vec \beta$ of type constants and a formal sum of lists of type constants,
${\vec \alpha_1 + \ldots + \vec \alpha_n}$.
(Here we are abbreviating a list $(\beta_1\ldots\beta_m)$ using
the notation $\bar \beta$, etc.)
{\FGCBV} can be extended along an effect signature by
adding, for every effect constant 
$\arityj{\geff}{\vec \beta}{\vec \alpha_1 + \ldots +
  \vec \alpha_n}$, a typing judgement
\begin{equation} \label{eq:geff:intro}
 \begin{prooftree}
\vj{\Gamma}{V_1}{\beta_1} \quad\dots\quad
\vj{\Gamma}{V_m}{\beta_m}
\justifies
\pj{\Gamma}{\geff(V_1,\dots,V_m)}{\vec \alpha_1 + \ldots + \vec \alpha_n}
\end{prooftree}
\end{equation}
where
$\vec\beta=(\beta_1,\dots,\beta_m)$. 
In the conclusion, the vectors $\vec \alpha_i$ should be understood as 
the product of the types in the vector. 

Here are some examples of effect signatures:
\begin{itemize}
\item The theory of reading/flipping a bit of memory (\S~\ref{sec:overview})
has no value type constants.
It has two effect constants, $\eread$ and $\eflip$.
The effect constant $\eread$ has arity $\arit1{1+1}$
and the effect constant $\eflip$ has arity $\arit11$,
where $1$ is the empty string.
\item The theory for  storing an abstract datatype
(\ref{eqn:gs-typing})
has one value type constant $\val$
and a pair of effect constants
$(\geffreadcell{l} \co 1; \val)$ 
and
$(\geffwritecell{l} \co \val ; 1)$.
In this case term constants in the value theory can be used to add basic operations manipulating values in~$\val$:
we could ask that the storable values form a ring.
(In future, it would be interesting to allow $\val$ to
have more structure, for example as an inductive type such as the natural numbers,
but it is not clear how to extend the proof of 
Theorem~\ref{thm:embedding:ECBV:EEC} to inductive types.)
%In the case of nondeterminism, the effect constant $\random$ has arity $\arit{1}{1+1}$.
\end{itemize}

\noindent An \emph{effect theory} comprises an effect signature and a set of
equations. 
The equations are pairs 
of producer terms-in-context
${\peq{\Gamma}{M}{N}{\tau}}$
of a restricted kind: 
they must be built from the first-order fragment of fine-grain call-by-value in 
Figure~\ref{figure:effecttheories:typing}.
\begin{figure*}[tp]
\framebox{
\begin{minipage}{.96\linewidth}
\emph{Types.}
\begin{align*}
\STA,\STB \, & ::= \, \VconstA \,\mid \,1\,\mid\,\STA\times \STB \,\mid\,0
\,\mid \,\STA+\STB
\end{align*}
\begin{center}
\line(1,0){350}\gnl
\end{center}
\emph{Terms.}
The grammar for terms is as follows.
Typing judgements are in Figure~\ref{fig:lambdac:typing},
Figure~\ref{fig:lambdac:sumtyping},
and Equations~(\ref{eq:value:constant})
and~\eqref{eq:geff:intro}.
\begin{center}
\begin{align*}
V\,::=\ \,&
x
\,\mid\,
f(V_1,\dots,V_n)
\,\mid\,
\star
\,\mid\,
\fst V
\,\mid\,
\snd V
\,\mid\,
\pair {V_1}{V_2}
\\\mid\ &
\cbvimage {V}
\,\mid\,
\cbvin 1 V
\,\mid\,
\cbvin 2 V
\,\mid\,
\cbvcase V {x_1}{W_1}{x_2}{W_2}
\\[5pt]
M\,::=\ \,&
\return V
\,\mid\,
\slet x M N
\,\mid\,
\cbvcasep M {x_1}{N_1}{x_2}{N_2}
\end{align*}
\end{center}
\hide{\begin{center}\mbox{}\\
\line(1,0){350}\gnl
\end{center}
\emph{Equality.} We assume $\alpha$-equivalence, 
reflexivity, symmetry, transitivity and 
congruence laws.}
\end{minipage}}
\caption{The sub-calculus of fine-grain call-by-value 
that is used for describing effect theories}
\label{figure:effecttheories:typing}
\end{figure*}
This notion of `effect theory' is connected with
the classical notion of algebraic theory in Section~\ref{sec:overview}
as follows.
If the value theory is empty, with no type constants ($\alpha$) 
and no function symbols,
then the lists of type constants in (\ref{eq:geff:intro})
must all empty,
and each generic effect is an operation of arity $n$. 
This is the generic effect presentation
of an algebraic theory as described in Section~\ref{sec:fgcbv-geff}.

\hide{We write $\CBV{E}$ for the fine-grain call-by-value language augmented with 
an effect theory~$E$.}

\subsection{Effect theories and the state-passing translation}
\label{sec:effsps}
\renewcommand{\CBVtoEEC}{\CBVtoEECbase\EECstate}
\renewcommand{\CBVtoEECV}{\CBVtoEECbase \EECstate}
\renewcommand{\CBVtoEECP}[1]{\CBVtoEECbase \EECstate{#1}_\svar}
An effect theory is an extension to the fine-grain call-by-value language.
In Section~\ref{sec:translation} we 
explained how the linear-use state-passing translation
goes from FGCBV to ECBV. 
We now explain how ECBV needs to be extended to
support this state-passing translation.
The restricted nature of the effect theories makes this particularly
straightforward and illuminating.
\\
\emph{Value theory:}
\begin{itemize}
\item For each type constant in the value theory, we 
assume a value type constant.
\item For each term constant $f:\vec \alpha\to\beta$ in the value theory
we add the following term formation rule:
\[
\begin{prooftree}
\tj\Gamma {t_1}{\alpha_1}\qquad\dots\qquad
\tj\Gamma{t_n}{\alpha_n}
\justifies
\tj\Gamma{f(t_1,\dots, t_n)}{\beta}
\end{prooftree}
\]
\item 
An equation in the value theory must only be formed from variable introduction
and the rule~(\ref{eq:value:constant}).
Thus each equation $\veq \Gamma V W\beta$ in the value theory 
can be understood as an equation $\teq\Gamma VW \beta$ between 
value judgements.
\end{itemize}
\emph{Effect signature:}
\begin{itemize}
\item 
We assume a chosen computation type constant $\EECstate$.
\item For each effect constant
$\arityj e{\vec\beta}{\vec \alpha_1 + \ldots + \vec \alpha_n}$
we add a constant at value type,
\begin{equation} \label{eq:state:acc:type}
\tj {-}{e}{\ltensortype{\vec\beta}\EECstate
\lpop\ltensortype{(\vec \alpha_1 + \ldots + \vec \alpha_n)}\EECstate}
\end{equation}
For the theory of reading/flipping a bit, 
this yields the constants $\sread$ and $\sflip$ in (\ref{eqn:sreadflip}).
For the theory of storing an abstract datatype,
this yields two state access operations:
\[
\sderef : \EECstate\lpop \ltensortype \val\EECstate
\qquad\qquad
\sassign : \ltensortype\val\EECstate\lpop\EECstate
\]
\end{itemize}
\emph{State-passing translation:}
\begin{itemize}
\item 
Recall that the state-passing translation~(\S\ref{sec:translation})
takes a producer judgement of FGCBV
$\pj\Gamma M \sigma$ 
to a computation judgement of ECBV
$\aj{\CBVtoEEC\Gamma}{\svar\co\EECstate} {\CBVtoEECP M} {\ltensortype{\CBVtoEEC\sigma}{\EECstate}}$.
We extend the state-passing translation to operate on effects:
\[
\CBVtoEECP{(\geff(V_1\dots V_m))}
\defeq 
\geff(\ltensorterm{(\CBVtoEECV{V_1},\dots,\CBVtoEECV{V_m})}\svar)
\]
\item 
We use this extended state-passing translation to translate
the equations in the effect theory into ECBV,
in such a way that the extended translation is sound by construction.
Each equation in the effect theory
$\peq \Gamma M N \tau$ 
becomes an equation in ECBV:
\[
\aeq \Gamma {s\colon \EECstate}
{\CBVtoEECP{M}}
{\CBVtoEECP{N}}
{\ltensortype\tau\EECstate}\text.
\]
Notice that we do not need to translate the types in $\Gamma$ because
equations in an effect theory must be from the first-order fragment of
fine-grain call-by-value (Fig.~\ref{figure:effecttheories:typing})
which is shared with enriched call-by-value.  For instance, the
equations in the effect theory for reading/flipping a
bit~(\ref{eqn:effstore}) give rise to the equations on the state
object~(\ref{eqn:statestore}).  The three equations for storing an 
abstract datatype (\ref{eq:GS1}--\ref{eq:GS3}) become the following
equations for a state object $\EECstate$:
\[
\begin{aligned}
&\aeq{-}{\svar\co\EECstate}{\svar}{\sassign[\sderef[\svar]]}\EECstate
\\
&\aeq{x\co\val}{\svar\co\EECstate}{\sderef[\sassign[\ltensorterm x\svar]]}
{\ltensorterm x {(\sassign[\ltensorterm x \svar])}}{\ltensortype\val\EECstate}
\\
&\aeq{x,y\co\val}{\svar\co\EECstate}
{\sassign[\ltensorterm y{\sassign [\ltensorterm x  \svar]}]}
{\sassign[\ltensorterm y \svar]}
\EECstate
\end{aligned}
\]
\end{itemize}

%We write $\ECBVS{E}$ for the enriched call-by-value language augmented with 
%a state type $\states$ and constants and equations corresponding to the effect theory~$E$
%as described above.
%\todo{Is this notation needed? can it be defined where it is needed?}x

\subsection{Models and comodels of effect theories}
\label{sec:effmodels}
Our analysis of effect theories in Sections~\ref{sec:eff:theories}
and~\ref{sec:effsps} has been syntactic.
We now provide a model-theoretic treatment.
We define the interpretation of effect theories in Kleisli models
(\S\ref{sec:eff-kleislimodels})
and enriched models~(\S\ref{sec:eff-ecbvmodels}).
We then 
define what it means to be a model of an effect theory
in general terms~(\S\ref{sec:eff-cpsmodels}).

\subsubsection{Models of value theories.}
\label{sec:modelvaluetheories}
Let $\VCat$ be a distributive category.
An interpretation of a value signature in~$\VCat$ is given by 
interpretations of
the type constants
$\alpha$
as objects $\den\alpha$ of~$\VCat$,
and interpretations of term constants ${f \co
\vec\alpha \to \beta}$ as morphisms ${\den f \co
  \den{\vec\alpha} \to
  \den{\beta}}$.
(Here, if $\vec \alpha=(\alpha_1,\dots,\alpha_n)$ then
$\den{\vec\alpha}\defeq \den{\alpha_1}\times\dots\times\den{\alpha_n}$.)
This interpretation is extended to interpret a term in context
${\ltj{\Gamma}{V}{\beta}}$
as a morphism ${\den V\colon\den\Gamma\to\den\beta}$.
An interpretation of a value theory is an interpretation of the signature 
such that ${\den V = \den W}$ for each equation
${\veq{\Gamma}{V}{W}{\beta}}$ in the value theory.

\subsubsection{Interpreting effect theories in Kleisli models.}
\label{sec:eff-kleislimodels}
Let $(\VCat,\CCat,J)$ be a distributive Kleisli model
(Def.~\ref{def:closeddistrFreyd}) 
and suppose an interpretation of the type constants 
$\alpha,\beta$ is given.
An interpretation of an effect theory~$E$
in~$(\VCat,\CCat,J)$ 
is given by an interpretation of the value theory in $\VCat$
and 
an interpretation of each effect constant
$\geff\colon \vec \beta;\vec\alpha_1+\dots+\vec\alpha_n$ in~$E$
as a morphism 
$\den\geff\colon \den{\vec\beta}\to (\den{\vec\alpha_1}+\dots+\den{\vec\alpha_n})$ in $\CCat$,
satisfying the equations of the theory.

\subsubsection{Comodels of effect theories in enriched models.}
\label{sec:eff-ecbvmodels}
Let $(\VCat,\CCat)$ be a distributive enriched model 
in the sense of Definition \ref{def:distrenrichedmodel}.
Thus $\VCat$ is a distributive category and $\CCat$ is a category enriched
in $\VCat$ with copowers and coproducts.
A \emph{comodel} of the effect theory in $\CCat$ 
is an object $\stateobj$ of $\CCat$ 
together with a morphism
$\den\geff\colon \ltensor{\den{\vec\beta}}{\stateobj} \to \ltensor{(\den{\vec\alpha_1}+\dots+\den{\vec\alpha_n})}{\stateobj}$
in $\CCat$
for every effect constant 
$\arityj{\geff}{\vec\beta}{\vec \alpha_1 + \ldots + \vec \alpha_n}$
such that 
for each equation $\peq\Gamma M N \tau$ in the effect theory,
the interpretations of $M$ and $N$ under the state
passing style yield equal morphisms:
\[
\den{\CBVtoEECP{M}}=
\den{\CBVtoEECP{N}}:
\ltensor {\den{\Gamma}}\stateobj\to\ltensor{\den \tau}\stateobj\text.\]

\hide{\begin{prop}
Let  $(\VCat,\CCat, \stateobj)$ be an {\enrmodel} with a chosen state object, and let $E$ be an effect theory. The interpretation of {\ECBV} in $(\VCat,\CCat)$ extends to a sound interpretation of $\ECBVS{E}$ with $\den\states = \stateobj$ iff $\stateobj$ carries a comodel structure for $E$.
\end{prop}}
%By construction, the translations between 
%FGCBV and ECBV remain fully complete 
%when extended with effect theories.
%This fact transfers straightforwardly to the 
%categorical setting of Section~\ref{sec:relating:models}:
%the 2-functor $\FreydToECBV$ maps $E$-models to $E$-comodels
%and $\ECBVToFreyd$ maps $E$-comodels to $E$-models.

\subsubsection{Models of effect theories in dual enriched models.}
\label{sec:eff-cpsmodels}

We now justify our use of the term `comodel' in Section~\ref{sec:eff-ecbvmodels}
by showing that it is a generalization of the standard usage, i.e., dual to the concept of 
model of an algebraic theory known from classical algebra. 
Further investigations of the notion of effect theory used in this 
paper along with relations to existing notions of enriched algebraic 
theories \cite{kp-enrichedalgtheories,p-varieties} could be an interesting topic of further research. 

To dualize the notion of comodel, 
let $(\VCat,\opcat\CCat)$ be a distributive enriched model, i.e.,
let~$\VCat$ be a distributive category and $\CCat$ a category enriched
in~$\VCat$ with powers and products, as in Section~\ref{sec:cps}.
A \emph{model} of the effect theory in $\CCat$ 
is a comodel in $\opcat\CCat$.
Explicitly this amounts to an object $\retobj$ of $\CCat$ 
together with a morphism
$\den\geff\colon \lpower{(\den{\vec\alpha_1}+\dots+\den{\vec\alpha_n})}{\retobj}\to \lpower{\den{\vec\beta}}{\retobj}$
between powers in $\CCat$
for every effect constant 
$\arityj{\geff}{\vec\beta}{\vec \alpha_1 + \ldots + \vec \alpha_n}$
such that 
for each equation $\peq\Gamma M N \tau$ in the effect theory,
the interpretations of $M$ and $N$ in continuation
passing style yield equivalent morphisms:
\renewcommand{\CBVtoEEC}{\CBVtoEECbase\EECarbret}%
\renewcommand{\CBVtoEECV}{\CBVtoEECbase \EECarbret}%
\renewcommand{\CBVtoEECP}[1]{\CBVtoEECbase \EECarbret{#1}_\kvar}%
\[
\den{\CBVtoEECP{M}},
\den{\CBVtoEECP{N}}:
\lpower{\den \tau}\retobj\to \lpower {\den{\Gamma}}\retobj\text.\]
\newcommand{\ACat}{\mathbf{A}} \newcommand{\denot}[1]{\llbracket
  #1\rrbracket} \newcommand{\Amodel}{\retobj} %
\newcommand{\pjc}[3]{#1 \mathrel{\vdash^p_\mathrm{n}} #2 \colon \! #3}%
%\subsubsection*{Interpretations of effect theories in general.}
Because the terms of the effect theory are of a restricted kind,
it is straightforward to directly describe the interpretations of 
effect terms as morphisms between powers of $\retobj$,
%In any model~$\Amodel$ of an effect signature, we interpret 
%an effect term typing judgement 
%${\pj\Gamma M  \tau}$ as a morphism ${\den{M}\colon \Amodel^{\denot{\tau}}
%  \to \Amodel^{\denot{\Gamma}}}$ in~$\ACat$, 
by
induction on the structure of typing derivations.
For instance, consider the $\mathtt{case}^p$ rule in%
~(\ref{eq:derived:rules:sums:p}).  Given interpretations
${\denot{\CBVtoEECP{M}}\colon \Amodel^{\denot{\sigma_1+\sigma_2}} \to
\Amodel^{\denot{\Gamma}}}$ and $\denot{\CBVtoEECP{(N_i)}}\colon \Amodel^{\denot\tau}
\to \Amodel^{\denot{\Gamma, \sigma_i}}$ (${i=1,2}$), 
the interpretation
$\denot{\CBVtoEECP{(\cbvcasep M {x_1} {N_1} {x_2} {N_2})}}$ is the composite
\[
\Amodel^{\denot\tau}
\xrightarrow{(\denot{\CBVtoEECP{(N_1)}},\denot{\CBVtoEECP{(N_2)}})}
\Amodel^{\denot{\Gamma,\sigma_1}}
\times 
\Amodel^{\denot{\Gamma,\sigma_2}}
\cong
\Amodel^{(\denot{\sigma_1+\sigma_2})\times \denot\Gamma}
\xrightarrow{\denot{\CBVtoEECP{M}}^{\denot{\Gamma}}}
\Amodel^{\denot\Gamma\times\denot\Gamma}
\xrightarrow{\Amodel^\Delta}
\Amodel^{\denot\Gamma}\text.
\]
As another example, 
$\denot{\CBVtoEECP{(\return V)}} \defeq \Amodel^{\denot{V}}:
\Amodel^{\denot\tau}\to\Amodel^{\denot\Gamma}$
if $\vj \Gamma V \tau$.
\hide{A model of an effect theory in $\ACat$ is a model 
of the effect signature such that 
every effect equation 
${\Gamma\vdash M\equiv N\colon \tau}$
in the theory is satisfied, i.e.
$\den M=\den N$.}

We now return to  the setting of classical algebra, 
when the value theory has no type constants or value constants.
We will show that models in the above sense
are models of algebraic theories in the classical sense.
If there are no type constants then every type in the effect language
(Figure~\ref{figure:effecttheories:typing})
is isomorphic to one of the form
${1+1+\dots+1}$.
The arity of an effect constant
${\arityj{\geff}{\vec\beta}{\vec \alpha_1 + \ldots + \vec \alpha_n}}$
must comprise $\vec\beta$ as the empty list (since there are no
type constants) and
$\vec \alpha_1 + \ldots + \vec \alpha_n$
must be a sequence of $n$ empty lists.
Thus the interpretation of $\geff$ in the model is a morphism
$\den\geff\colon \lpower{n}{\retobj}\to {\retobj}$.

We now explain the nature of equations in this restricted setting.
In what follows, we will 
write a natural number $n$ for the type 
that is the $n$-fold sum of~$1$.
We will write $\cbvin i\star$ for the $i$th
injection of type $n$, where $1\leq i\leq n$,
and we will make use of $n$-ary case constructions,
$\cbvcasen{V}\star{W_1}{n}\star{W_n}$
to destruct terms~$V$ of type~$n$.
These are just syntactic shorthand for terms that can be defined in 
the language in Figure~\ref{figure:effecttheories:typing}.

We can focus on equations where $\Gamma$ is empty.
This is because every context has a finite number of ground valuations
---
if a variable $x$ in a context has type~$n$,
then it could be valued with $\cbvin 1\star\dots \cbvin n\star$
--- and, moreover,
an effect equation 
$\peq\Gamma M N n$ is satisfied 
if and only if it is satisfied at each ground instantiation.

The next step is to note that every effect term
$\pj{-}M{1+1\dots+1}$
is equal to one built from the following rules:
\begin{gather*}
\begin{prooftree}
\pjc{-}{M_1}{n}
\quad\dots
\quad
\pjc{-}{M_m}{n}
\justifies
\pjc-{\cbvcasenp{\geff(\star)}{\star}{M_1}{m}{\star}{M_m}}
n
\using{(\geff:-;m)}
\end{prooftree}
\\[6pt]
\begin{prooftree}
\justifies
\pjc{-}{\cbvinp i{\return{\star}}}
{n}
\using
{1\leq i\leq n}
\end{prooftree}
\end{gather*}
It is informative to look at the interpretation of these normalized terms.
Given interpretations $\den {M_1},\dots,\den{M_m}\colon \retobj^n\to\retobj$,
we have 
\begin{align*}
\den{\cbvcasenp{\geff(\star)}\star{M_1}{m}\star{M_m}}
&=
\retobj^n\xrightarrow {(\den{M_1},\dots,\den{M_m})}\retobj^m
\xrightarrow {\den \geff}
\retobj\text,
\\
\den {\cbvinp i{\return\star}}
&=
\retobj^n\xrightarrow{\pi_i}\retobj\text.
\end{align*}
Thus we see that, in 
the situation where there are no type constants 
or value constants, 
the new general notion of model is 
the classical notion of model for an algebraic theory.

\subsection{Examples of set-theoretic models and comodels}
\label{sec:effmodels-example}
We revisit the simple example effect theories from
Sections~\ref{sec:overview} and~\ref{sec:eff:theories}
from the model-theoretic perspective.
In each case, we find that there are comodels that are state-like.
\subsubsection{Storage}
The category $\Set$ is enriched in itself with
copowers given by products and the enrichment given by 
the function space.
The set $2=\{0,1\}$ is a comodel for the theory of 
accessing a bit of memory~(\S\ref{sec:overview}),
with $\sread(x)=(x,x)$ and $\sflip(x)= \neg x$.
This is a comodel for the theory 
in the enriched model
$(\Set,\Set)$.  
Power and Shkaravska~\cite{Power:Shkaravska:04} showed that~$2$ is 
the final comodel in $\Set$ for the theory of accessing a bit of memory.

As an aside, we note that $\Set$ is actually equivalent to 
the category of models for the theory of accessing a bit of memory.
The theory of \emph{reading} a bit of memory is sometimes called the theory
of `rectangular bands' because 
every model is isomorphic to one of the form $X\times Y$,
with $\oread{(x,y)}{(x',y')}\defeq (x,y')$. 
The anti-involution operation $\oflipname$ enforces that the model is
isomorphic to 
one of the form $X\times X$, and thus determined by a single set.
This phenomenon has been investigated in a more general setting by 
M\'etayer~\cite{metayer-state} and Mesablishvili~\cite{mesablishvili-state}.

We can consider set theoretic models of 
the theory of storing an abstract datatype~(\ref{eqn:gs-typing}).
What is needed is an interpretation
$\val$ of the value sort,
which also plays the role of the state object.
We let 
$\sread(x)=(x,x)$ and $\sassign(v,s)=v$.
This is a comodel for the theory for
store in the enriched model
$(\Set,\Set)$.  

In both cases, 
the induced monad on $\Set$ is 
the store monad $((-)\times\stateobj)^{\stateobj}$.

\subsubsection{Printing}
\newcommand{\TwoStarAct}{\textbf{$2^*$-Act}}
Let $\TwoStarAct$ be the category
of algebras for 
the theory of printing a bit.
The objects are triples $(X,\oprint 0,\oprint 1)$ 
where $\oprint 0,\oprint 1: X\to X$,
and the morphisms $(X,\oprint {X,0},\oprint {X,1})\to (Y,\oprint {Y,0},\oprint {Y,1})$
are functions $X\to Y$ that commute with the operations.
As any ordinary category, this category is enriched in $\Set$.
It has copowers given by 
\[
\ltensor A(X,\oprint {X,0},\oprint{X,1})
\defeq
(A \times X,\oprint {(\ltensor AX),0},\oprint{(\ltensor A X),1})
\qquad\text{where\ }
\oprint{(\ltensor AX),i}(a,x)
\defeq (a,\oprint {X,i}(x))\text.\]
Thus $(\Set,\TwoStarAct)$ is an enriched model.

The algebra structure of each algebra 
equips it with the structure of a comodel
in the category of algebras.
The leading example
is the set $2^*$ of strings over $\{0,1\}$, 
with $\oprint{2^*,i}(s)=si$.
The induced state monad $\TwoStarAct(2^*,\ltensor{(-)}2^*)$ 
is isomorphic to the monad
$2^*\times(-)$ on $\Set$.
We can understand a string in $2^*$ as a state: it is the 
list of things output so far.

\subsubsection{Probability}
\newcommand{\MeanValueAlgebras}{\textbf{MVAlg}}
Let $\MeanValueAlgebras$ be the category
of mean-value algebras.
The objects are pairs $(X,\odot)$ 
of a set $X$ and a binary operation~$\odot$ 
satisfying the laws of mean-value algebras~(\S\ref{sec:furtherexamplesalgtheories}).
The pair $(\Set,\MeanValueAlgebras)$ is an enriched model.

The one-element set is trivially a mean-value algebra,
and it can be given the structure of a comodel in the category of 
mean-value algebras.
We can understand the one-element set as a set of states:
this captures the idea that probability is a stateless notion of computation.
Nonetheless, this `state object' induces a `state monad' on $\Set$.
This can be understood as a monad $D$ of 
finite dyadic probability 
distributions.
By a finite dyadic probability distribution 
on a set $X$, 
we mean a function $p:X\to[0,1]$ 
such that 
$\mathsf{supp}(p)=
\{x\in X~|~p(x)\neq 0\}$ is finite,
$\sum_{x\in\mathsf{supp}(p)}p(x)=1$,
and for all $x$, $p(x)$ has a finite binary representation.
%\begin{align*}
%D(X)\defeq \{p:X\to [0,1]~|~&\mathsf{supp}(p)=
%\{x\in X~|~p(x)\neq 0\}\text{ is finite}\}
%\\&\text{and $\textstyle{\sum_{x\in\mathsf{supp}(p)}p(x)=1}$}
%\\&\text{and for all $x$, $p(x)$ has a finite binary representation}\}\text.
%\end{align*}
The monad $D$ has 
$D(X)$ as the set of all finite dyadic probability distributions;
the unit picks out the Kronecker distributions,
and multiplication 
$\mu_X:D(D(X))\to D(X)$ 
takes a distribution $p:D(X)\to [0,1]$ on $D(X)$ 
to a distribution $\mu_X(p):X\to[0,1]$ on $X$,
given by $\mu_X(p)(x)\defeq\sum_{q\in \mathsf{supp}(p)}( p(q)\times q(x))$.

In general, when the state object is a terminal object
then the induced monad preserves terminal objects.
A terminal-object-preserving monad is sometimes called 
affine \cite[Thm.~2.1]{bilinearity-and-cc-monads}
and the corresponding effects are said to be 
discardable (e.g. \cite{DBLP:conf/fossacs/Fuhrmann02}, \cite[Def.~4.2.4]{hayo-phd}) since the following rule is admissible
in the fine-grain call-by-value language.
\[
\begin{prooftree}
\pj \Gamma t A
\quad\pj \Gamma u B
\justifies
\peq \Gamma {\slet x t u} u B
\using{\ (x\text{ not free in }u)}
\end{prooftree}
\]

\subsection{Relating notions of (co)model for effect theories}

We now extend Theorem~\ref{thm:stateaccess-geneff-algop} to show, for each effect theory $E$, a bijective correspondence between 
the following: comodel structures on $\stateobj$, interpretations of $E$ in the Kleisli model 
$(\VCat, \KlCat{\VCat}{\CCat}{\stateobj}, J_{\stateobj})$, algebraic operations equipping each 
$U_\stateobj\,\algX$ with a model structure for $E$. The latter notion requires some explanation, because the definition of model given in 
Section~\ref{sec:eff-cpsmodels} defines only what it means for $\retobj$ in a category $\CCat$ 
to be a model of $E$ if $(\VCat, \opcat{\CCat})$ 
is an enriched model, and in the setting of Theorem~\ref{thm:stateaccess-geneff-algop} $\opcat{\VCat}$ is generally not 
$\VCat$-enriched.

\newcommand{\retVobj}{R}%
To generalize the notion of model, let $\VCat$ be a distributive category, let $\retVobj$ be a fixed object of $\VCat$ such that
all exponents of the form $\retVobj^{A}$ exist
(i.e.~$\VCat(-\times A,\retVobj)$ is representable for all $A$)
and let an interpretation of 
the value theory of $E$ be given. We define a model structure for $E$ on $\retVobj$ to be an
interpretation of $E$ 
in the Kleisli model $(\VCat, \KlCat{\VCat}{\VCat}{\retVobj^{\retVobj^{(-)}}}, J)$
where $\KlCat{\VCat}{\VCat}{\retVobj^{\retVobj^{(-)}}}$ has the same objects as
 $\VCat$, but where a morphism $A\to B$ in $\KlCat{\VCat}{\VCat}{\retVobj^{\retVobj^{(-)}}}$ is a morphism $\retVobj^{B} \to \retVobj^{A}$ in $\VCat$.
 This is isomorphic to the Kleisli category for the strong monad 
$\retVobj^{\retVobj^{(-)}}$ on $\VCat$. 
By construction, the model structure interprets each effect constant 
$\arityj{\geff}{\vec\beta}{\vec \alpha_1 + \ldots + \vec \alpha_n}$ as a morphism
\[
\den\geff\colon \lpower{(\den{\vec\alpha_1}+\dots+\den{\vec\alpha_n})}{\retVobj}\to \lpower{\den{\vec\beta}}{\retVobj} \,.
\]
If $\VCat$ is cartesian closed then $(\VCat, \opcat\VCat)$ is an enriched model and the above definition of a model structure for~$E$ on~$R$ is equivalent to 
the one given in Section~\ref{sec:eff-cpsmodels}.

\begin{thm}
\label{thm:correspondence:thm}
Let $(\VCat,\CCat)$ be an enriched model with sums, let $\stateobj$ be an object
of $\CCat$, let $E$ be an effect theory and let an interpretation of the value theory of $E$ in $\VCat$ 
be given. 
The following data are equivalent:
\begin{enumerate}
\item \label{item:comodel} A comodel structure for $E$ on $\stateobj$
\item \label{item:Kleisli:model} An interpretation of $E$ in the Kleisli model $(\VCat, \KlCat{\VCat}{\CCat}{\stateobj}, J_{\stateobj})$ 
\item \label{item:alg:op:model} For each effect constant $\arityj{\geff}{\vec\beta}{\vec \alpha_1 + \ldots + \vec \alpha_n}$ in $E$ an algebraic operation:
a $\VCat$-natural family of morphisms in $\VCat$
\[\left\{\Prod_{i=1}^n {(U_\stateobj\,\algX)^{\den{\vec\alpha_i}}}\to (U_\stateobj\,\algX)^{\den{\vec\beta}}\right\}_{\algX\in\CCat}\]
equipping each $U_\stateobj(\algX)\defeq \CCat(\stateobj,\algX)$ with a model structure for $E$.
\end{enumerate}
\end{thm}

\renewcommand{\CBVtoEEC}{\CBVtoEECbase\EECstate}
\renewcommand{\CBVtoEECV}{\CBVtoEECbase \EECstate}
\renewcommand{\CBVtoEECP}[1]{\CBVtoEECbase \EECstate{#1}_\svar}
\begin{proof}
We first prove equivalence of (\ref{item:comodel}) and (\ref{item:Kleisli:model}). First note that in both 
cases, an effect constant 
$\arityj{\geff}{\vec\beta}{\vec \alpha_1 + \ldots + \vec \alpha_n}$ is modelled as a morphism
\[
\den\geff\colon \ltensor{\den{\vec\beta}}{\stateobj} \to \ltensor{(\den{\vec\alpha_1}+\dots+\den{\vec\alpha_n})}{\stateobj} \,.
\]
It thus suffices to show that for any term 
$\pj{\Gamma}{M}{\STA}$ of the fragment
of Figure~\ref{figure:effecttheories:typing} 
the morphisms $\den{\CBVtoEECP{M}}, \den{M}:\ltensor{\den\Gamma}\stateobj
\to \ltensor{\den \STA}\stateobj$ are equal, where $\den{\CBVtoEECP{M}}$
is the \ECBV\ term $\CBVtoEECP{M}$ interpreted 
in the enriched model $(\VCat, \CCat, \stateobj)$ and $\den{M}$ is
the fine-grain call-by-value term $M$ interpreted in the Kleisli model $(\VCat, \KlCat{\VCat}{\CCat}{\stateobj}, J_{\stateobj})$. This can be
proved by induction on the structure of $M$.

To prove equivalence of (\ref{item:comodel}) and (\ref{item:alg:op:model}) we build on the equivalence of 
state access operations and algebraic operations of Theorem~\ref{thm:stateaccess-geneff-algop}. First we show that
for any term $\pj{\Gamma}{M}{\STA}$ of the fragment
of Figure~\ref{figure:effecttheories:typing} the equation 
\begin{equation} \label{eq:yon:pres:fgcbv}
\CCat(\den{\CBVtoEECP{M}},\algX) = \den{M} \co \CCat(\ltensor{\den\STA}{\stateobj},\algX) \to \CCat(\ltensor{\den\Gamma}{\stateobj},\algX)
\end{equation}
holds. This time the denotation brackets on the left hand side refers to the interpretation of ECBV in 
$(\VCat, \CCat, \stateobj)$, and 
the denotation brackets on the right hand side refer to the interpretation of fine-grain call-by-value in the Kleisli model 
$(\VCat, \KlCat{\VCat}{\CCat}{\retVobj^{\retVobj^{(-)}}}, J)$, where ${\retVobj=U_\stateobj(\algX)}$. As above, $\den{\CBVtoEECP{M}}$ is simply the interpretation  
of $M$ in the Kleisli model $(\VCat, \KlCat{\VCat}{\CCat}{\stateobj}, J_{\stateobj})$. 
Equation (\ref{eq:yon:pres:fgcbv}) can be proved by induction on the structure of $M$. 
(There is also a categorical perspective on this, based around the construction
$\CCat(-,\algX)$ which gives rise to
an identity-on-objects functor
$\KlCat{\VCat}{\CCat}{\stateobj}\to\KlCat{\VCat}{\CCat}{\retVobj^{\retVobj^{(-)}}}$
that preserves sums and the action of $\VCat$, although it does not preserve 
the enrichment.)

% Essentially, this boils down to the observation that 
% the pair
% \[
% (\id, \CCat(-,\algX)) \co (\VCat, \KlCat{\VCat}{\CCat}{\stateobj}, J_{\stateobj}) \to (\VCat, \KlCat{\VCat}{\CCat}{\retVobj^{\retVobj^{(-)}}}, J) 
% \]
% is almost a strict morphism of Kleisli models, except that it doesn't preserve the enriched structure.
% is a morphism of Kleisli models in the strict sense, i.e., commutes with all structure on the nose, it preserves the 
% interpretation of fine-grain call-by-value, thus proving (\ref{eq:yon:pres:fgcbv}). 

From (\ref{eq:yon:pres:fgcbv}) we deduce the equivalence of (1) and (3). 
In fact (1)$\implies$(3) is immediate:
suppose $\stateobj$ is a comodel, and consider an 
equation ${\Gamma\vdash M\equiv N\colon \tau}$ in the theory.
Since $\stateobj$ is a comodel, we have 
${\den{\CBVtoEECP{M}} = \den{\CBVtoEECP{N}}}$
and so 
$\den M = \den N$ as interpreted in $(\VCat, \KlCat{\VCat}{\CCat}{\retVobj^{\retVobj^{(-)}}}, J)$. 

For (3)$\implies$(1), suppose that 
$\retVobj = \CCat(\stateobj,\algX)$ is a model for every $\algX$, 
naturally in $\algX$. Then by (\ref{eq:yon:pres:fgcbv})  
\[ 
\CCat(\den{\CBVtoEECP{M}},\algX) = \CCat(\den{\CBVtoEECP{N}},\algX)
\ :
\ 
\CCat(\ltensor{\den\tau}\stateobj,\algX)\to \CCat(\ltensor{\den\Gamma}\stateobj,\algX)
\] 
holds for all equations ${\Gamma\vdash M\equiv N\colon \tau}$ and all $\algX$.
The enriched Yoneda embedding
is full and faithful and so
$\den{\CBVtoEECP{M}} = \den{\CBVtoEECP{N}}:\ltensor{\den\Gamma}\stateobj\to\ltensor{\den\tau}\stateobj$, proving that $\stateobj$ is a comodel.

We remark that the equivalence of (2) and (3) is in the spirit of
\cite[\S 6]{Plotkin:Power:03}.
\end{proof}

\subsection{Generalizing full completeness to the case of effects}

\label{sec:fullcomplete-effects}

The full completeness result of Theorem~\ref{thm:full:faithful}
extends verbatim to the case of the calculi augmented with an effect
theory $E$. The proof is based on an extension of the coreflection
theorem (Theorem~\ref{thm:adj}) which we state below.

First 2-categories $\Freydtheory{E}$ and $\CATECBVtheory{E}$ of distributive Kleisli models of $E$ and 
distributive enriched models of $E$ are defined. These are defined similarly to $\Freyd$
and $\CATECBV$ except morphisms are required to preserve coproducts (up to isomorphism) and 
the interpretation of $E$ (strictly). Details can be found in Appendix~\ref{app:def:cats:eff}. 

\begin{lem} \label{lem:functors:eff}
The assignments $\FreydToECBV(\VCat, \CCat, J) \defeq (\VCat, \CCat, 1)$ and 
$\ECBVToFreyd(\VCat, \CCat, \stateobj) \eqdef (\VCat, \KlCat{\VCat}{\CCat}{\stateobj}, J_{\stateobj})$ 
extend to 2-functors 
\begin{align*}
\FreydToECBV & \co \Freydtheory{E} \to \CATECBVtheory{E} \\
\ECBVToFreyd & \co \CATECBVtheory{E} \to \Freydtheory{E} 
\end{align*}
\end{lem}

\begin{proofsketch}
We just show that these are well-defined on objects. The case of $\ECBVToFreyd$ is simply the implication
from (\ref{item:comodel}) to (\ref{item:Kleisli:model}) of Theorem~\ref{thm:correspondence:thm}.

In the case of $\FreydToECBV$ we must show that $1$ carries a comodel structure 
for $E$ whenever $(\VCat, \CCat, J)$ models~$E$. An effect constant 
${\arityj{\geff}{\vec\beta}{\vec \alpha_1 + \ldots + \vec \alpha_n}}$  can be modelled 
as the composite
\[\ltensor{\den{\vec\beta}}{1}
\xrightarrow{J(\pi_1)}
\den{\vec\beta}
\xrightarrow{\den e} 
\den{\vec \alpha_1 + \ldots + \vec \alpha_n} 
\xrightarrow{J\pair{\id}{\bang{}}}
\ltensor{\den{\vec \alpha_1 + \ldots + \vec \alpha_n}}{1}
\]
where $\den{e}$ refers to the interpretation of $e$ in the given
$E$-model structure of $(\VCat, \CCat, J)$.  We must show that this
defines a comodel, i.e., that the equations are satisfied.  To this
end one can prove that for any term $\pj{\Gamma}{M}{\STA}$ of the
fragment of fine-grain call-by-value used for effect theories
(Figure~\ref{figure:effecttheories:typing}) the equation
$\den{\CBVtoEECP{M}} = J(\pair{\id}{\bang{}}) \circ \den M \circ
J(\pi_1)$ holds. Here, on the left hand side the double brackets refer to the
interpretation of {\ECBV} in $(\VCat, \CCat, 1)$ and the brackets on
the right hand side refer to the interpretation of fine-grain
call-by-value in $(\VCat, \CCat, J)$. This equation is proved by
induction on typing derivations. Thus, for any equation $\peq \Gamma M
N \tau$ in $E$, since $(\VCat, \CCat, J)$ models $E$ we have $\den M =
\den N$, and thus also $\den{\CBVtoEECP{M}} = \den{\CBVtoEECP{N}}$,
proving that $1$ is indeed a comodel.
%
%The proof that $\ECBVToFreyd$ is well defined on objects is similar. We must show that if $\stateobj$ is a comodel of $E$, then 
%$(\VCat, \KlCat{\VCat}{\CCat}{\stateobj}, J_{\stateobj})$ models $E$.
%The interpretation of effect constants is given by the correspondence between generic effects and state access operations of 
%Theorem~\ref{thm:stateaccess-geneff-algop}. To prove that this model structure satisfies the equations of $E$, one 
%can prove that for any $\pj{\Gamma}{M}{\STA}$ of the fragment
%of Figure~\ref{figure:effecttheories:typing} the equation $\den{\CBVtoEECP{M}} = \den{M}$ holds, where the 
%brackets on the left hand side refer to the interpretation of {\ECBV} in $(\VCat, \CCat, \stateobj)$ and the right hand side refers to 
%the interpretation in $(\VCat, \KlCat{\VCat}{\CCat}{\stateobj}, J_{\stateobj})$. 
\end{proofsketch}
We end this section by stating the coreflection theorem for models of effect theories.

\begin{thm} \label{thm:adj:eff}
The 2-functor $\FreydToECBV \co \Freydtheory{E} \to \CATECBVtheory{E}$ is left biadjoint to $\ECBVToFreyd$, i.e., 
for any pair of objects $(\VCat, \CCat, J)$  and $(\VCat', \CCat', \stateobj)$ of $\Freydtheory{E}$ and $\CATECBVtheory{E}$ respectively, there is an equivalence of categories
\[\CATECBVtheory{E}(\FreydToECBV(\VCat, \CCat, J),(\VCat', \CCat',
\stateobj)) 
\ \simeq \ 
\Freydtheory{E}((\VCat, \CCat,J),\ECBVToFreyd(\VCat', \CCat', \stateobj))
\]
natural in $(\VCat, \CCat, J)$  and $(\VCat', \CCat', \stateobj)$. 
Moreover, the unit of the adjunction $\eta \co \id_{\CATECBVtheory{E}} \to \ECBVToFreyd\circ\FreydToECBV$ is an isomorphism.
\end{thm}

%%% Local Variables: 
%%% mode: latex
%%% TeX-master: "mogelberg-staton"
%%% End: 

\section{Relationship with Atkey's parameterized monads}
%The enriched call-by-value calculus (\S\ref{sec:ecbv}) is the
%centre-piece of our study of the linear-use state-passing translation.
%To conclude this article we relate the enriched call-by-value calculus
%with two other developments: Atkey's parameterized
%monads~(\cite{a-parammonad}, this section) and the enriched effect
%calculus~(\cite{Mogelberg:CSL:09,EEC:journal}, \S\ref{sec:ecbvtoeec}).

Atkey's work on parameterized monads~\cite{a-parammonad},
has proven relevant to 
functional programming~(e.g.~\cite[\S5.2]{kps-funwithtypefunctions}).
In this section we show that parameterized monads are essentially the same as 
enriched models.

%(A historical note: while our work was based on the enriched effect
%calculus from the beginning, we only noticed the tight relationship
%with parameterized monads relatively recently.)
\label{sec:parammonad}
\newcommand{\SCat}{\mathcal S}
\renewcommand{\ACat}{\mathcal A}
\newcommand{\BCat}{\mathcal B}

Recall that, in general category theory, if a functor 
$F: \ACat\times \SCat\to\BCat$ 
is such that $F(-,S):\ACat\to\BCat$ has a 
right adjoint $G(S,-):\BCat\to\ACat$ for each $S$,
then these right adjoints together form a 
functor $G:\opcat\SCat\times \BCat\to\ACat$
called the parameterized right adjoint.
Atkey has carried out a study of a generalized form of 
monad 
that arises from parameterized adjunctions:
the functor 
$G(-_1,F(-_2,-_3)):\opcat\SCat\times \ACat\times\SCat\to\ACat$
is called a parameterized monad.
Thus a parameterized monad is a functor 
$T:\opcat\SCat\times\ACat\times \SCat\to \ACat$ 
together with extranatural families of morphisms
\begin{align*}
\eta_{S,A}:&A\to T(S,A,S)
\\
\mu_{S_1,S_2,S_3,A}:&
T(S_1,T(S_2,A,S_3),S_2)
\to
T(S_1,A,S_3)
\end{align*}
satisfying monad laws.
A first example of a parameterized monad is the parameterized state
monad on the category of sets:
$T(S_1,A,S_2)\defeq [S_1\Rightarrow A\times S_2]$.

Every enriched model $(\VCat,\CCat)$ 
contains a parameterized adjunction, 
since $\CCat(-_1,-_2):\opcat\CCat\times \CCat\to\VCat$
is by definition a parameterized right adjoint 
for $\ltensor{(-_1)}{(-_2)} :\VCat\times\CCat\to\CCat$.

Conversely, in the theory of parameterized monads,
the following Kleisli construction \cite[Prop.~1]{a-parammonad}
plays a key role.
Given a parameterized monad $T:\opcat \SCat\times \ACat\times \SCat\to \ACat$, 
the objects of the Kleisli category are pairs 
$(A,S)$ of an object of $\ACat$ and an object of $\SCat$,
and a morphism
$(A,S)\to (A',S')$ is a morphism
$A\to T(S,A',S')$ in~$\VCat$.
This is a first step towards building an enriched model from a parameterized
monad.

Plain parameterized monads are not especially relevant to 
the theory of programming languages,
just as plain monads are not very relevant.
In his study of parameterized monads, 
Atkey focuses on
\emph{strong parameterized monads with Kleisli exponentials}~\cite[\S 2.4.1]{a-parammonad}.
He uses these to provide
semantics for a `command calculus', 
which is a term language closely related to our
basic enriched call-by-value calculus~(Figure~\ref{figure:effects:typing}).
\begin{prop} Let $\VCat$ be a category with finite products.
The following data are equivalent.
\begin{enumerate}
\item A strong parameterized monad on $\VCat$ with Kleisli exponentials,
taking parameters in a category $\SCat$ \cite[\S 2.4.1]{a-parammonad}.
\item A category $\CCat$ enriched in $\VCat$ with copowers 
(i.e., an enriched model -- \S\ref{sec:adjmodels}) 
with a chosen subcategory $\SCat$ of $\CCat$
such that every object $\algX$ of $\CCat$ is of the 
form $\algX= \ltensor A \stateobj$ for $A$ in $\VCat$ and $\stateobj$ in $\SCat$.
\end{enumerate}
\end{prop}
\begin{proofnotes}
Given a strong parameterized monad 
$T:\opcat\SCat\times\VCat\times\SCat\to\VCat$, we 
let $\CCat$ be the Kleisli category for~$T$, as above.
The pair $(\VCat,\CCat)$ forms an enriched model
with $\ltensor A{(B,\stateobj)}\defeq (A\times B,\stateobj)$:
this is a rephrasing of
what it means for $T$ to be strong and 
have Kleisli exponentials. 
Moreover $\SCat$ can be identified with a subcategory of $\CCat$ 
whose objects are of the form $(1,\stateobj)$
and whose morphisms are induced by the morphisms in $\SCat$.

Conversely, suppose we are given an enriched model $(\VCat,\CCat)$
and a chosen subcategory~$\SCat$ of $\CCat$. 
We define a parameterized monad
$T:\opcat\SCat\times\VCat\times \SCat\to \VCat$ 
by 
\[
T(\stateobj,A,\stateobj') \defeq \CCat(\stateobj,\ltensor A \stateobj')\quad\text.\]
It is routine to check that the two constructions are mutually inverse,
up-to equivalence of categories. 
\end{proofnotes}

\section{Relationship with the enriched effect calculus}
\label{sec:ecbvtoeec}

Our enriched call-by-value calculus (ECBV) is a fragment of the enriched
effect calculus~(EEC, \cite{Mogelberg:CSL:09,EEC:journal}) which was
designed to analyze linear usage in effectful computation.  
We now show that ECBV is not only a syntactic fragment of EEC: every model of
the enriched call-by-value calculus embeds in a model of the enriched
effect calculus.

The enriched effect calculus extends the enriched
call-by-value calculus that we introduced in Section~\ref{sec:ecbv}
with some type constructions:
\[\begin{array}{r@{}lclcl}
\VA,\VB~ & ::=\, \alpha \,\mid\,1\,\mid\,\VA\times \VB\,\mid\,\CA\lpop\CB\,
&\mid&
0\,\mid\,\VA+\VB
&\mid&\VA\to\VB \,\mid\, \CconstA \,\mid  \,   \algzero\,\mid\,\CA \algplus \CB \,\mid\,\ltensortype{\VA}{\CB}\,\mid\,{!}\VA \\[5pt]
\CA,\CB~  & ::=\,
\underbrace{\CconstA \,\mid\,\ltensortype \VA\CB\hspace{2.1cm}}_%
{\text{\qquad \emph{Figure~\ref{figure:effects:typing}}}}
&
\mid&\underbrace{\algzero\,\mid\,\CA\oplus\CB}_%
{\text{\quad\emph{Figure~\ref{figure:effects:sumtyping}}}}
&\mid&
\underbrace{1 \,\mid\,\CA \times \CB \,\mid\,\VA\to\CB\,\mid\,{!}\VA
\enspace .\hspace{2cm}}%
_{\text{\qquad \emph{Full EEC~\cite{Mogelberg:CSL:09}}}} 
\end{array}\]
The additional types are: products $(\CA\times \CB)$ 
and powers $(\VA\to\CB)$ of computation types;
an operation to coerce a value type $\VA$ into a computation type ${!}\VA$;
a space of pure functions $(\VA\to\VB)$ between value types;
and an implicit inclusion of computation types as value types.

These additional types have been used to
describe other aspects of effectful computation.
We briefly considered a linear-use CPS translation
in Section~\ref{sec:cps}, based on \cite{Mogelberg:fossacs:10,EEC:LCPS:journal},
for which we needed products and powers of computation types.
Egger et al.~\cite{Mogelberg:CSL:09} also describe
monadic call-by-name and call-by-value 
interpretations, for which they use the coercion of 
value types into computation types and the implicit
inclusion of computation types in value types.

The additional types of EEC do not affect the full completeness
of the linear state-passing translation 
(Thm.~\ref{thm:full:faithful}), for the following reason.
In Theorem~\ref{thm:embedding:ECBV:EEC}
we show that every model of ECBV embeds in a model of EEC;
conservativity of $\EEC$ over {\ECBV} 
then 
follows from
a strong normalisation result for $\EEC$~\cite{EEC:journal}.
Thus the linear-use state-passing translation
of Section~\ref{sec:translation}
can be understood as a fully complete translation into $\EEC$.

\begin{defi}
\label{def:eecmodel}
  A \emph{closed enriched model}
  is a pair of categories 
  $(\VCat,\CCat)$ 
  such that $\VCat$ is cartesian closed 
  with coproducts
  and $\CCat$ is $\VCat$-enriched with 
  powers and copowers and products and coproducts.

  \emph{A model of EEC} $(\VCat,\CCat,F,U)$ \cite{Mogelberg:CSL:09}
  is 
  a closed enriched model
  $(\VCat,\CCat)$ 
  together with a
  $\VCat$-enriched adjunction $F\dashv U\colon \CCat\to\VCat$.
\end{defi}
We refer to~\cite{Mogelberg:CSL:09} 
for the term calculus and interpretation of EEC in EEC models. 
Here, we will 
analyze how the notion of EEC model compares to
the other notions of model that we have discussed so far.
One simple observation is that 
every closed enriched model 
is a distributive enriched model
in the sense of Definition~\ref{def:distrenrichedmodel}.
Another observation is
that
the adjunction part of an EEC model can be 
equivalently given by a `state' object of $\CCat$
(see \cite[Proof of Thm.~4]{Mogelberg:CSL:09} and 
Section~\ref{sec:kleisli-to-enriched-models}).
\hide{the enriched adjunction $F\vdash U$ is uniquely 
determined by the state object $F(1)$ in $\CCat$ 
(Proposition~\ref{prop:adjmodels}). }
\begin{thm} 
\label{thm:embedding:ECBV:EEC}
Every enriched model 
embeds in a closed enriched model.
\end{thm}
\newcommand{\Pshf}[1]{\widehat{#1}}%{[\opcat{#1},\Set]}%
\newcommand{\yoneda}{\mathbf{y}}%{[\opcat{#1},\Set]}%
\newcommand{\WeakMonoid}{\mathcal M}
\begin{proof}
The difference between
enriched models
and closed enriched models is that 
in an enriched model $(\VCat,\CCat)$ the value category 
$\VCat$ need not be cartesian closed nor have coproducts,
and the computation category $\CCat$ need not have coproducts,
products and powers.

We use the Yoneda embedding to embed an 
enriched model in a closed enriched model
For any small category $\ACat$
we consider the category $\Pshf\ACat$ of contravariant functors,
$\opcat\ACat\to \Set$,
and natural transformations between them.
The Yoneda embedding $A\mapsto\ACat(-,A)$ 
is a functor $\yoneda_\ACat:\ACat\to \Pshf\ACat$
that exhibits $\Pshf\ACat$ as a cocompletion of $\ACat$.
That is: $\Pshf\ACat$ is cocomplete,
with colimits computed pointwise,
and for any other cocomplete category 
$\BCat$ and any functor $F\colon \ACat\to\BCat$
there is a colimit-preserving functor
$F_!:\Pshf\ACat\to\BCat$ given by 
${F_!(P)\defeq \colim((\yoneda_\ACat\downarrow P)\xrightarrow \pi\ACat
\xrightarrow F \BCat)}$, where
$(\yoneda_\ACat\downarrow P)$ is the category of elements of $P$;
this colimit-preserving functor is essentially unique
such that $F\cong F_!\cdot\yoneda_\ACat$.

Let $(\VCat,\CCat)$ 
be an enriched model. 
We will show that $(\Pshf\VCat,\Pshf\CCat)$ is a closed enriched model,
and that $(\VCat,\CCat)$ embeds in it as an enriched model.

We proceed by considering the following 2-categorical situation.
Because the construction $\Pshf{(-)}$ is a free cocompletion,
it can be understood as a weak 2-functor
from the 2-category $\CAT$ of small categories, functors and
 natural transformations 
to the 2-category $\CocompCat$ of categories with all colimits, 
colimit-preserving functors 
and natural transformations.
\hide{This is because $\Pshf{(-)}$ is a free cocompletion.
The construction takes a functor ${F\colon \ACat\to\BCat}$ to 
its left Kan extension along the Yoneda embedding,
${\Pshf F:\Pshf\ACat\to\Pshf\BCat}$,
which can be understood via a coend formula
or as a colimit:
\[
(\Pshf F(P))(B)\defeq \int^{A}P(A)\times \BCat(B,A)
\ \cong \ 
\colim((\yoneda_{\ACat}\downarrow P)\xrightarrow \pi
\ACat\xrightarrow F\BCat\xrightarrow {\BCat(B,-)} \Set)\text.
\]}

In fact it is necessary to be slightly more general than this:
we will understand $\CAT$ and $\CocompCat$ as 2-multicategories.
Recall that a 2-multicategory is a $\Cat$-enriched multicategory.
So it is like a 2-category except that 
the domains of the 1-cells are sequences of objects.
\begin{itemize}
\item The 2-multicategory $\CAT$ is defined as follows.
The 0-cells are small categories 
with finite coproducts.
The 1-cells 
${F\colon (\CatA_1,\dots,\CatA_n)\to \CatB}$ in $\CAT$ are
functors in $n$ arguments, i.e. functors 
${F\colon\CatA_1\times \dots\times \CatA_n\to \CatB}$.
The 2-cells are natural transformations.
\item 
The objects of the 2-multicategory $\CocompCat$ are 
categories with all colimits. 
The 1-cells 
$F:(\CatA_1,\dots,\CatA_n)\to \CatB$ in $\CocompCat$ are
functors $F:\CatA_1\times \dots\times \CatA_n\to \CatB$ 
that preserve colimits in each argument,
i.e.~that for fixed $A_1\in\CatA_1$, \dots, $A_n\in\CatA_n$
and for ${1\leq i\leq n}$,
the functor 
${F(A_1,\dots,-_i,\dots A_n):\CatA_i\to \CatB}$ 
preserves colimits.
The 2-cells are natural transformations.
\item
The construction
$\Pshf{(-)}$ extends to a weak morphism of 2-multicategories
from ${\CAT}$ to ${\CocompCat}$.
A 1-cell
$F:(\CatA_1,\dots,\CatA_n)\to \CatB$ in $\CAT$
is extended to 
a 1-cell in $\CocompCat$, i.e.~a functor 
${F_!:\Pshf{\CatA_1}\times\dots\times \Pshf{\CatA_n}\to \Pshf{\CatB}}$
which preserves colimits in each argument.
This construction is done by iteratively applying the 
following idea.
If an $n$-ary functor 
${G\colon (\ACat_1,\dots,\ACat_k,\dots,\ACat_n)\to \BCat}$ is 
 such that~$\BCat$ 
is cocomplete, $\ACat_1\dots\ACat_{k-1}$ are cocomplete, 
$G$ preserves colimits in each of the first ${(k-1)}$ arguments,
and $\ACat_k$ is small, 
then there is an $n$-ary functor 
$G_{!k}\colon (\ACat_1,\dots,\Pshf{\ACat_k},\dots,\ACat_n)\to \BCat$ 
that preserves colimits in each of the first $k$ arguments
such that 
$G\cong G_{!k}\cdot (\ACat_1,\dots,\yoneda_\ACat,\dots,\ACat_n)$.
This is because the $n$-ary functor $G:(\ACat_1,\dots,\ACat_n)\to \BCat$ 
can be curried to a functor
\[\ACat_k\to \CocompCat(\ACat_1,\dots,\ACat_{k-1};\CAT(\ACat_{k+1},\dots,\ACat_n;\BCat))\]
whose codomain is cocomplete,
and which can thus be extended to a colimit-preserving
functor 
using the universal property of $\hat{\ACat_k}$:
\[\hat{\ACat_k}\to \CocompCat(\ACat_1,\dots,\ACat_{k-1};\CAT(\ACat_{k+1},\dots,\ACat_n;\BCat))\text;\]
this can be uncurried to give
$G_{!k}\colon (\ACat_1,\dots,\Pshf{\ACat_k},\dots,\ACat_n)\to \BCat$.
\hide{The functor $G_{!k}$ is given by
\[G_{!k}(A_1,\dots,A_{k-1},P,A_{k+1},\dots, A_n)=
\colim(\yoneda_{\ACat_k}\downarrow P\xrightarrow \pi 
\ACat_k\xrightarrow {G(A_1,\dots,A_{k-1},-,A_{k+1},\dots,A_n)}\BCat)\text.\]}
Ultimately, the extension 
$F_!:\Pshf{\CatA_1}\times\dots\times \Pshf{\CatA_n}\to \Pshf{\CatB}$
satisfies the following coend formula:
\[
F_!(P_1,\dots,P_n)(B)\cong\int^{A_1,\dots,A_n}P_1(A_1)\times\dots\times P_n(A_n)\times
\CatB(B,F(A_1,\dots,A_n))
\]
\item
Recall the following consequence of the special adjoint functor theorem:
a morphism
$F:(\Pshf{\CatA_1},\dots,\Pshf{\CatA_n})\to \CatB$
in $\CocompCat$ 
can be equivalently described as a functor 
that has a right adjoint 
in each argument,
i.e.~a right adjoint for each functor 
$F(P_1,\dots,-_i,\dots,P_n):\Pshf{\CatA_i}\to \CatB$. 
\item
Aside from size issues, there is a forgetful morphism of 2-multicategories
$\CAT\to\CocompCat$
and $\Pshf{(-)}$ is left biadjoint to it.
The Yoneda embedding is the unit for this adjunction.
\end{itemize}
With the general situation explained, the proof of 
Theorem~\ref{thm:embedding:ECBV:EEC} 
is straightforward.
We begin by considering 
the evident notion of `weak monoid' in a 2-multicategory%
~$\TwoMultiCat$.
This comprises an object $\WeakMonoid$ of $\TwoMultiCat$ and
two 1-cells: $m:(\WeakMonoid,\WeakMonoid)\to\WeakMonoid$ 
and $e:()\to\WeakMonoid$,
with three coherence 2-isomorphisms.
A morphism of 2-multicategories 
$\TwoMultiCat\to\TwoMultiCat'$ takes weak monoids in~$\TwoMultiCat$
to weak monoids in~$\TwoMultiCat'$.
In particular a monoidal category 
is a weak monoid in $\CAT$,
and the construction $\Pshf{(-)}$ takes 
it to a weak monoid in $\CocompCat$, which is
a cocomplete biclosed monoidal category. 
The Yoneda embedding 
$\WeakMonoid\to\Pshf\WeakMonoid$ 
preserves the weak monoid structure.
This is Day's convolution construction~\cite{day-closed,ik-univday}.

In particular, the value category $\VCat$ of our enriched model 
has products
and this exhibits it as a monoidal category.
It follows that $\Pshf\VCat$ is cartesian closed 
and that the Yoneda embedding ${\VCat\to\Pshf\VCat}$ preserves the 
product structure in $\VCat$.
Given a weak monoid $\WeakMonoid$ in a 2-multicategory $\TwoMultiCat$,
we consider the evident notion of weak action for $\WeakMonoid$:
an object $\CatA$ of $\TwoMultiCat$ and
a 1-cell $(\WeakMonoid,\CatA)\to\CatA$ satisfying the laws of 
monoid actions up-to coherent isomorphism.
A morphism of 2-multicategories takes weak monoid actions to weak monoid
actions.
In particular, given a monoidal category~$\WeakMonoid$,
an action of $\WeakMonoid$ on another category
$\CatA$ 
induces an enrichment of $\Pshf\CatA$ in $\Pshf\WeakMonoid$ 
with powers and copowers.
The Yoneda embedding $\CatA\to\Pshf\CatA$ 
preserves the monoidal action.
Moreover since it is 2-natural it preserves any enrichment or 
powers that already exist in $\CatA$.

In particular, in our enriched model,
$\VCat$ acts on $\CCat$ and so
$\Pshf\CCat$ is enriched in $\Pshf\VCat$ with powers and copowers,
and the Yoneda embedding $\CCat\to \Pshf\CCat$
is enriched in $\VCat$  and preserves copowers.
\end{proof}
The crux of the proof
is that the Yoneda embedding 
adds closed structure --- cartesian closed structure and 
powers --- while preserving the other structure.
Although the Yoneda embedding does not \emph{freely}
add the closed structure, it is considerably simpler that the free
closure. This is one reason why Yoneda
embeddings are a common technique in semantics.
For instance,
the enriched Yoneda embedding is used by Egger et al.~\cite{Mogelberg:CSL:09}
to show that Levy's call-by-push-value embeds in the enriched 
effect calculus (but the \emph{enriched} Yoneda embedding 
is not appropriate in our proof because it does not preserve copowers).

Theorem~\ref{thm:embedding:ECBV:EEC}
explains that EEC is conservative over ECBV,
but it neglects sum types.
Sum types play an important role in the study of 
generic effects and state access operations (\S\ref{sec:modelsoftheories}).
We now show that EEC is conservative over ECBV with sum types.
\begin{prop} 
\label{prop:embedding:ECBV:EEC:sums}
Every distributive enriched model (\S\ref{def:distrenrichedmodel})
embeds in a closed enriched model (Def.~\ref{def:eecmodel}).
\end{prop}
\newcommand{\Sind}[1]{\mathrm{FP}(\opcat{#1},\Set)}%
\newcommand{\Sindop}[1]{\mathrm{FP}(#1,\Set)}%
\begin{proof}
Our proof of Proposition~\ref{prop:embedding:ECBV:EEC:sums}
follows the same outline as our proof of Theorem~\ref{thm:embedding:ECBV:EEC}.
We must modify that proof because the Yoneda embedding
$\yoneda_{\ACat}:\ACat\to\Pshf\ACat$ does not preserve coproducts.
We use the following
variation on the Yoneda embedding. 
For any category $\ACat$ with finite coproducts,
let $\Sind\ACat$ be the category of finite-product-preserving functors
$\opcat\ACat\to\Set$ and natural transformations between them.
Assuming $\ACat$ has coproducts,
the category $\Sind\ACat$ has coproducts too.
The Yoneda embedding $A\mapsto\ACat(-,A)$ 
is a functor $\ACat\to \Sind\ACat$
which preserves coproducts.
In fact, the Yoneda embedding exhibits 
$\Sind\ACat$ as 
the cocompletion of $\ACat$ 
as a category with finite coproducts.
The category $\Sind\ACat$ is cocomplete 
(although not all colimits are computed pointwise),
and for any coproduct-preserving functor $F\colon \ACat\to\BCat$ 
into a cocomplete category~$\BCat$
there is a colimit-preserving functor $F_!:\Sind\ACat\to \BCat$
given by 
$F_!(P)\defeq \colim((\yoneda_\ACat\downarrow P)\xrightarrow \pi\ACat\xrightarrow F \BCat)$; this colimit-preserving functor is essentially 
unique such that $F\cong F_!\cdot \yoneda_\ACat$
(see e.g.~%
\cite[Thms~5.86,~6.11]{Kelly:Book},
\cite{PowerRobinson:Premonoidal},
\cite{Fiore:enrichment}).
Since a distributive enriched model has coproducts,
this is the right variation of the Yoneda embedding to use.

We now mimic the proof of Theorem~\ref{thm:embedding:ECBV:EEC},
replacing the cocompletion construction
$\Pshf{(-)}$ with 
the cocompletion $\Sind{(-)}$
of a category with coproducts.
Consider the 2-multicategory $\CoprodCat$:
the {0-cells} are small categories 
with finite coproducts;
the 1-cells 
${F\colon (\CatA_1,\dots,\CatA_n)\to \CatB}$ are
functors 
${F\colon\CatA_1\times \dots\times \CatA_n\to \CatB}$ 
that preserve coproducts in each argument;
the 2-cells are natural transformations.
The construction
$\Sind{(-)}$ extends to a morphism of 2-multicategories
from ${\CoprodCat}$ to ${\CocompCat}$.
By the special adjoint functor theorem,
a 
morphism ${(\Sind{\CatA_1}, \dots, \Sind{\CatA_n})\to \CatB}$ 
in $\CocompCat$ is a functor that has a right adjoint  in each argument.
\hide{A 1-cell 
$F:(\CatA_1,\dots,\CatA_n)\to \CatB$ in $\CoprodCat$
is extended to 
a functor
$F_!:\Sind{\CatA_1}\times\dots\times \Sind{\CatA_n}\to \Sind{\CatB}$.
As in the proof of Theorem~\ref{thm:embedding:ECBV:EEC},
it can be expressed in terms of coends in $\Set$:
\[
F_!(P_1,\dots,P_n)(B)\defeq\int^{A_1,\dots,A_n}P_1(A_1)\times\dots\times P_n(A_n)\times
\CatB(B,F(A_1,\dots,A_n))
\]
There are two things to check:
\begin{itemize}
\item The functor $F_!(P_1,\dots,P_n):\opcat\CatB\to\Set$ preserves products.

\item
Each functor
$F_!(P_1,\dots,-_i,\dots,P_n):\Pshf{\CatA_i}\to\Pshf\CatB$ 
preserves colimits as a functor
$\Sind{\CatA_i}\to\Sind{\CatB}$.
\end{itemize}
Recall the following consequence of the special adjoint functor theorem:
a morphism
$F:\Sind{\CatA_1},\dots,\Sind{\CatA_n}\to \CatB$
in $\CocompCat$ 
can be equivalently described as a functor 
that has a right adjoint 
in each argument,
i.e. a right adjoint for each functor 
$F(P_1,\dots,-_i,\dots,P_n):\Sind{\CatA_i}\to \CatB$. 
}

A weak monoid~$\WeakMonoid$ in $\CoprodCat$ 
is  a distributive monoidal category,
i.e.,~a monoidal category with coproducts
such that the tensor preserves coproducts in each argument.
The construction $\Sind-$ takes 
it to a weak monoid in $\CocompCat$, which is
a cocomplete biclosed monoidal category. 
The Yoneda embedding 
$\WeakMonoid\to\Sind\WeakMonoid$ 
preserves the weak monoid structure and coproducts.

In particular, if $(\VCat,\CCat)$ is a
distributive enriched model then $\VCat$ has distributive products
and this exhibits it as  a distributive monoidal category.
It follows that $\Sind\VCat$ is cartesian closed with coproducts
and that the Yoneda embedding ${\VCat\to\Sind\VCat}$ preserves the 
coproduct and product structure in $\VCat$.

An action of a weak monoid
in $\CoprodCat$ is the same thing as 
a distributive action in the sense of Section~\ref{sec:sums:enriched}.
Given a distributive monoidal category~$\WeakMonoid$,
a distributive action of $\WeakMonoid$ on a category
$\CatA$ with finite coproducts
induces an enrichment of $\Sind\CatA$ in $\Sind\WeakMonoid$ 
with powers and copowers.
The Yoneda embedding $\CatA\to\Sind\CatA$ 
preserves coproducts and the monoidal action.
Moreover since it is 2-natural it preserves any enrichment or 
powers that already exist in $\CatA$.

In particular, if $(\VCat,\CCat)$ is a distributive enriched model
then $\VCat$ acts on $\CCat$, and so
$\Sind\CCat$ is enriched in $\Sind\VCat$ with powers and copowers,
and the Yoneda embedding $\CCat\to \Sind\CCat$
is enriched in $\VCat$ and preserves coproducts and copowers.

\hide{
We now show that $\Sind\CCat$ is enriched in $\Sind\VCat$ 
with powers and copowers.
Recall the construction of Day~\cite{day-closed}, 
which induces a monoidal biclosed 
structure on $\hat \ACat$ (the category of functors $\opcat\ACat\to\Set$) 
for every monoidal 
structure on a category~$\ACat$.
We develop this in two ways.
First, the monoidal action of $\VCat$ on $\CCat$
induces a monoidal action 
${\hat\VCat\times\hat\CCat\to\hat\CCat}$
which has right adjoints in both arguments.
Secondly, the monoidal action of $\hat\VCat$ on $\hat\CCat$
restricts
to a monoidal action 
$\Sind\VCat\times\Sind\CCat\to\Sind\CCat$
and the right adjoints restrict too.
Thus $\Sind\CCat$ is enriched in $\Sind\VCat$ with 
copowers and powers.

Finally, 
we need an enriched adjunction between $\Sind\VCat$ and $\Sind\CCat$.
The Yoneda embedding takes
the chosen object $\stateobj$ of $\CCat$
to the presheaf $\CCat(-,\stateobj)$ in $\Sind\CCat$,
and this induces 
an enriched adjunction 
by
Proposition~\ref{prop:adjmodels}.
The left adjoint takes
a presheaf $P$ in $\Sind\VCat$ to the presheaf
${\int^A P(A)\times \CCat(-,\ltensor A \stateobj)}$
in $\Sind\CCat$.
The right adjoint takes a presheaf 
$Q$ in $\Sind\CCat$ to the presheaf
$Q(\ltensor -\stateobj)$ in $\Sind\VCat$ .
}
\end{proof}
\newcommand{\FinSet}{\mathbf{Set}_\mathrm{f}}
\newcommand{\Lawv}{\mathbb{T}}
The construction in this proof is related to the following
natural situation.
Let $\FinSet$ be the category of finite sets, and let $\Lawv$ 
be a Lawvere theory.
Then $(\FinSet,\opcat\Lawv)$ is almost an enriched model, except 
that the category $\opcat\Lawv$ is typically not $\FinSet$-enriched.
Nonetheless, our construction applied to $(\FinSet,\opcat\Lawv)$ yields 
the  basic motivating example of an EEC model:
$\Sind\FinSet$ is the category of sets
(since $\FinSet$ is the free category with finite coproducts on one generator)
and $\Sindop{\Lawv}$ is the category of algebras of the Lawvere theory.
(See also \cite[Thm.~38]{Power:GenericModels:06}.)

\bibliographystyle{plain} 
\bibliography{lin-state}

\begin{thebibliography}{10}

\bibitem{DBLP:journals/jfp/AchtenP95}
Peter Achten and Marinus~J. Plasmeijer.
\newblock The ins and outs of {C}lean {I}/{O}.
\newblock {\em J. Funct. Program.}, 5(1):81--110, 1995.

\bibitem{on-mean-values}
J~Acz\'el.
\newblock On mean values.
\newblock {\em Bull. Amer. Math. Soc.}, 54(4):392--400, 1948.

\bibitem{a-parammonad}
Robert Atkey.
\newblock Parameterised notions of computation.
\newblock {\em J. Funct. Program.}, 19(3--4):335--376, 2009.

\bibitem{Berdine:02}
Josh Berdine, Peter~W. O'Hearn, Uday~S. Reddy, and Hayo Thielecke.
\newblock Linear continuation-passing.
\newblock {\em Higher-Order and Symbolic Computation}, 15(2-3):181--208, 2002.

\bibitem{day-closed}
Brian Day.
\newblock On closed categories of functors.
\newblock In {\em Lect. Notes Math. 137}, pages 1--38. Springer, 1970.

\bibitem{dk-alggraph}
Eduardo~J Dubuc and G.M Kelly.
\newblock A presentation of topoi as algebraic relative to categories or
  graphs.
\newblock {\em J. Algebra}, 81(2):420 -- 433, 1983.

\bibitem{EEC:journal}
J.~Egger, R.E. M{\o}gelberg, and A.~Simpson.
\newblock The enriched effect calculus: syntax and semantics.
\newblock To appear in Journal of Logic and Computation.

\bibitem{Mogelberg:CSL:09}
J.~Egger, R.E. M{\o}gelberg, and A.~Simpson.
\newblock Enriching an effect calculus with linear types.
\newblock In {\em {CSL}'09}, pages 240--254. Springer, 2009.

\bibitem{Mogelberg:fossacs:10}
J.~Egger, R.E. M{\o}gelberg, and A.~Simpson.
\newblock Linearly-used continuations in the enriched effect calculus.
\newblock In {\em Proc.~FOSSACS'10}, volume 6014, pages 18--32. Springer, 2010.

\bibitem{EEC:LCPS:journal}
J.~Egger, R.E. M{\o}gelberg, and A.~Simpson.
\newblock Linear-use {CPS} translations in the enriched effect calculus.
\newblock {\em Logical Methods in Computer Science}, 8(4), 2012.

\bibitem{Fiore:enrichment}
Marcelo~P. Fiore.
\newblock Enrichment and representation theorems for categories of domains and
  continuous functions.
\newblock Unpublished manuscript, March 1996.

\bibitem{DBLP:conf/fossacs/Fuhrmann02}
Carsten F{\"u}hrmann.
\newblock Varieties of effects.
\newblock In {\em Proc.~FOSSACS'02}, pages 144--158, 2002.

\bibitem{DBLP:journals/tcs/Girard87}
Jean-Yves Girard.
\newblock Linear logic.
\newblock {\em Theor. Comput. Sci.}, 50:1--102, 1987.

\bibitem{GordonPower:EnrichmentThroughVariation}
R.~Gordon and A.J. Power.
\newblock Enrichment through variation.
\newblock {\em J. Pure Appl. Algebra}, 120:167--185, 1997.

\bibitem{Hasegawa:Flops:02}
M.\ Hasegawa.
\newblock Linearly used effects: Monadic and {CPS} transformations into the
  linear lambda calculus.
\newblock In {\em Proc.\ 6th International Symposium on Functional and Logic
  Programming (FLOPS)}, volume 2441 of {\em LNCS}, pages 167--182. Springer,
  2002.

\bibitem{h-probdom}
Reinhold Heckmann.
\newblock Probabilistic domains.
\newblock In {\em Proc.~Trees in Algebra and Programming -- CAAP'94}, volume
  787 of {\em Lecture Notes in Computer Science}, pages 142--156. Springer,
  1994.

\bibitem{ik-univday}
Geun~Bin Im and G~M Kelly.
\newblock A universal property of the convolution monoidal structure.
\newblock {\em J. Pure Appl. Alg.}, 43:75--88, 1986.

\bibitem{JanelidzeKelly:actions}
G.~Janelidze and G.~M. Kelly.
\newblock A note on actions of a monoidal category.
\newblock {\em Theory Appl.\ of Categ.}, 9(4):61--91, 2001.

\bibitem{jeffrey-premonoidal-graphics}
Alan Jeffrey.
\newblock Premonoidal categories and a graphical view of programs.
\newblock Available at
  \url{ftp://outside.cs.bell-labs.com/who/ajeffrey/papers/premonA.pdf}, 1997.

\bibitem{kelly-enrichedadj}
G.~M. Kelly.
\newblock Adjunction for enriched categories.
\newblock In {\em Lect. Notes Math. 106}, pages 166--177. Springer, 1969.

\bibitem{Kelly:Book}
G.~M. Kelly.
\newblock {\em Basic Concepts of Enriched Category Theory}.
\newblock Cambridge University Press, 1982.

\bibitem{k-twocatlimits}
G~M Kelly.
\newblock Elementary observations on 2-categorical limits.
\newblock {\em Bull. Austral. Math. Soc.}, 39(2):301–317, 1989.

\bibitem{kp-enrichedalgtheories}
G~M Kelly and A~J Power.
\newblock Adjunctions whose counits are coequalizers, and presentations of
  finitary enriched monads.
\newblock {\em J. Pure Appl. Algebra}, 89:163--179, 1993.

\bibitem{kps-funwithtypefunctions}
Oleg Kiselyov, Simon {Peyton~Jones}, and C.~Shan.
\newblock Fun with type functions.
\newblock In {\em Reflections on the Work of {C}.{A}.{R}. {H}oare}, pages
  301--331. Springer, 2010.

\bibitem{bilinearity-and-cc-monads}
Anders Kock.
\newblock Bilinearity and cartesian closed monads.
\newblock {\em Math. Scand.}, 29:161--174, 1971.

\bibitem{Levy:book}
P.~B. Levy.
\newblock {\em Call By Push Value}.
\newblock Kluwer, December 2003.

\bibitem{Levy:03}
P.B. Levy, J.~Power, and H.~Thielecke.
\newblock Modelling environments in call-by-value programming languages.
\newblock {\em Inform.\ and Comput.}, 185, 2003.

\bibitem{lhj-monad-transformers}
Sheng Liang, Paul Hudak, and Mark~P. Jones.
\newblock Monad transformers and modular interpreters.
\newblock In {\em Proc.~POPL 1995}, pages 333--343, 1995.

\bibitem{mesablishvili-state}
Bachuki Mesablishvili.
\newblock Monads of effective descent type and comonadicity.
\newblock {\em Theory and Applications of Categories}, 16(1):1--45, 2006.

\bibitem{metayer-state}
F~M\'etayer.
\newblock State monads and their algebras.
\newblock arXiv:math/0407251v1, 2004.

\bibitem{Moggi:89}
E.~Moggi.
\newblock Computational lambda-calculus and monads.
\newblock In {\em Proceedings of the 4th Annual Symposium on Logic in Computer
  Science}, pages 14--23, Asiloomar, CA, 1989. IEEE Computer Society Press.

\bibitem{OHearn:R:00}
P.~W. O'Hearn and J.~C. Reynolds.
\newblock From {A}lgol to polymorphic linear lambda-calculus.
\newblock {\em J. ACM}, 47(1):167--223, 2000.

\bibitem{pmp-mec}
Pierre-Marie P\'edrot.
\newblock On the semantics of the effect calculus.
\newblock Master's thesis, ENS Lyon, 2010.
\newblock Available at
  \url{http://perso.ens-lyon.fr/pierremarie.pedrot/reports/rapport-m1-hasegawa%
.pdf}.

\bibitem{Plotkin:65}
G.~Plotkin.
\newblock Call-by-name, call-by-value, and the $\lambda$-calculus.
\newblock {\em Theoret. Comp. Sci.}, 1:125--159, 1975.

\bibitem{Plotkin:Power:08}
G.~Plotkin and J.~Power.
\newblock Tensors of comodels and models for operational semantics.
\newblock In {\em Proc.~MFPS~XXIV}, volume 218 of {\em Electr. Notes Theor.
  Comput. Sci}, pages 295--311. Elsevier, 2008.

\bibitem{Plotkin:Pretnar:09}
G.~Plotkin and M.~Pretnar.
\newblock Handlers of algebraic effects.
\newblock In {\em Proc.~ESOP'09}, volume 5502 of {\em LNCS}, pages 80--94.
  Springer, 2009.

\bibitem{PlotkinPower:fossacs02}
G.~D. Plotkin and J.~Power.
\newblock Notions of computation determine monads.
\newblock In {\em Proc.~FOSSACS'02}, volume 2620. Springer, 2002.

\bibitem{Plotkin:Power:03}
G.~D. Plotkin and J.~Power.
\newblock Algebraic operations and generic effects.
\newblock {\em Appl. Categ. Structures}, 11(1):69--94, 2003.

\bibitem{p-varieties}
Gordon~D Plotkin.
\newblock Some varieties of equational logic.
\newblock In {\em Essays dedicated to Joseph A. Goguen}, volume 4060 of {\em
  Lect. Notes in Comput. Sci}, pages 150--156. Springer, 2006.

\bibitem{Power:Shkaravska:04}
A.~J. Power and O.~Shkaravska.
\newblock From comodels to coalgebras: State and arrays.
\newblock In {\em Proc.~CMCS'04}, volume 106 of {\em Electr. Notes Theor.
  Comput. Sci}, pages 297--314. Elsevier, 2004.

\bibitem{Power:GenericModels:06}
John Power.
\newblock Generic models for computational effects.
\newblock {\em Theoret. Comput. Sci.}, 364(2):254--269, 2006.

\bibitem{PowerRobinson:Premonoidal}
John Power and Edmund Robinson.
\newblock Premonoidal categories and notions of computation.
\newblock {\em Math. Structures Comput. Sci.}, 7(5):453--468, 1997.

\bibitem{Pretnar-Thesis}
Matija Pretnar.
\newblock {\em The logic and handling of algebraic effects}.
\newblock PhD thesis, School of Informatics, University of Edinburgh, 2010.

\bibitem{DBLP:conf/mfcs/Sieber94}
Kurt Sieber.
\newblock Full abstraction for the second order subset of an {A}lgol-like
  language.
\newblock In {\em Proc.~MFCS}, pages 608--617, 1994.

\bibitem{DBLP:journals/jlp/SomogyiHC96}
Zoltan Somogyi, Fergus Henderson, and Thomas~C. Conway.
\newblock The execution algorithm of {M}ercury, an efficient purely declarative
  logic programming language.
\newblock {\em J. Log. Program.}, 29(1-3):17--64, 1996.

\bibitem{Strachey72}
C.~Strachey.
\newblock The varieties of programming language.
\newblock In {\em Proc. International Computing Symposium}, pages 222--233.
  Cini Foundation, Venice, 1972.
\newblock {A}lso Tech. Monograph PRG-10, Univ. Oxford (1973).

\bibitem{hayo-phd}
Hayo Thielecke.
\newblock {\em Categorical structure of continuation passing style}.
\newblock PhD thesis, Univ.~Edinburgh, 1997.

\end{thebibliography}

\appendix

% !TEX root = Mogelberg-staton.tex

\section{Categories of models}
\label{app:cats:of:models}

%For the rest of this section we work with {\denrmodel}s as well as {\dKlmodel}s. Everything also works in the case where coproducts are not assumed. 

\subsection{The 2-category $\ENR$ of enriched models}
\label{app:CATENR}
We first define a notion of morphism of {\enrmodel} and
transformations between these. This gives a 2-category $\ENR$. Let
$(\VCat, \CCat)$ and $(\VCat', \CCat')$ be {\enrmodel}s
(Def.~\ref{def:enrichedmodel}). A morphism from $(\VCat, \CCat)$ to
$(\VCat', \CCat')$ is a triple $(F,G, \ltensoriso)$ such that $F \co
\VCat \to \VCat'$ and $G \co \CCat \to \CCat'$ are functors, and
$\ltensoriso$ is a natural family of isomorphisms
\[
\ltensoriso_{\SA,\algB} \co 
G(\ltensor\SA\algB) \iso 
\ltensor{F(\SA)}{G(\algB)}
\]
The following three conditions must be satisfied:
\begin{itemize}
\item $F$ preserves products (up to isomorphism).
\item The following two coherence diagrams commute
\begin{diagram}[LaTeXeqno] \label{eq:1:cell:coherence:cond:1}
G(\ltensor{1}{\algB}) & \rTo^{\ltensoriso} & \ltensor{F1}{G\algB} \\
\dTo^{\iso} && \dTo_{\ltensor{\bang{}}{\stateobj}} \\
{G\algB} & \lTo^{\iso} & \ltensor{1}{G\algB}
\end{diagram}
\begin{diagram}[LaTeXeqno] \label{eq:1:cell:coherence:cond:2}
G(\ltensor{(\SA\times \SB)}{\algC})  & \rTo^{{\iso}} & G(\ltensor{\SA}{\ltensor{\SB}{\algC}})  & \rTo^{\ltensoriso} & \ltensor{F\SA}{G(\ltensor{\SB}{\algC})} \\
\dTo^{\ltensoriso} &&&& \dTo_{\ltensor{F\SA}{\ltensoriso}} \\
\ltensor{F(\SA \times\SB)}{G\algC}  & \rTo^{\ltensor{\pair{F\pi_1}{F\pi_2}}{G\algC}} & \ltensor{(F\SA \times F\SB)}{G\algC}  & \rTo^{\iso} & \ltensor{F\SA}{\ltensor{F\SB}{G\algC}}
\end{diagram}
\item The mate of $\inv\ltensoriso$ is an isomorphism $F(\CHom{\algB}{\algC}) \iso \CHomp{G\algB}{G\algC}$
\end{itemize}
Recall that the mate of 
$\inv\ltensoriso$ is the adjoint correspondent of
\begin{diagram}
\ltensor{F(\CHom{\algB}{\algC})}{G\algB} & 
\rTo^{\inv\ltensoriso} & 
G(\ltensor{\CHom{\algB}{\algC}}{\algB}) & \rTo^{G(\ev)} & G(\algC)
\end{diagram}
where $\ev$ is the counit of the adjunction $\ltensor{(-)}{\algB} \dashv \CHom{\algB}{(-)}$.

The 2-cells between
morphisms
$(F,G,\ltensoriso),(F',G',\ltensoriso'):(\VCat,\CCat)\to(\VCat',\CCat')$
are natural isomorphisms $\VTwoCell \co F \iso F'$,  $\CTwoCell \co G \iso G'$ 
making the following diagram commute:
\begin{diagram}[LaTeXeqno] \label{eq:coherence:cond}
G(\ltensor{\SA}{\algB}) & \rTo^{\ltensoriso_{\SA,\algB}} & 
\ltensor{F\SA}{G\algB} \\ % & \GAP \text{and} \GAP & G\stateobj & \rTo^{\stateiso} & \stateobj  \\
\dTo^{\CTwoCell} && \dTo_{\ltensor{\VTwoCell}{\CTwoCell}} \\ %&& \dTo^{\CTwoCell} & \ruTo^{\stateiso'} \\
G'(\ltensor{\SA}{\algB}) & \rTo^{\ltensoriso'_{\SA,\algB}} & \ltensor{F'\SA}{G'\algB} % &&G'\stateobj
\end{diagram}
%and 
%\begin{diagram}
%G\stateobj & \rTo^{\stateiso} & \stateobj \\
%\dTo^{\CTwoCell} & \ruTo^{\stateiso'} \\
%G'\stateobj
%\end{diagram}

%\begin{lem} \label{lemma:coherence:mate}
%If $(\VTwoCell, \CTwoCell)$ is a 2-cell then the mate isomorphism $F(\CHom{\algB}{\algC}) \iso \CHomp{G\algB}{G\algC}$ satisfies the following coherence condition
%\begin{diagram}%[height=1.5em]%[LaTeXeqno] \label{eq:coherence:cond}
%F(\CHom{\algB}{\algC}) & \rTo^{\iso} & \CHomp{G\algB}{G\algC} \\
%\dTo^{{\VTwoCell}} &&& \rdTo^{\CHomp{\id}{\CTwoCell}} \\
%F'(\CHom{\algB}{\algC}) & \rTo^{\iso} & \CHomp{G'\algB}{G'\algC} & \rTo^{{\CHomp{\CTwoCell}{\id}}} &\CHomp{G\algB}{G'\algC}
%\end{diagram}
%\end{lem}
%

The composition of 1-cells is defined as $(F',G',\ltensoriso') \circ (F,G,\ltensoriso) = (F'F, G'G, \ltensoriso'\circ G'\ltensoriso)$:
\begin{diagram}
G'G(\ltensor{\SA}{\algB}) & \rTo^{G'\ltensoriso} & G'(\ltensor{F\SA}{G\algB}) & \rTo^{\ltensoriso'} & \ltensor{F'F\SA}{G'G\algB}
\end{diagram}
The composition of 2-cells is simply pointwise $(\VTwoCell', \CTwoCell') \circ (\VTwoCell, \CTwoCell) = (\VTwoCell' \circ \VTwoCell, \CTwoCell'\circ \CTwoCell)$

%\begin{lem}
%$\ENR$ is a 2-category.
%\end{lem}

\subsection{The 2-category $\CATECBV$}
\label{app:CATECBV}

The objects of the 2-category $\CATECBV$ 
are tuples $(\VCat,\CCat,\stateobj)$
where $(\VCat,\CCat)$ is an {\enrmodel} and
$\stateobj$ is a `state' object in $\CCat$.
%together with given interpretations of the base value types
%$\alpha,\beta,\dots$ in $\VCat$ and
%a given interpretation of the distinguished base computation type
%$\EECstate$ in $\CCat$.
A 1-cell $(\VCat,\CCat,\stateobj)\to(\VCat',\CCat', \stateobj')$ 
is a quadruple $(F,G, \ltensoriso, \stateiso)$ such that $(F,G, \ltensoriso) \co (\VCat,\CCat)\to(\VCat',\CCat')$ is a morphism of {\enrmodel}s and $\stateiso$ is an isomorphism
\[\stateiso  \co G\stateobj  \iso \stateobj' 
\]
A 2-cell from $(F,G, \ltensoriso, \stateiso)$ to $(F',G', \ltensoriso', \stateiso')$ is a 2-cell of {\enrmodel}s 
(\S\ref{app:CATENR})
\[(\VTwoCell,\CTwoCell) \co (F,G) \to (F',G')\] such that 
\begin{diagram}[LaTeXeqno] \label{eq:stateobj:coherence:cond}
G\stateobj & \rTo^{\stateiso} & \stateobj'  \\
\dTo^{\CTwoCell} & \ruTo^{\stateiso'} \\
G'\stateobj
\end{diagram}

Composition of 1-cells and 2-cells in $\CATECBV$ is defined as in
$\ENR$. The state object isomorphisms are composed as follows:
\begin{diagram}
G'G\stateobj & \rTo^{G'\stateiso} & G'\stateobj' & \rTo^{\stateiso'} & \stateobj''\text.
\end{diagram}

%\begin{lem}
%$\CATECBV$ is a 2-category.
%\end{lem}

%\begin{proof}
%I actually have checked most details of this on paper.
%\end{proof}

\subsection{The 2-category $\Freyd$}
\label{app:Kleisli:cat}

We introduce the 2-category $\Freyd$ whose objects are {\Klmodel}s
(\S\ref{def:monadmodel}). Morphisms from $(\VCat, \CCat, J)$ to
$(\VCat', \CCat', J')$ are morphisms of enriched models $(F,G,
\ltensoriso) \co (\VCat, \CCat) \to (\VCat', \CCat')$ such that
$J'\circ F = G \circ J$. Note in particular that this means that $F$
and $G$ agree on objects, since $J$ and $J'$ are required to be
identities on objects. A 2-cell $(F,G, \ltensoriso) \to (F',G',
\ltensoriso')$ is a 2-cell of enriched models $(\VTwoCell,\CTwoCell)
\co (F,G) \to (F',G')$ such that $\CTwoCell J = J' \VTwoCell \co GJ
\to J'F'$.

%[[or should this be relaxed to preservation up to iso????? Alternative definition: $\alpha\co GJ \to J'F$ and 2-cells such that $\alpha' \circ \CTwoCell J = J'\alpha' \circ J'\gamma$]]

\subsection{Sums}
\label{app:sums}

The category $\dENR$ is defined to have {\denrmodel}s as objects, 1-cells as in $\ENR$ with the restriction that $F$ and $G$ must both 
preserve coproducts, and 2-cells as in $\ENR$. The definitions of the 2-categories $\CATECBV$ and $\Freyd$ extend to 2-categories 
$\dCATECBV$ and $\dFreyd$ of {\denrmodel}s with state objects and {\dKlmodel}s.

\subsection{Effect theories}
\label{app:def:cats:eff}
We now explain how to extend the 2-categories 
to models that support effect theories, 
expanding on Section~\ref{sec:fullcomplete-effects}.
We first define what it means for a functor to preserve a value theory.

\begin{defi} \label{def:pres:value:theory} Let $\VCat$ and
  $\VCat'$ be distributive categories with given interpretations of
  some fixed value theory (\S\ref{sec:modelvaluetheories}). 
  A functor $F \co \VCat \to \VCat'$
  preserving products and coproducts \emph{preserves the
    interpretation of the value theory} if $F\den{\alpha} =
  \den{\alpha}$ for all type constants $\alpha$ of the theory, and the
  diagram
\begin{diagram}
F(\den{\vec\alpha}) & \rTo^{F\den{f}} & F(\den{\beta}) \\
\dTo^{\tuple{F\pi_1}{F\pi_n}} & \ruTo^{\den{f}} \\
F\den{\alpha_1} \times \dots \times F\den{\alpha_n}
\end{diagram}
commutes for all term constants $f\co \vec \alpha \to \beta$ of the theory. 
\end{defi}

\subsubsection{The 2-category $\CATECBVtheory{E}$}

For any effect theory $E$, 
the objects of the 2-category $\CATECBVtheory{E}$ are 
{\denrmodel}s with state and a
chosen $E$-comodel structure on the state object (\S\ref{sec:eff-ecbvmodels}). 
Morphisms
(1-cells) are morphisms $(F,G, \ltensoriso, \stateiso)$ of $\dCATECBV$
such that
\begin{itemize}
\item $F$ preserves the value theory of $E$ as in Definition~\ref{def:pres:value:theory}.
\item The comodel structure is preserved, i.e., for each effect constant $\arityj e{\vec\beta}{\vec \alpha_1 + \ldots + \vec \alpha_n}$ of $E$, the following diagram commutes
\begin{diagram}[LaTeXeqno] \label{eq:pres:eff:const:ecbv}
G(\ltensor{\den{\vec \beta}}{\stateobj}) & \rTo^{G\den e} & G(\ltensor{(\den{\vec \alpha_1} + \dots + \den{\vec \alpha_n})}{\stateobj}) \\
\dTo && \dTo \\
\ltensor{\den{\vec \beta}}{\stateobj'} & \rTo^{\den e} & \ltensor{(\den{\vec \alpha_1} + \dots + \den{\vec \alpha_n})}{\stateobj'}
\end{diagram}
where 
%we use the notation $\vec F\den{\vec \beta} \eqdef F\den{\beta_1} \times \dots \times F\den{\beta_m}$ (note that this is equal to $\den{\vec\beta}$) and 
the vertical maps are constructed using $\ltensoriso, \stateiso$ and the preservation of products and coproducts. (This makes sense 
because $F(\den\alpha) = \den\alpha$ for all type constants~$\alpha$.)
\end{itemize}
A 2-cell of $\CATECBVtheory{E}$ is a 2-cell $(\VTwoCell, \CTwoCell) \co (F,G,\stateiso) \to (F',G',\stateiso')$ of $\dCATECBV$ such that $\VTwoCell_{\den{\alpha}}$ is the identity for all type constants $\alpha$ of the theory $E$.

%\begin{lem}
%$\CATECBVtheory{E}$ is a 2-category.
%\end{lem}

\subsubsection{The 2-category $\Freydtheory{E}$}

For any effect theory $E$,
the objects of the 2-category $\Freydtheory{E}$ are {\dKlmodel}s 
with given interpretations of $E$. 
Morphisms are morphisms of $\dFreyd$ such that 
\begin{itemize}
\item $F$ preserves the value theory of $E$ as in Definition~\ref{def:pres:value:theory}
\item The interpretation of effect constants is preserved, i.e., for each $\arityj e{\vec\beta}{\vec \alpha_1 + \ldots + \vec \alpha_n}$ of $E$, the following diagram commutes
\begin{diagram}[LaTeXeqno] \label{eq:pres:eff:const:freyd}
G(\den{\vec \beta}) & \rTo^{G\den e} & G(\den{\vec \alpha_1} + \dots + \den{\vec \alpha_n}) \\
\dTo^{\ituple{GJ(\pi_{i})}{i}} && \dTo \\
\den{\vec \beta} & \rTo^{\den e} & (\den{\vec \alpha_1} + \dots + \den{\vec \alpha_n})
\end{diagram}
%\begin{diagram}
%F(\Prod_i \CHom{\stateobj}{\algB}^{\den{\alpha_i}}) & \rTo^{F\den e_{\algB}} & F(\CHom{\stateobj}{\algB}^{\den{\beta}}) \\
%\dTo && \dTo \\
%\Prod_i \CHom{\stateobj'}{G\algB}^{F\den{\alpha_i}} & \rTo^{\den e_{G\algB}} & \CHom{\stateobj'}{G\algB}^{\den{\beta}}
%\end{diagram}
where the left vertical map types because $G\den{\beta_{i}} = \den{\beta_{i}}$, and the right vertical map is constructed using the fact that $G$ preserves coproducts.
\end{itemize}
The 2-cells of $\Freydtheory{E}$ are the 2-cells $(\VTwoCell, \CTwoCell)$ of $\dFreyd$ with $\VTwoCell_{\den{\alpha}}$ the identity for all type constants $\alpha$ of the theory $E$.

%\begin{lem}
%$\Freydtheory{E}$ is a 2-category.
%\end{lem}

\section{Bi-initiality of syntactic models}

We sketch a proof of Theorem~\ref{thm:ecbv:biinitial}:
the syntactic enriched model is bi-initial in the 2-category $\CATECBV$
of enriched models.
Bi-initiality of the syntactic monad model (Theorem~\ref{thm:fgcbv:biinitial}) 
can be proved in a similar manner.

\begin{lem} \label{lem:iso:interp:ecbv}
Suppose $(F,G, \stateiso) \co (\VCat,\CCat,\stateobj)\to(\VCat',\CCat', \stateobj')$  is a morphism in $\CATECBVtheory{E}$. 
Let $\den -$ be the interpretation of ECBV in $(\VCat,\CCat)$,
and let $\den-'$ be the interpretation of ECBV in $(\VCat',\CCat')$
(as in \S\ref{sec:adjmodels}, \S\ref{sec:eff-ecbvmodels}).
The family of isomorphisms given by 
\begin{align*}
\id & \co \den \alpha' \to F\den\alpha \\%  \text{ and } \\
\inv{\stateiso} & \co \den\states' = \stateobj' \to G\stateobj = G\den\states
\end{align*}
extends uniquely to a type-indexed family of isomorphisms 
$\den{\VA}' \iso F(\den{\VA})$ and $\den{\CB}' \iso G(\den{\CB})$ 
such that 
if $\tj{\Gamma}{t}{\VA}$ and $\aj{\Gamma}{\Delta}{u}{\CB}$ then
\begin{diagram}[LaTeXeqno] \label{eq:iso:interp:ecbv1}
\den{\Gamma}' & \rTo^{\den{t}'} & \den{\VA}' & \GAP \text{and} \GAP & \ltensor{\den{\Gamma}'}{\den{\Delta}'} & \rTo^{\den{u}'} & \den{\CB}' \\
 \dTo^{\iso} && \dTo^{\iso} && 
\dTo^{\iso} && \dTo^{\iso} \\
F\den{\Gamma} & \rTo^{F\den{t}} & F\den{\VA} && G(\ltensor{\den{\Gamma}}{\den{\Delta}}) & \rTo^{G\den{u}} & G\den{\CB}
\end{diagram}
%\item If $\aj{\Gamma}{\Delta}{t}{\CB}$ then 
%\begin{diagram}
%\ltensor{\den{\Gamma}}{\den{\Delta}} & \rTo^{\den{t}} & \den{\CB} \\
%\dTo^{\iso} && \dTo^{\iso} \\
%G(\ltensor{\den{\Gamma}}{\den{\Delta}}) & \rTo^{G\den{t}} & G\den{\CB}
%\end{diagram}
%\end{enumerate}
Moreover, if $(\VTwoCell, \CTwoCell)$ is a 2-cell then 
\begin{diagram}[LaTeXeqno] \label{eq:iso:interp:ecbv2}
\den{\VA}' & \rTo^{\iso} & F\den{\VA}  & \GAP \text{and} \GAP & \den{\CB}' & \rTo^{\iso} & G\den{\CB}\\
& \rdTo^{\iso} & \dTo_{\VTwoCell} && & \rdTo^{\iso} & \dTo_{\CTwoCell}\\
&& F'\den{\VA} && && G'\den{\CB}
\end{diagram}
\end{lem}

\begin{proofnotes}
The isomorphisms are defined by induction on $\VA$ and $\CB$. For example, in the case of $\ltensortype{\VA}{\CB}$ we use the composite
\begin{diagram}
 \den{\ltensortype{\VA}{\CB}}'  = \ltensor{\den{\VA}'}{\den{\CB}'} & \rTo^{\iso} & \ltensor{(F\den{\VA})}{(G\den{\CB})} & \rTo^{\inv{\ltensoriso}} & G(\ltensor{\den{\VA}}{\den{\CB}}) =  G\den{\ltensortype{\VA}{\CB}}
\end{diagram}
Commutativity of (\ref{eq:iso:interp:ecbv1}) and
(\ref{eq:iso:interp:ecbv2}) is proved by induction on structure of
terms and types respectively. We omit the lengthy verifications
entirely, but remark that (\ref{eq:1:cell:coherence:cond:1}) is used
to prove the case of linear variable introduction, and diagram
(\ref{eq:1:cell:coherence:cond:2}) is used to prove the case of linear
function application and the elimination rule for
$\ltensortype{\VA}{\CB}$. The case of effect constants is exactly the
requirement (\ref{eq:pres:eff:const:ecbv}).

Uniqueness is proved by induction on types. %The base case of a type constants $\alpha$ follows from (\ref{eq:iso:interp:ecbv2}) by taking $t$ to be $\tj{x \co \alpha}{x}{\alpha}$.
\end{proofnotes}

Bi-initiality of the syntactic model
(Theorem~\ref{thm:ecbv:biinitial}) follows from Lemma~\ref{lem:iso:interp:ecbv} as follows.
Given any other model, the unique morphism is given by interpretation of the syntax. To prove uniqueness, suppose $(F,G, \stateiso)$ is a 1-cell from $\CATECBVtheory{E}$ to  $(\VCat, \CCat, \stateobj)$. Then Lemma~\ref{lem:iso:interp:ecbv} gives the natural isomorphism: since interpretation in the syntactic model of terms with one variable is simply the identity, diagrams  (\ref{eq:1:cell:coherence:cond:1}) prove naturality of the isomorphism. The required commutative diagrams (\ref{eq:coherence:cond}) for 2-cells follow directly from definitions. 

%\subsection{Initiality of the syntactic Kleisli model}
%
%\todo{Prove and write.}

\section{The adjunction between Kleisli models and enriched models}
\label{app:adjunction}

\subsection{The 2-functor $\FreydToECBV \colon \Freyd \to \CATECBV$.} 
We define the 2-functor $\FreydToECBV \colon \Freyd \to \CATECBV$ by the following data:
\begin{align*}
\FreydToECBV(\VCat, \CCat, J) & \,\defeq (\VCat, \CCat, 1) \\
\FreydToECBV(F,G) & \,\defeq (F,G, J'(\bang{})) && \text{for } (F,G) \co (\VCat, \CCat, J) \to (\VCat', \CCat', J') \\
\FreydToECBV(\VTwoCell, \CTwoCell) & \,\defeq (\VTwoCell, \CTwoCell) 
\end{align*}
Note that $J'(\bang{})$ has the right type:
\[
G(1) = G(J(1))  =  J'(F(1)) \xrightarrow{J'(\bang{})} J'(1) = 1
\text.\]

\subsection{The 2-functor $\ECBVToFreyd \colon \CATECBV \to \Freyd$.} 
\label{app:ECBVToFreyd}

Recall that $\ECBVToFreyd$ is defined on objects in Section~\ref{sec:relating:models} as
\begin{align*}
\ECBVToFreyd(\VCat, \CCat, \stateobj) \defeq (\VCat, \KlCat{\VCat}{\CCat}{\stateobj}, J_{\stateobj})
\end{align*}
where the category $\KlCat{\VCat}{\CCat}{\stateobj}$ has the same objects as $\VCat$ and homsets
\[
\Homset{\KlCat{\VCat}{\CCat}{\stateobj}}{\SA}{\SB} \defeq \Homset{\CCat}{\ltensor{\SA}{\stateobj}}{\ltensor{\SB}{\stateobj}}
\]

On morphisms we define $\ECBVToFreyd(F,G,\ltensoriso) \defeq (F,\Kl{F,G}, \ltensorisoKl)$ where 
\[\Kl{F,G} \co \KlCat{\VCat}{\CCat}{\stateobj} \to \KlCat{\VCat'}{\CCat'}{\stateobj'}\]
is the functor that maps an object $\SA$ to $F\SA$ and a morphism $f
\co \ltensor{\SA}{\stateobj} \to \ltensor{\SB}{\stateobj}$ to the
following composite:
\begin{diagram}[LaTeXeqno] \label{eq:defn:Kl:functor}
\ltensor{F\SA}{\stateobj'} & \rTo^{\ltensor{\SA}{\stateiso}} & \ltensor{F\SA}{G\stateobj} & \rTo^{\inv{\ltensoriso}} & G(\ltensor{\SA}{\stateobj})  \\
\dTo^{\Kl{F,G}(f)} && && \dTo_{G(f)} \\
\ltensor{F\SB}{\stateobj'} & \lTo^{\ltensor{\SB}{\inv{\stateiso}}} & \ltensor{F\SB}{G\stateobj} & \lTo^{\ltensoriso} & G(\ltensor{\SB}{\stateobj})
\end{diagram}
The natural transformation 
\[\ltensorisoKl \co \Kl{F,G}(\Klltensor{\SA}{\SB}) \to \Klltensor{F(\SA)}{\Kl{F,G}(\SB)} \]
has components given by 
\[
\ltensor{\pair{F\pi_1}{F\pi_2}}{\stateobj'} \co \ltensor{F(\SA \times \SB)}{\stateobj'} \to \ltensor{(F(\SA) \times F(\SB))}{\stateobj'}
\]

On 2-cells we define $\ECBVToFreyd(\VTwoCell, \CTwoCell) = (\VTwoCell, \Kl{\VTwoCell,\CTwoCell})$ where $\Kl{\VTwoCell, \CTwoCell}$ is the natural transformation whose components are $\Kl{\VTwoCell,\CTwoCell}_{\SA} = \ltensor{\VTwoCell_{\SA}}{\stateobj'} \co \ltensor{FA}{\stateobj'} \to \ltensor{F'A}{\stateobj'}$. Note that this defines a morphism in $\KlCat{\VCat}{\CCat}{\stateobj'}$ from $\Kl{F,G}(\SA) = F\SA$ to $\Kl{(F',G')}(\SA) = F'\SA$ as required. 

\subsection{The unit}

If $(\VCat, \CCat, J)$ is a {\dKlmodel}, then 
\[\ECBVToFreyd(\FreydToECBV(\VCat, \CCat, J)) = (\VCat, \KlCat{\VCat}{\CCat}{1}, J_{1}) 
\]
where $\KlCat{\VCat}{\CCat}{1}$ has the same objects as $\VCat$ and $\Homset{\KlCat{\VCat}{\CCat}{1}}{\SA}{\SB} = \Homset{\CCat}{\SA \times 1}{\SB \times 1}$ and $J_{1}(\SA) = \SA$, $J_{1}(f) = J(f \times 1) = \ltensor{f}{1}$.

We define the unit $\eta_{(\VCat, \CCat, J)}$ as 
\[
\eta_{(\VCat, \CCat, J)} \defeq (\id, H, \id) \co (\VCat, \CCat, J) \to (\VCat, \KlCat{\VCat}{\CCat}{1}, J_{1}) 
\]
where $H(\SA) = \SA$, and $H(f \co \SA \to \SB)$ is the composition
\begin{diagram}
\SA \times 1 & \rTo^{J(\pi_{1})} & \SA & \rTo^{f} & \SB & \rTo^{J\pair{\id}{\bang{}}} & \SB\times 1
\end{diagram}
% Clearly $H$ defines a functor. %(in fact an isomorphism of categories)
The third component of $\eta$ should be a natural transformation $H(\ltensor{\SA}{\SB}) \to \Klltensor{\SA}{H\SB}$. Since both sides are equal to $\SA \times \SB$ it makes sense to take the identity transformation.

\begin{lem} \label{lem:eta:iso}
Each $\eta_{(\VCat, \CCat, J)}$ is an isomorphism in $\Freyd$.
\end{lem}
%
%\begin{proof}
%Follows from Lemma~\ref{lemma:isos:ENR}. 
%\end{proof}

The next lemma states that the unit is a 2-natural transformation.

\begin{lem} \label{lem:eta:biadj} For each pair of {\dKlmodel}s
  $(\VCat, \CCat, J)$ and $(\VCat', \CCat', J')$, the 
  following diagram 
  of functors commutes.
\hide{\begin{diagram}
\Freyd((\VCat, \CCat, J), (\VCat', \CCat', J')) & \rTo^{\ECBVToFreyd\circ \FreydToECBV} & \Freyd((\VCat, \KlCat{\VCat}{\CCat}1, J_1), (\VCat', \KlCat{\VCat}{\CCat}1', J_1')) \\
\dTo^{(\eta_{(\VCat', \CCat', J')})_*} & \ldTo_{\eta_{(\VCat, \CCat, J)}^*} \\
\Freyd((\VCat, \CCat, J), (\VCat', \KlCat{\VCat}{\CCat}1', J_1'))
\end{diagram}}
\[\xymatrix@C+1cm{
\Freyd((\VCat, \CCat, J), (\VCat', \CCat', J')) 
\ar[r]^-{\ECBVToFreyd\circ \FreydToECBV} 
\ar@/_1pc/[dr]_-{(\eta_{(\VCat', \CCat', J')})_*} 
& \Freyd((\VCat, \KlCat{\VCat}{\CCat}1, J_1), (\VCat', {\KlCat{\VCat}{\CCat}1}', J_1')) 
\ar[d]^{(\eta_{(\VCat, \CCat, J)})^*}
\\
&\Freyd((\VCat, \CCat, J), (\VCat', {\KlCat{\VCat}{\CCat}1}', {J_1}'))
}\]
where $f^*$ and $g_*$ are the functors given by pre- and postcomposition by $f$ and $g$ respectively. 
\end{lem}

\subsection{The counit}

\newcommand{\I}[1]{I_{#1}}
\newcommand{\epsnat}[1]{\epsilon_{#1}}

Let $(\VCat, \CCat, \stateobj)$ be a model of enriched call-by-value. Then
\[
\FreydToECBV(\ECBVToFreyd(\VCat, \CCat, \stateobj)) = (\VCat, \KlCat{\VCat}{\CCat}{\stateobj}, 1)
\]
Recall that $\KlCat{\VCat}{\CCat}{\stateobj}$ has the same objects as $\VCat$ and set of morphisms defined as 
\[
\Homset{\KlCat{\VCat}{\CCat}{\stateobj}}{\SA}{\SB} = \Homset{\CCat}{\ltensor{\SA}{\stateobj}}{\ltensor{\SB}{\stateobj}}
\]
Define the counit $\epsilon$ as 
\[
\epsilon_{(\VCat, \CCat, \stateobj)} \defeq (\id, \I{\stateobj}, \delta) \co (\VCat, \KlCat{\VCat}{\CCat}{\stateobj}, 1) \to (\VCat, \CCat, \stateobj)
\]
where $\I{\stateobj}$ is the functor defined by $\I{\stateobj}(\SA) = \ltensor{\SA}{\stateobj}$, $\I{\stateobj}(f) = f$, and $\delta$ is the isomorphism
\[\delta \co \I{\stateobj}(1) =  \ltensor 1{\stateobj} \to \stateobj
\]

The counit is not strictly natural, but it does satisfy the following naturality condition.

\begin{lem}\label{lem:eps:biadj}
Let $(\VCat, \CCat, \stateobj), (\VCat', \CCat', \stateobj')$ be two given objects of $\CATECBV$. 
The following diagram commutes up to natural isomorphism.
%The family $(\epsnat{(F,G)})_{(F,G)}$ induces a natural isomorphism between the two composite functors in the below diagram
\hide{\begin{diagram}
\CATECBV((\VCat, \CCat, \stateobj), (\VCat', \CCat', \stateobj')) & \rTo^{\FreydToECBV \circ \ECBVToFreyd} & \CATECBV((\VCat, \KlCat{\VCat}{\CCat}{\stateobj}, 1), (\VCat', \KlCat{\VCat'}{\CCat'}{\stateobj'}, 1)) \\
\dTo^{\epsilon_{(\VCat, \CCat, \stateobj)}^*} & \ldTo_{(\epsilon_{(\VCat', \CCat', \stateobj')})_*} \\
\CATECBV((\VCat, \KlCat{\VCat}{\CCat}{\stateobj}, 1), (\VCat', \CCat', \stateobj'))
\end{diagram}}
\[\xymatrix@C+1cm{
\CATECBV((\VCat, \CCat, \stateobj), (\VCat', \CCat', \stateobj')) 
\ar[r]^{\FreydToECBV \circ \ECBVToFreyd} 
\ar@/_1pc/[dr]_-{(\epsilon_{(\VCat, \CCat, \stateobj)})^*} 
&\CATECBV((\VCat, \KlCat{\VCat}{\CCat}{\stateobj}, 1), (\VCat', \KlCat{\VCat'}{\CCat'}{\stateobj'}, 1)) 
\ar@{}[dl]|[.01]\cong
\ar[d]^-{(\epsilon_{(\VCat', \CCat', \stateobj')})_*} \\
&\CATECBV((\VCat, \KlCat{\VCat}{\CCat}{\stateobj}, 1), (\VCat', \CCat', \stateobj'))
}\]
%Where the top map is the composite $\FreydToECBV \circ \ECBVToFreyd$.
\end{lem}

\subsection{Triangle equalities}

One of the triangle equalities only holds up to isomorphism. %, and so we will prove that 
%\begin{eqnarray*}
%\ECBVToFreyd(\epsilon) \circ \eta_{\ECBVToFreyd}  = \id \\
%\epsilon_{\FreydToECBV} \circ \FreydToECBV(\eta)  \iso \id
%\end{eqnarray*}

\begin{prop}\label{prop:triangle:eq:1}
Let $(\VCat, \CCat, \stateobj)$ be a {\denrmodel} with state. The composite 1-cell 
\[\ECBVToFreyd(\epsilon_{(\VCat, \CCat, \stateobj)}) \circ \eta_{\ECBVToFreyd(\VCat, \CCat, \stateobj)}\] 
is the identity on $\ECBVToFreyd(\VCat, \CCat, \stateobj)$.
\end{prop}

\begin{prop} \label{prop:triangle:eq:2}
Let $(\VCat, \CCat, J)$ be a {\dKlmodel}. The composite 1-cell
\[\epsilon_{\FreydToECBV(\VCat, \CCat, J)} \circ \FreydToECBV(\eta_{(\VCat, \CCat, J)})
\]
is naturally isomorphic to the identity on $\FreydToECBV(\VCat, \CCat, J)$.
\end{prop}

\subsection{Proof of Theorem~\ref{thm:adj}}
We now prove Theorem~\ref{thm:adj}: $\eta$ and $\epsilon$ induce
an equivalence
of categories
  \[\CATECBV(\FreydToECBV(\VCat, \CCat, J),(\VCat', \CCat',
  \stateobj)) \ \simeq\ \Freyd((\VCat, \CCat, J),\ECBVToFreyd(\VCat',
  \CCat', \stateobj))\text.
\]

In the following we use $X$ to denote an object of $\CATECBV$ (rather than the 
much longer $(\VCat, \CCat, \stateobj)$) and $Y$ to denote an object of $\Freyd$.
The required equivalence of categories consists of the composite functors
\begin{diagram}
\CATECBV(\FreydToECBV(X),Y) & \rTo^{\ECBVToFreyd} & \Freyd(\ECBVToFreyd(\FreydToECBV(X)),\ECBVToFreyd(Y)) & \rTo^{(\eta_X)^*} & \Freyd(X,\ECBVToFreyd(Y)) \\
\Freyd(X,\ECBVToFreyd(Y)) & \rTo^{\FreydToECBV} & \CATECBV(\FreydToECBV(X),\FreydToECBV(\ECBVToFreyd(Y))) & \rTo^{(\epsilon_Y)_*} & \CATECBV(\FreydToECBV(X),Y) \\
%\CATECBV(X,\FreydToECBV(Y)) & \rTo^{\ECBVToFreyd} & \Freyd(\ECBVToFreyd(X),\ECBVToFreyd(\FreydToECBV(Y))) & \rTo^{(\epsilon_Y)_*} & \Freyd(\ECBVToFreyd(X),Y) \\
%\Freyd(\ECBVToFreyd(X),Y) & \rTo^{\FreydToECBV}  & \CATECBV(\FreydToECBV(\ECBVToFreyd(X)),\FreydToECBV(Y)) & \rTo^{\eta_X^*} & \CATECBV(X,\FreydToECBV(Y)) \\
\end{diagram}
We need to prove that the two composites of these are naturally isomorphic to identities. For the first of the composites, consider the following sequences of identities and natural isomorphisms
\begin{align*}
(\epsilon_Y)_* \circ \FreydToECBV \circ \eta_X^* \circ \ECBVToFreyd  
& = (\epsilon_Y)_* \circ (\FreydToECBV(\eta_X))^* \circ \FreydToECBV \circ \ECBVToFreyd \\
& = (\FreydToECBV(\eta_X))^*  \circ (\epsilon_Y)_*\circ \FreydToECBV \circ \ECBVToFreyd \\
& \iso (\FreydToECBV(\eta_X))^*  \circ (\epsilon_{\FreydToECBV(Y)})^*  & \text{Lemma~\ref{lem:eps:biadj}} \\ 
& = (\epsilon_{\FreydToECBV(Y)} \circ \FreydToECBV(\eta_X))^*  \\
& \iso \id^* & \text{Proposition~\ref{prop:triangle:eq:2}} \\
& = \id
\end{align*}
%\begin{align*}
%\eta_X^* \circ \ECBVToFreyd \circ (\epsilon_Y)_* \circ \FreydToECBV & = \eta_X^*  \circ (\ECBVToFreyd(\epsilon_Y))_*\circ \ECBVToFreyd \circ \FreydToECBV \\
%& =   (\ECBVToFreyd(\epsilon_Y))_* \circ \eta_X^*  \circ \ECBVToFreyd \circ \FreydToECBV \\
%& = (\ECBVToFreyd(\epsilon_Y))_* \circ (\eta_{\ECBVToFreyd(Y)})_*  && \text{Lemma~\ref{lem:eta:biadj}} \\ 
%& = (\ECBVToFreyd(\epsilon_Y) \circ \eta_{\ECBVToFreyd(Y)})_*  \\
%& = \id_* && \text{Lemma~\ref{prop:triangle:eq:1}} \\
%& = \id
%\end{align*}
The other composite is isomorphic to the identity for a similar reason. \qed

%%% Local Variables: 
%%% mode: latex
%%% TeX-master: "mogelberg-staton"
%%% End: 

% \input{ref-app}

\end{document}